\documentclass[12pt]{article}

\usepackage[margin=1.5cm]{geometry}
\usepackage[T1]{fontenc}
\usepackage[english]{babel}
\usepackage{amsmath,amssymb,amsfonts,amsthm,amsopn}
\usepackage{bm}
\usepackage{booktabs}
\usepackage{graphicx}
\usepackage{subcaption}
\usepackage{enumitem}
\usepackage{xspace}
\usepackage{cite}
\usepackage[colorlinks=true,linkcolor=blue,citecolor=blue,urlcolor=blue]{hyperref}
\usepackage[nameinlink]{cleveref}

\setlist[enumerate]{leftmargin=.5in}
\setlist[itemize]{leftmargin=.5in}

\numberwithin{equation}{section}
\numberwithin{figure}{section}
\numberwithin{table}{section}

\newtheorem{theorem}{Theorem}[section]
\newtheorem{lemma}[theorem]{Lemma}
\newtheorem{proposition}[theorem]{Proposition}

\theoremstyle{remark}
\newtheorem{remark}[theorem]{Remark}

\crefname{hypothesis}{Hypothesis}{Hypotheses}
\crefname{fact}{Fact}{Facts}

\newenvironment{reading}{\par\medskip\noindent\textbf{Further reading.}\par\smallskip}{\par\medskip}

\PassOptionsToPackage{force}{filehook}
\usepackage{tikz}
\usetikzlibrary{shapes,arrows,positioning,shapes.symbols,shapes.callouts,patterns,calc}
\tikzstyle{cloud} = [draw, ellipse, fill=red!20, node distance=0.87cm, minimum height=2em]
\tikzstyle{line} = [draw, -latex']

\usepackage{pgfplots}
\pgfplotsset{compat=1.17}

\def\R{\mathcal{R}}
\newcommand{\F}{\mathcal{F}}
\def\L{\mathcal{L}}
\newcommand{\W}{\mathcal{W}}
\newcommand{\IR}{\mathbb{R}}

\def\by{\bm{y}}
\def\bE{\bm{E}}

\def\bX{\bm{X}}

\def\b0{\bm{0}}

\title{Dose-limited interventions in an epidemiological model}

\author{A. S. Abdramane\thanks{Department of Mathematics, University of N'Djamena, N'Djamena, Chad}
    \and H. Djimramadji\footnotemark[1]
    \and M. S. Daoussa Haggar\footnotemark[1]
    \and P. M. Tchepmo Djomegni\thanks{School of Mathematics and Statistical Sciences, North-West University, South Africa}
    \and Julien Arino\thanks{Department of Mathematics, University of Manitoba, Winnipeg, Manitoba, Canada (\texttt{julien.arino@umanitoba.ca})}
}

\date{\today}

\begin{document}

\maketitle

%%%%%%%%%%%%%%%%%%%%%
\begin{abstract}
    We consider a susceptible-latent-infectious-asymptomatic-recovered-susceptible (SLIARS) mathematical epidemiology model with potentially dose-limited vaccination and treatment.
    Mathematically, we show that most scenarios actually reduce to classic well-known scenarios: having an unreplenished number of doses is akin to having none, while being able to restore stocks is (often) equivalent to the classic situation with vaccination and treatment.
    We also perform a computational analysis, illustrating some of the transient and stochastic dynamics that diverge from deterministic long-term behavior, as well as the impact of budgetary constraints.
\end{abstract}
\vspace{1ex}\noindent\textbf{Keywords:} Mathematical epidemiology; public health interventions; vaccination; treatment; stockpile management\\

\vspace{1ex}

%%%%%%%%%%%%%%%%%%%%%
%%%%%%%%%%%%%%%%%%%%%
%%%%%%%%%%%%%%%%%%%%%
%%%%%%%%%%%%%%%%%%%%%
\section{Introduction}
\label{sec:introduction}
Most mathematical models involving vaccination or treatment make the hypothesis that doses allowing such control mechanisms are available in sufficient quantities.
The effect of varying vaccination or treatment rates is then investigated.
However, in low-income countries or in emergency situations such as those that arise in refugee camps, vaccination or treatment options are not always readily available \cite{cameron2011differences,dimitrova2023essential,sabahelzain2025implications,wang2025global}.
The recent SARS-CoV-2 pandemic highlighted the inequity in vaccine distribution between rich and poor countries, exposing structural limitations in global healthcare delivery \cite{acharya2021access,bayati2022inequality,privor2023vaccine,tatar2022covid}.
Furthermore, even when doses reach their intended destinations, there are often subsequent logistical issues spanning supply chains, such as cold-chain failures and expiration dates, that further reduce the actual dose availability and efficacy on the ground \cite{ashok2017improving}.

Mathematically, constraints on life-saving epidemiological interventions have primarily been investigated through three main lenses: optimal control theory, saturating healthcare capacity functions, and system dynamics stock-flow models.
In optimal control frameworks, the objective is typically to determine time-varying intervention schedules that minimize disease burden, mortality, or financial cost over a fixed horizon subject to budget or isoperimetric resource constraints \cite{hansen2011optimal,lee2010optimal,lenhart2007optimal,sharomi2017optimal}.
Such studies often apply Pontryagin's Maximum Principle to characterize optimal vaccination or treatment deployment (frequently yielding ``bang-bang'' or singular control profiles) or use operations research methods to prioritize limited supplies across demographic or geographical groups \cite{balcik2022mathematical,bozzani2021building,duijzer2018dose,kasaie2013simulation,rao2021optimal,rao2024vaccination,wang2024optimal}.
A second line of work incorporates resource limits directly into autonomous epidemic models via piecewise-linear or saturating treatment functions that cap the maximum number of individuals treated per unit time, reflecting instantaneous bottlenecks such as hospital bed capacity or staffing limits and often inducing backward bifurcations and bistability \cite{wang2006backward,zhou2012dynamics}.
More recently, system dynamics modeling has highlighted that classical compartmental models relying on fixed intervention rates implicitly assume that medical supplies can adjust instantaneously and without bound, proposing instead to explicitly couple disease transmission with finite resource stocks and supply-chain constraints \cite{andrade2024preparing}.

While optimal control approaches are invaluable for prescribing how a central planner \emph{should} dynamically ration a limited budget when full epidemic trajectories can be anticipated, and saturating treatment models capture instantaneous delivery bottlenecks, neither framework directly characterizes the baseline autonomous dynamics that unfold when clinics simply administer available doses at capacity until stockpiles are exhausted.
In many resource-limited or emergency settings, public health authorities cannot execute finely tuned time-varying control schedules; instead, interventions operate until local stocks run out or until discrete replenishment shipments arrive.
In order to investigate the direct dynamical and stochastic consequences of dose depletion, spoilage, and replenishment, we formulate a compartmental model in which both vaccine and treatment stockpiles ($D_V$ and $D_T$) are explicitly tracked as dynamical state variables coupled to the epidemiological flow.
The model builds on \cite{ArinoMilliken2022a}, where only vaccination was considered and where the number of available doses was assumed to be unlimited.
An accessible overview of the modeling framework and mathematical methods for non-specialists is provided in Appendix~\ref{app:non-math-summary}.
%%%%%%%%%%%%%%%%%%%%%
%%%%%%%%%%%%%%%%%%%%%
%%%%%%%%%%%%%%%%%%%%%
%%%%%%%%%%%%%%%%%%%%%
\section{The model}
\label{sec:model}
In order to facilitate the work, we closely follow the model of \cite{ArinoMilliken2022a}.
The population is divided into six epidemiological compartments: \emph{susceptible} individuals $S$, \emph{latently infected} individuals $L$ who are incubating with the disease, \emph{symptomatically infectious} individuals $I$, \emph{asymptomatically infectious} individuals $A$, recovered individuals $R$, and \emph{vaccinated} individuals $V$.
Two additional compartments $D_V$ and $D_T$ represent the number of vaccine and treatment \emph{doses} available, respectively.
Transitions between compartments are illustrated in Figure~\ref{fig:flow-diagram}.
Parameters are summarized in Table~\ref{tab:parameter_values}.

\begin{figure}[htbp]
    \centering
    % Define skips so things can be adjusted quickly
    \def\hhskip{*3}
    \def\vvskip{*3}
    \begin{tikzpicture}[auto,
        scale=0.8, transform shape,
        % Style for parameter boxes
        param/.style={draw, fill=white, rectangle, inner sep=2pt, align=center},
        cloud/.style={minimum width={width("N-1")+2pt}, draw, ellipse, fill=gray!20},
        line/.style={draw, -latex'}]
    %%% Vaccine doses
    \node [cloud, fill=blue!50] at (-1.5\hhskip,0.5\vvskip) (D_V) {$D_V$};
    %%% Treatment doses
    \node [cloud, fill=blue!50] at (4\hhskip,0.5\vvskip) (D_T) {$D_T$};
    %%% States
    \node [cloud, fill=green!50] at (0\hhskip,1\vvskip) (S) {$S$};
    \node [cloud, fill=yellow!50] at (1\hhskip,0\vvskip) (L) {$L$};
    \node [cloud, fill=red!50] at (2\hhskip,1\vvskip) (I) {$I$};
    \node [cloud, fill=red!50] at (2\hhskip,-1\vvskip) (A) {$A$};
    \node [cloud, fill=cyan!35] at (3\hhskip,0\vvskip) (R) {$R$};
    \node [cloud, fill=blue!50] at (0\hhskip,-1\vvskip) (V) {$V$};
    % Inflow and outflow vaccines
    \node [cloud, above=0.5\hhskip of D_V, draw=none, fill=none] (inflowD) {};
    \node [cloud, below=1.75\hhskip of D_V, draw=none, fill=none] (outflowD) {};
    % Inflow and outflow treatment
    \node [cloud, above=0.5\hhskip of D_T, draw=none, fill=none] (inflowT) {};
    \node [cloud, below=1.75\hhskip of D_T, draw=none, fill=none] (outflowT) {};
    % Births and deaths
    \node [cloud, left=0.8\hhskip of S, draw=none, fill=none] (birthS) {};
    \node [cloud, left=0.8\hhskip of V, draw=none, fill=none] (birthV) {};
    \node [cloud, right=0.3\hhskip of S, draw=none, fill=none] (deathS) {};
    \node [cloud, right=0.3\hhskip of L, draw=none, fill=none] (deathL) {};
    \node [cloud, above=0.3\vvskip of I, draw=none, fill=none] (deathI) {};
    \node [cloud, below=0.3\vvskip of A, draw=none, fill=none] (deathA) {};
    \node [cloud, below=0.3\vvskip of R, draw=none, fill=none] (deathR) {};
    \node [cloud, below=0.3\vvskip of V, draw=none, fill=none] (deathV) {};
    %% Flows
    %%% Inflow/outflow vaccines
    \path [line, very thick] (inflowD) to node [midway, above, sloped, allow upside down] {$a_V(D_V)$} (D_V);
    \path [line, very thick] (D_V) to node [midway, above, sloped, allow upside down] {$p(D_V)b+v(D_V)S+\kappa_VD_V$} (outflowD);
    %%% Inflow/outflow treatment
    \path [line, very thick] (inflowT) to node [midway, below, sloped, allow upside down] {$a_T(D_T)$} (D_T);
    \path [line, very thick] (D_T) to node [midway, below, sloped, allow upside down] {$g(D_T)I+\kappa_TD_T$} (outflowT);
    %%% Births/deaths
    \path [line, very thick] (birthS) -- (S) node [midway, above, sloped] {$(1-p(D_V))b$};
    \path [line, very thick] (S) to node [midway, above] {$dS$} (deathS);
    \path [line, very thick] (L) to node [midway, above] {$dL$} (deathL);
    \path [line, very thick] (I) to node [midway, right] {$(d+\delta)I$} (deathI);
    \path [line, very thick] (A) to node [midway, right] {$dA$} (deathA);
    \path [line, very thick] (R) to node [midway, right] {$dR$} (deathR);
    \path [line, very thick] (V) to node [midway, right] {$dV$} (deathV);
    \path [line, very thick] (birthV)--(V) node [very thin, midway, above=2pt, sloped, param, name=birthV_arrow] {$p(D_V)b$} (V);
    % Regular flows
    \path [line, very thick] (S) to node [midway, above, sloped] {$\beta S(I+\eta A)$} (L);
    \path [line, very thick] (L) to node [midway, above, sloped] {$\pi\varepsilon L$} (I);
    \path [line, very thick] (L) to node [midway, above, sloped] {$(1-\pi)\varepsilon L$} (A);
    \path [line, very thick] (I) to [bend right=15] node [midway, below, sloped] {$\gamma I$} (R);
    \path [line, very thick] (A) to node [midway, above, sloped] {$\gamma A$} (R);
    \draw [>=latex,->, very thick, rounded corners] (R) -- (3\hhskip,1.5\vvskip) -- (0\hhskip,1.5\vvskip) node[sloped, midway, below] {$\omega_r R$} -- (S);
    % Vaccination flows
    \path [line, very thick] (S) to [bend right=15]
        node [very thin, midway, below=2pt, sloped, allow upside down, param, name=vaccS_label] {$v(D_V)S$} (V);
    \path [line, very thick, bend right=15] (V) to node [midway, above, sloped] {$\omega_v V$} (S);
    \path [line, very thick] (V) to node [midway, below, sloped] {$(1-\sigma)\beta V(I+\eta A)$} (L);
    % Vaccine flows individuals
    \draw [dotted, ->] (D_V) to (birthV_arrow);
    \draw [dotted, ->] (D_V) to (vaccS_label);
    % Treatment flows
    \path [line, very thick] (I) to [bend left=15]
        node [very thin, midway, above=2pt, sloped, allow upside down, param, name=treatI_label] {$g(D_T)I$} (R);
    \draw [dotted, ->] (D_T) to (treatI_label);
    \end{tikzpicture}
    \caption{Flow diagram of the model.
    Dotted lines show the ``inhibiting'' role of vaccine and treatment doses.}
    \label{fig:flow-diagram}
\end{figure}

Recruitment in the population happens at the constant rate $b$.
We assume that there is no recruitment into the $L$, $I$, $A$, or $R$ compartments.
All compartments are subject to \emph{per capita} exit rates $d$.

Vaccination can happen at two different stages.
A fraction $p(D_V)$ of newly recruited individuals are vaccinated upon joining the population, with the remaining $1-p(D_V)$ entering the susceptible compartment.
An ongoing vaccination program also targets susceptible individuals at the \emph{per capita} rate $v(D_V)$.
Functions $p:\IR_+\to[0,1]$ and $v:\IR_+\to\IR_+$ are $C^1$ and nondecreasing.
Furthermore, we assume that $p(0)=v(0)=0$, so that no vaccination occurs when vaccine stocks are empty, and that
\[
\lim_{D_V(t)\to\infty} p(D_V(t)) = p_{\max}\in[0,1] \text{ and } \lim_{D_V(t)\to\infty} v(D_V(t)) = v_{\max}<\infty.
\]
Regarding replenishment of vaccine doses, we use the $C^1$ function $a_V:\IR_+\to\IR_+$ to denote the rate at which doses are \emph{acquired}.
A natural assumption for $a_V(D_V)$ is that it is a nonincreasing function of the number of vaccine doses available: the maximum rate of dose acquisition occurs when no doses are available, and it does not increase thereafter.
We further assume that
\[
\lim_{D_V(t)\to\infty}a_V(D_V(t))=0,
\]
i.e., no purchases are made if enough vaccine doses are available.
We repeatedly use the assumption that the number of doses is unlimited.
This is explained in more detail in Section~\ref{subsec:model-unlimited-vaccine}; for now, it suffices to say that in this idealized limit, stockpiles are held at constant levels $\bar{D}_V$ and $\bar{D}_T$, and we denote $p^\star=p(\bar{D}_V)$ and $v^\star=v(\bar{D}_V)$ the corresponding proportion of newly recruited individuals vaccinated and \emph{per capita} vaccination rate of susceptible individuals, respectively, together with the constant treatment rate $g^\star=g(\bar{D}_T)$.

Whether wholly susceptible or vaccinated, individuals become infected when they have an infecting contact with an infectious individual in compartment $I$ or $A$.
The rate at which contacts transform into infections is $\beta$.
We use mass action incidence and as a consequence, $\beta$ has units per individual per unit time.
The rate $\beta$ is further lowered for vaccinated individuals by the vaccine efficacy $\sigma\in[0,1]$, which acts on $\beta$ as $1-\sigma$.

Upon infection, individuals first go through an (exponentially distributed) incubation period of mean duration $1/\varepsilon$.
At the end of the incubation period, they enter the infectious compartments $I$ and $A$.
Asymptomatically infectious individuals have their infectiousness modified by the parameter $\eta$, typically taken to be in $[0,1]$.
Individuals with symptomatic infections are subject to disease-induced death at the \emph{per capita} rate $\delta$ additionally to natural death.
Infectious individuals recover at the \emph{per capita} rate $\gamma$.
Symptomatically infectious individuals can also be treated at the \emph{per capita} rate $g(D_T)$, dependent on the number of treatment doses available.
Upon recovery or successful treatment, individuals move on to the recovered compartment, in which they are temporarily immune to reinfection.
They lose this immunity on average after $1/\omega_r$ time units.
The dynamics of treatment doses mimics that of vaccine doses, with $a_T(D_T)$ the rate of acquisition of new doses.
The functions $a_T:\IR_+\to\IR_+$ and $g:\IR_+\to\IR_+$ have the same regularity and monotonicity properties as $a_V$ and $v$, with $g(0)=0$ and $\lim_{D_T(t)\to\infty}g(D_T(t))=g_{\max}<\infty$.

We assume that vaccine and treatment doses decay at the \emph{per capita} rates $\kappa_V>0$ and $\kappa_T>0$, respectively.
This models issues such as finite vaccine or treatment shelf life or cold-chain deficiencies.

\begin{table}[htpb]
    \centering
    \begin{tabular}{clc}
        \toprule
        \textbf{Parameter} & \textbf{Description} & \textbf{Baseline value} \\
        \midrule
        $N$ & Total disease-free population & $100{,}000$ \\
        $1/d$ & Average lifetime & 70 years \\
        $b$ & Birth rate & Computed ($dN$) \\
        $\beta$ & Transmission coefficient & Computed \\
        $\eta$ & Asymptomatic relative infectiousness (prop.) & $0.5$ \\
        $1/\varepsilon$ & Average latent period & 10 days \\
        $\pi$ & Proportion of symptomatic cases (prop.) & $0.7$ \\
        $\delta$ & Disease-induced mortality rate & $5\times 10^{-4}$ \\
        $1/\gamma$ & Average infectious period & 10 days \\
        $1/\omega_v$ & Average vaccine-induced immune period & 200 days \\
        $1/\omega_r$ & Average disease-induced immune period & 5 years \\
        $\sigma$ & Vaccine efficacy (prop.) & $0.7$ \\
        $1/\kappa_V$ & Average shelf life for vaccine doses & 1 year \\
        $1/\kappa_T$ & Average shelf life for treatment doses & 3 years \\
        $p_{\max}$ & Maximal proportion vaccinated at birth (prop.) & $0.8$ \\
        $v_{\max}$ & Maximal routine vaccination rate & $0.02$ \\
        $g_{\max}$ & Maximal treatment rate & $0.2$ \\
        $D_p^s, D_v^s, D_g^s$ & Hill half-saturation constants (doses) & $0.5$ \\
        $k_p, k_v, k_g$ & Hill steepness coefficients (dimensionless) & $20$ \\
        \bottomrule
    \end{tabular}
    \caption{Model parameters and baseline values used for numerical simulations.
    Time units are days (durations given in years are converted to days).
    Rate parameters have units of \emph{per day}, except for dimensionless proportions marked (prop.), the birth rate $b$ in \emph{individuals per day}, the half-saturation constants in \emph{doses}, and the dimensionless Hill exponents.
    The dose-dependent functions $p(D_V)$, $v(D_V)$, and $g(D_T)$ follow the Hill form \eqref{eq:p-function} with maximal values $p_{\max}$, $v_{\max}$, and $g_{\max}$; continuous replenishment is absent ($a_V \equiv a_T \equiv 0$) in the unreplenished and impulsive scenarios of Section~\ref{sec:computational}.}
    \label{tab:parameter_values}
\end{table}

The dynamics is then described by the following system of differential equations,
\begin{subequations}
   \label{sys:general-form}
   \begin{align}
        S' &= (1 - p(D_V)) b + \omega_rR+\omega_vV - \beta S(I+\eta A) - (v(D_V) + d) S,
        \label{sys:general-form-dS} \\
        L' &= \beta (S+(1-\sigma) V)(I+\eta A)
        - (\varepsilon + d) L,
        \label{sys:general-form-dL} \\
        I' &= \pi\varepsilon L
        -(\gamma+g(D_T)+\delta +d) I,
        \label{sys:general-form-dI} \\
        A' &= (1-\pi)\varepsilon L - (\gamma+d) A,
        \label{sys:general-form-dA} \\
        R' &= (\gamma+g(D_T))I+\gamma A-(\omega_r+d)R,
        \label{sys:general-form-dR} \\
        V' &= p(D_V)b + v(D_V)S
        -(1-\sigma)\beta V(I+\eta A) -(\omega_v+d)V,
        \label{sys:general-form-dV} \\
        D_V' &= a_V(D_V) - p(D_V) b - v(D_V) S -\kappa_VD_V,
        \label{sys:general-form-dD} \\
        D_T' &= a_T(D_T) - g(D_T)I -\kappa_TD_T.
        \label{sys:general-form-dT}
    \end{align}
\end{subequations}
Throughout the mathematical analysis, we assume that the parameters satisfy $b, d, \beta, \varepsilon, \gamma > 0$, $\delta, \omega_r, \omega_v \geq 0$, $\pi \in (0, 1]$, $\eta \geq 0$, $\sigma \in [0, 1]$, and $\kappa_V, \kappa_T > 0$.
System~\eqref{sys:general-form} is considered with nonnegative initial conditions.
To avoid a trivial case, we further assume that $L(0)+I(0)+A(0)>0$.

%%%%%%%%%%%%%%%%%%%%%
%%%%%%%%%%%%%%%%%%%%%
%%%%%%%%%%%%%%%%%%%%%
%%%%%%%%%%%%%%%%%%%%%
\section{Mathematical analysis}
\label{sec:math-analysis}
All proofs in this section as well as some of the less formal discussions can be found in Appendix~\ref{supp:math-analysis}.

%%%%%%%%%%%%%%%%%%%%%%%%%%%%%%%
%%%%%%%%%%%%%%%%%%%%%%%%%%%%%%%
\subsection{The model without intervention}
\label{subsec:isolated-no-vacc}

The following model without any vaccination and treatment, i.e., without intervention, is important in and of itself as the baseline model.
We see later (Section~\ref{subsec:analysis-nodose-novacc}) that it also serves in other circumstances.
\begin{subequations}
   \label{sys:no-intervention}
   \begin{align}
        S' &= b + \omega_rR - \beta S(I+\eta A) - d S,
        \label{sys:no-vacc-dS} \\
        L' &= \beta S(I+\eta A)
        - (\varepsilon + d) L,
        \label{sys:no-vacc-dL} \\
        I' &= \pi\varepsilon L
        -(\gamma+\delta +d) I,
        \label{sys:no-vacc-dI} \\
        A' &= (1-\pi)\varepsilon L - (\gamma+d) A,
        \label{sys:no-vacc-dA} \\
        R' &= \gamma(I+A)-(\omega_r+d)R.
        \label{sys:no-vacc-dR}
    \end{align}
\end{subequations}

Model \eqref{sys:no-intervention} is classic, so we briefly summarize its properties here.
There are two equilibria, the disease-free equilibrium (DFE)
\[
    \bE_{0}^\eqref{sys:no-intervention}
    :=(S,L,I,A,R)=
    \left(\frac{b}{d},0,0,0,0\right)
\]
and the endemic equilibrium $\bE_{\star}^\eqref{sys:no-intervention}=(S^\star,L^\star,I^\star,A^\star,R^\star)$ taking the form
\begin{equation}\label{eq:EE-no-intervention}
\bE_{\star}^\eqref{sys:no-intervention}
    :=\left(
    \frac{b}{d \R_0},
    \frac{\gamma + \delta + d}{\pi \varepsilon} I^\star,
    I^\star,
    \frac{(1 - \pi)(\gamma + \delta + d)}{\pi (\gamma + d)} I^\star,
    \frac{\gamma \Theta}{\omega_r + d} I^\star
    \right),
\end{equation}
where
\[
    I^\star = \frac{b}{\frac{(\varepsilon + d)(\gamma + \delta + d)}{\pi \varepsilon} - \frac{\omega_r \gamma \Theta}{\omega_r + d}}\
    \frac{\R_0 - 1}{\R_0},
    \quad
    \Theta = 1+\frac{(1 - \pi)(\gamma + \delta + d)}{\pi(\gamma + d)}
\]
and the basic reproduction number $\R_0$ is given below in \eqref{eq:R0-no-vaccination}.

Here and elsewhere, we use the equation number of the system under consideration as a superscript to allow to easily understand what system the quantity is referring to.
However, $\R_0$ is the basic reproduction number for the disease without intervention and is therefore also ``shared'' by all special cases.
As a consequence, it is not superscripted with \eqref{sys:no-intervention}.

Following the reasoning in Appendix~\ref{supp:math-analysis-proof-no-intervention}, we find the basic reproduction number as
\begin{equation}\label{eq:R0-no-vaccination}
    \R_0 =
    \frac{\beta\varepsilon}{\varepsilon+d}
    \left(
    \frac{\pi}{\gamma + \delta + d} +
    \frac{ (1-\pi)\eta}{\gamma + d}
    \right)
    \frac bd.
\end{equation}
The following threshold result then holds.
\begin{proposition}\label{prop:isolated-novacc-GAS}
    If $\R_0<1$, then the disease-free equilibrium $\bE_{0}^\eqref{sys:no-intervention}$ is globally asymptotically stable in $\IR_+^5$ and the endemic equilibrium $\bE_{\star}^\eqref{sys:no-intervention}$ is not biologically relevant.
    If $\R_0>1$, then the disease-free equilibrium $\bE_{0}^\eqref{sys:no-intervention}$ is unstable and the endemic equilibrium $\bE_{\star}^\eqref{sys:no-intervention}$ is biologically relevant.
\end{proposition}

See the proof in Appendix~\ref{supp:math-analysis-proof-no-intervention}.
Note that the model without vaccination \eqref{sys:no-intervention} was studied in a slightly simpler case in \cite{lamichhane2015global}: that model did not have loss of immunity, i.e., a flow back to $S$ from $R$, so that global asymptotic stability of the endemic equilibrium could be claimed.

%%%%%%%%%%%%%%%%%%%%%%%%%%%%%%%%%%%%%%%%%%%%%%%%%%%%%%%%%%%%%%%%%%%%%%%%%%%%%%%%%%%%%%%%%%%

%%%%%%%%%%%%%%%%%%%%%%%%%%%%%%%
%%%%%%%%%%%%%%%%%%%%%%%%%%%%%%%
\subsection{Unreplenished doses lead to no-intervention scenario}
\label{subsec:analysis-nodose-novacc}
Next, let us consider the no-replenishment scenario, i.e., $a_V(D_V)\equiv a_T(D_T)\equiv 0$.
Recall that by assumption, $p$, $v$, and $g$ are $C^1$, nondecreasing, and satisfy $p(0)=v(0)=g(0)=0$.
Denote $\bE_{0}^{\eqref{sys:general-form}}:=(\bE_{0}^{\eqref{sys:no-intervention}},0,0,0)$ and $\bE_{b}^{\eqref{sys:general-form}}:=(\bE_{\star}^{\eqref{sys:no-intervention}},0,0,0)$ the vectors $\bE_{0}^{\eqref{sys:no-intervention}}$ and $\bE_{\star}^{\eqref{sys:no-intervention}}$ padded with three terminal zeros.
We call $\bE_{b}^{\eqref{sys:general-form}}$ a boundary equilibrium because it lies on the boundary $\{V=D_V=D_T=0\}$ of $\IR_+^8$, even though the disease remains endemic there ($L,I,A>0$).
The following result holds.

\begin{lemma}\label{lemma:nodose-novacc}
Assume that the vaccine and treatment supplies are never replenished, i.e., $a_V(D_V) \equiv a_T(D_T)\equiv 0$, and that $\kappa_V,\kappa_T>0$.
Then regardless of the initial conditions $D_V(0)\geq 0$ and $D_T(0)\geq 0$, every solution of \eqref{sys:general-form} in $\IR_+^8$ satisfies $(V(t), D_V(t), D_T(t))\to (0,0,0)$ as $t\to\infty$, so \eqref{sys:general-form} limits to the system without intervention \eqref{sys:no-intervention}, and its threshold dynamics are governed by $\R_0$ in \eqref{eq:R0-no-vaccination} as follows.
\begin{itemize}
    \item If $\R_0<1$, then the disease-free equilibrium $\bE_{0}^{\eqref{sys:general-form}}$ of \eqref{sys:general-form} is globally asymptotically stable in $\IR_+^8$ and the boundary equilibrium $\bE_{b}^{\eqref{sys:general-form}}$ is not biologically relevant.
    \item If $\R_0>1$, then the disease-free equilibrium $\bE_{0}^{\eqref{sys:general-form}}$ is unstable and the boundary equilibrium $\bE_{b}^{\eqref{sys:general-form}}$ is biologically relevant.
\end{itemize}
\end{lemma}

See the proof in Appendix~\ref{supp:analysis-nodose-novacc}.

%%%%%%%%%%%%%%%%%%%%%%%%%%%%%%
%%%%%%%%%%%%%%%%%%%%%%%%%%%%%%
\subsection{Case of unlimited dose supply}
\label{subsec:model-unlimited-vaccine}
Let us now suppose that vaccine and treatment doses are replenished sufficiently frequently that interventions are never disrupted by supply shortages.
Mathematically, this requires the dose stockpiles $D_V(t)$ and $D_T(t)$ to remain strictly positive.
Since the susceptible and infected populations are bounded by the demographic carrying capacity ($S(t) \leq b/d$ and $I(t) \leq b/d$), the maximal instantaneous consumption rates of vaccine and treatment doses are also strictly bounded.
Therefore, sufficient conditions for doses to never be completely depleted are the existence of some constants $\underline{D}_V > 0$ and $\underline{D}_T > 0$ such that the replenishment rates outpace the worst-case consumption and decay rates, namely,
\begin{align}
    a_V(\underline{D}_V) &> p(\underline{D}_V) b + v(\underline{D}_V) \frac{b}{d} + \kappa_V \underline{D}_V, \label{eq:condition-unlimited-V} \\
    a_T(\underline{D}_T) &> g(\underline{D}_T) \frac{b}{d} + \kappa_T \underline{D}_T. \label{eq:condition-unlimited-T}
\end{align}
Assume that the functions $a_V$, $a_T$, $p$, $v$, and $g$ are such that \eqref{eq:condition-unlimited-V}--\eqref{eq:condition-unlimited-T} hold, so that the regions $\{D_V \geq \underline{D}_V\}$ and $\{D_T \geq \underline{D}_T\}$ are positively invariant once $N(t)\leq b/d$.
In the idealized unlimited-dose limit where $\underline{D}_V$ and $\underline{D}_T$ lie in the saturation plateau of $p$, $v$, and $g$ (or when the stockpiles are approximated by constant values $\bar{D}_V$ and $\bar{D}_T$), the intervention rates remain fixed at $p^\star = p(\bar{D}_V)$, $v^\star = v(\bar{D}_V)$, and $g^\star = g(\bar{D}_T)$.
Under this constant-rate approximation, \eqref{sys:general-form} reduces to
\begin{subequations}
   \label{sys:unlimited-vacc}
   \begin{align}
        S' &= (1 - p^\star) b + \omega_rR+\omega_vV - \beta S(I+\eta A) - (v^\star + d) S,
        \label{sys:unlimited-vacc-dS} \\
        L' &= \beta (S+(1-\sigma) V)(I+\eta A) - (\varepsilon + d) L,
        \label{sys:unlimited-vacc-dL} \\
        I' &= \pi\varepsilon L -(\gamma+g^\star+\delta +d) I,
        \label{sys:unlimited-vacc-dI} \\
        A' &= (1-\pi)\varepsilon L - (\gamma+d) A,
        \label{sys:unlimited-vacc-dA} \\
        R' &= (\gamma+g^\star)I+\gamma A-(\omega_r+d)R,
        \label{sys:unlimited-vacc-dR} \\
        V' &= p^\star b + v^\star S - (1-\sigma)\beta V(I+\eta A) -(\omega_v+d)V.
        \label{sys:unlimited-vacc-dV}
   \end{align}
\end{subequations}

System \eqref{sys:unlimited-vacc} is almost identical to the model in \cite{ArinoMilliken2022a}, save for the treatment term $g^\star$ that is absent in \cite{ArinoMilliken2022a} and the fact that we denote $\pi$ what is denoted $1-\pi$ there.
As a consequence, the analysis is very similar to that in \cite{ArinoMilliken2022a} and the following adapted results from \cite{ArinoMilliken2022a} are given here without proof.

First, the disease-free equilibrium of the system with unlimited doses \eqref{sys:unlimited-vacc} is
\begin{equation}\label{eq:DFE-unlimited-vacc}
\bE_0^{\eqref{sys:unlimited-vacc}} = (S_0,0,0,0,0,V_0),
\end{equation}
where $S_0$ and $V_0$ are given by
\[
S_0 = \frac{(1-p^\star)d+\omega_v}{v^\star+\omega_v+d}\; \frac bd
\quad\textrm{and}\quad
V_0 = \frac{p^\star d+v^\star}{v^\star+\omega_v+d}\;\frac bd.
\]
%%%%%%%%%%%%%%%%%%%%%%%%%%%%%%%%%%%%%%%%%%%%%%%%%%%%%%%%%%%%%%%%%%%%%%%%%%%%%%%%%%%%%%%%%%%%%%%%%%%%%%%%%%%%%%%%%%
In terms of the local asymptotic stability of this equilibrium, we have the following classical result.
\begin{proposition}\label{prop:unlimited-stability-DFE}
The intervention reproduction number for \eqref{sys:unlimited-vacc} is given by
\begin{equation}\label{eq:Rv_SVLIARS}
\R_v^{\eqref{sys:unlimited-vacc}}=\frac{\lambda}{\varepsilon+d}\left(S_0+(1-\sigma)V_0\right),
\end{equation}
where $S_0$ and $V_0$ are given as in \eqref{eq:DFE-unlimited-vacc} and $\lambda$ is defined as
\begin{equation}\label{eq:lambda}
\lambda = \beta\varepsilon
\frac{(1-\pi)\eta(\gamma+g^\star+\delta+d) + \pi(\gamma + d)}
{(\gamma + d)(\gamma+g^\star+\delta+d)}>0.
\end{equation}
If $\R_v^{\eqref{sys:unlimited-vacc}}<1$, then the disease-free equilibrium $\bE_0^{\eqref{sys:unlimited-vacc}}$ of \eqref{sys:unlimited-vacc} is locally asymptotically stable, while it is unstable if $\R_v^{\eqref{sys:unlimited-vacc}}>1$.
\end{proposition}
%%%%%%%%%%%%%%%%%%%%%%%%%%%%%%%%%%%%%%%%%%%%%%%%%%%%%%%%%%%%%%%%%%%%%%%%%%%%%%%%

Vaccine coverage can be evaluated at any point in state space, for instance at endemic equilibria, although a closed-form solution is rarely obtained.
When the system is at the disease-free equilibrium, though, the following expression holds (in which the dependence on $\bE_0^{\eqref{sys:unlimited-vacc}}$ is indicated),
\begin{equation}\label{eq:v-c-DFE-unlimited-vacc}
v_c\left(\bE_0^{\eqref{sys:unlimited-vacc}}\right)
:= \frac{V_0}{S_0+V_0}=\frac{p^\star d+v^\star}{v^\star+\omega_v+d}.
\end{equation}
This expression also provides a useful form of the nonzero components of the DFE as
\begin{equation}\label{eq:S0_V0_fct_vcE0}
S_0 = \left(1-v_c\left(\bE_0^{\eqref{sys:unlimited-vacc}}\right)\right)\bar S_0
\quad\text{and}\quad
V_0=v_c\left(\bE_0^{\eqref{sys:unlimited-vacc}}\right)\bar S_0,
\end{equation}
where $\bar S_0=b/d$ is the $S$ component of the disease-free equilibrium for the model without vaccination \eqref{sys:no-intervention}.
In the absence of intervention ($p^\star = v^\star = g^\star = 0$), the intervention reproduction number $\R_v^{\eqref{sys:unlimited-vacc}}$ evaluates to $\R_0$ as given in \eqref{eq:R0-no-vaccination}.

As can be expected, imperfect and waning vaccination can induce a classical backward bifurcation in \eqref{sys:unlimited-vacc}, creating bistability when $\R_v^{\eqref{sys:unlimited-vacc}}<1$.
In our Latin hypercube search across parameter space (Appendix~\ref{supp:backward-bifurcation}), parameter combinations satisfying the backward-bifurcation condition appeared only outside biologically plausible regimes for the diseases considered here; we provide the formal conditions and numerical search details in Appendix~\ref{supp:backward-bifurcation} for completeness.

%%%%%%%%%%%%%%%%%%%%%%%%%%%
%%%%%%%%%%%%%%%%%%%%%%%%%%%
\subsection{The full model \emph{at equilibrium} is typically unlimited}
\label{sec:analysis-full-model}
Having considered the particular cases of limited and unlimited doses, we now return to the full model \eqref{sys:general-form}.
To avoid the case of dose run-out already considered, we assume throughout this section that $a_V$ and $a_T$ are nonnegative, nonincreasing $C^1$ functions on $\IR_+$ satisfying $a_V(0)>0$ and $a_T(0)>0$, together with the standing assumption $\kappa_V,\kappa_T>0$.

We begin by considering in Appendix~\ref{supp:analysis-full-model-on-invariant-set} the dynamics on the positively invariant set $\Omega_{L=I=A=0}^{\eqref{sys:general-form}}$.
The system \eqref{sys:general-form-DFE} on this set involves $S$, $R$, $V$, $D_V$, and $D_T$, as well as the functions $a_V$, $a_T$, $p$, and $v$.
We show that \eqref{sys:general-form-DFE} has a unique equilibrium, with $R^\star=0$, $D_V^\star>0$, and $D_T^\star>0$, and that this equilibrium is globally asymptotically stable in $\Omega_{L=I=A=0}^{\eqref{sys:general-form}}$.

Consider then the full state space $\IR_+^8$, i.e., where initial conditions are not restricted to the invariant set $\Omega_{L=I=A=0}^{\eqref{sys:general-form}}$.
There, the equilibrium of the disease-free subsystem \eqref{sys:general-form-DFE} naturally extends to the disease-free equilibrium of the full system \eqref{sys:general-form} by setting zero components for the infected compartments $L$, $I$, and $A$.
Therefore, the disease-free equilibrium of \eqref{sys:general-form} takes the form
\begin{equation}\label{eq:DFE-full-system}
\bE_0^{\eqref{sys:general-form}} = (S^\star, 0, 0, 0, 0, V^\star, D_V^\star, D_T^\star),
\end{equation}
where $S^\star$ and $V^\star$ are the components $S_0$ and $V_0$ of the unlimited-dose DFE \eqref{eq:DFE-unlimited-vacc}, evaluated with the steady-state vaccination parameters $p^\star = p(D_V^\star)$ and $v^\star = v(D_V^\star)$.
We thus have the following result about the general system \eqref{sys:general-form}.

\begin{theorem}\label{th:full-model-DFE-and-persistence}
    Assume that $a_V$ and $a_T$ are nonnegative, nonincreasing $C^1$ functions on $\IR_+$ with $a_V(0)>0$ and $a_T(0)>0$, and that $\kappa_V,\kappa_T>0$.
    Then $\R_v^{\eqref{sys:general-form}}=\R_v^{\eqref{sys:unlimited-vacc}}$ as given by \eqref{eq:Rv_SVLIARS}, evaluated with $p^\star=p(D_V^\star)$, $v^\star=v(D_V^\star)$, and $g^\star=g(D_T^\star)$, and the following holds.
    \begin{itemize}
        \item If $\R_v^{\eqref{sys:general-form}} < 1$, then the disease-free equilibrium $\bE_0^{\eqref{sys:general-form}}$ of \eqref{sys:general-form} is locally asymptotically stable.
        \item If $\R_v^{\eqref{sys:general-form}} > 1$, then $\bE_0^{\eqref{sys:general-form}}$ is unstable; if in addition $\eta>0$, then the disease is uniformly persistent, i.e., there exists a constant $c > 0$ such that any solution of \eqref{sys:general-form} with initial conditions satisfying $L(0)+I(0)+A(0)>0$ satisfies
        \[
            \liminf_{t\to\infty} (L(t)+I(t)+A(t)) \geq c.
        \]
    \end{itemize}
\end{theorem}

See the proof in Appendix~\ref{supp:analysis-full-model-at-DFE}, which also notes how the persistence conclusion adapts to $L(t)+I(t)$ when $\eta=0$.

%%%%%%%%%%%%%%%%%%%%%%%%%%%
\subsubsection{The endemic equilibrium}
\label{sec:analysis-full-model-at-EE}

By analogous reasoning to the disease-free equilibrium, if the full system \eqref{sys:general-form} admits an endemic equilibrium $\bE_\star^{\eqref{sys:general-form}}=(\hat{S}, \hat{L}, \hat{I}, \hat{A}, \hat{R}, \hat{V}, \hat{D}_V, \hat{D}_T)$, its first six components must match an endemic equilibrium of the unlimited-dose system \eqref{sys:unlimited-vacc}, evaluated with the constant intervention parameters $\hat{p} = p(\hat{D}_V)$, $\hat{v} = v(\hat{D}_V)$, and $\hat{g} = g(\hat{D}_T)$.
Note that the endemic steady-state dose levels $\hat{D}_V$ and $\hat{D}_T$ generally differ from the disease-free values $D_V^\star$ and $D_T^\star$ in \eqref{eq:DFE-full-system} (in particular, $\hat{D}_T < D_T^\star$ whenever $g(\hat{D}_T)\hat{I} > 0$).

However, contrary to the disease-free equilibrium case where the dose equations could be decoupled from the infectious compartments, the components $\hat{S}$ and $\hat{I}$ now depend nontrivially on the intervention rates and thus on $\hat{D}_V$ and $\hat{D}_T$.
To find the endemic steady-state dose levels $\hat{D}_V$ and $\hat{D}_T$, one must solve the fully coupled system
\begin{align*}
    a_V(\hat{D}_V) &= p(\hat{D}_V)b + v(\hat{D}_V)\hat{S}(\hat{D}_V, \hat{D}_T) + \kappa_V \hat{D}_V, \\
    a_T(\hat{D}_T) &= g(\hat{D}_T)\hat{I}(\hat{D}_V, \hat{D}_T) + \kappa_T \hat{D}_T.
\end{align*}
Because the consumption terms $v(\hat{D}_V)\hat{S}(\hat{D}_V, \hat{D}_T)$ and $g(\hat{D}_T)\hat{I}(\hat{D}_V, \hat{D}_T)$ multiply a nondecreasing intervention rate by a compartment equilibrium that can decrease as dose availability increases, these consumption terms are not necessarily monotonic in $\hat{D}_V$ and $\hat{D}_T$.
For example, higher treatment availability increases $g(\hat{D}_T)$ while reducing endemic prevalence $\hat{I}(\hat{D}_V, \hat{D}_T)$, so their product $g(\hat{D}_T)\hat{I}(\hat{D}_V, \hat{D}_T)$ can decrease over part of its range.
Consequently, the coupled stock-balance equations above may admit multiple positive roots, meaning that endogenous stock-consumption feedback in \eqref{sys:general-form} can theoretically induce multiple endemic equilibria and bi- or multistability independently of the vaccine-induced backward bifurcation of \eqref{sys:unlimited-vacc}.

%%%%%%%%%%%%%%%%%%%%%
%%%%%%%%%%%%%%%%%%%%%
%%%%%%%%%%%%%%%%%%%%%
%%%%%%%%%%%%%%%%%%%%%
\section{Computational considerations}
\label{sec:computational}
The mathematical analysis in Section~\ref{sec:math-analysis} shows that in general, \eqref{sys:general-form} operates like much simpler systems that are either without intervention or not dose-limited.
This is, however, a consequence of the formulation as a system of ordinary differential equations (ODEs), precluding proper disease extinction, as well as the analysis focusing on long-term dynamics.
Here, we investigate the system computationally.

\begin{table}[htbp]
    \centering
    \footnotesize
    \setlength{\tabcolsep}{4pt}
    \caption{Typical shelf life of common vaccines and therapeutic drugs.}
    \label{tab:drug-shelf-life}
    \begin{tabular}{@{}llcl@{}}
        \toprule
        \textbf{Agent} & \textbf{Condition/formulation} & \textbf{Shelf life} & \textbf{Source} \\
        \midrule
        \multicolumn{4}{@{}l}{\textit{Vaccines}} \\
        \midrule
        Influenza & Seasonal formulation (unopened) & $\sim$1 year & FDA~\cite{fda2024influenza} \\
        \addlinespace
        MMR (measles) & Lyophilized (unopened, refrigerated) & 1--2 years & FDA~\cite{fda2023mmrii} \\
        & Reconstituted (refrigerated, no light) & 8 hours & FDA~\cite{fda2023mmrii} \\
        \addlinespace
        COVID-19 (mRNA) & Unopened (ultra-cold freezer) & 18--24 months & CDC/FDA~\cite{cdc2024covidvaccines} \\
        & Thawed/punctured (refrigerated) & 12 hours--30 days & CDC/FDA~\cite{cdc2024covidvaccines} \\
        \midrule
        \multicolumn{4}{@{}l}{\textit{Therapeutics \& antivirals}} \\
        \midrule
        Oseltamivir & Capsules (unopened, room temperature) & 5--7 years & FDA~\cite{fda2024expiration} \\
        (Tamiflu) & Oral suspension (reconstituted) & 10--17 days & FDA~\cite{fda2019tamiflu} \\
        \addlinespace
        Paxlovid & Tablets (unopened, room temperature) & 24 months & FDA~\cite{fda2024expiration} \\
        \addlinespace
        Amoxicillin & Solid capsules (unopened, room temperature) & 2--3 years & FDA~\cite{fda2024amoxicillin} \\
        & Oral liquid suspension (refrigerated) & 14 days & FDA~\cite{fda2024amoxicillin} \\
        \bottomrule
    \end{tabular}
\end{table}

Parameters are as listed in Table~\ref{tab:parameter_values}; for information, shelf lives of some common vaccines and treatments are given in Table~\ref{tab:drug-shelf-life}.
Regarding functions, we use Hill functions for $p(D_V)$, $v(D_V)$, and $g(D_T)$, which satisfy $p(0)=v(0)=g(0)=0$ without shifting.
The Hill-function definition (Figure~\ref{fig:hill-p}) and the numerical integration and stochastic simulation algorithms are detailed in Appendix~\ref{supp:computational}.

In Sections~\ref{sec:computational-unreplenished-doses}, \ref{subsec:numerics-unlimited-vaccine}, and \ref{subsec:incorporation-of-dose-cost}, plots considering the reduction in cumulative incidence are obtained as follows.
For the $\R_0$ under consideration, a baseline simulation is run using no intervention and initial conditions the endemic equilibrium.
The total number of new infections $\int_0^T \beta (S(t)+(1-\sigma)V(t))(I(t)+\eta A(t))\,dt$ (with $V(t)\equiv 0$ in the no-intervention baseline) and the number of observed new infections $\int_0^T \pi\varepsilon L(t)\,dt$ are computed for $T=1$ year and $T=5$ years.
Then the different scenarios are run with the same initial conditions and the same integrals are computed, from which we derive the reduction in disease burden.
In Section~\ref{subsec:numerics-impulsive-doses}, the setup is the same, but instead of the endemic equilibrium, the initial condition is taken as $L(0)=5$ and other infected compartments empty.
This is done so that the reactive schedule delivery evaluated there can be triggered.

%%%%%%%%%%%%%%%%%%%%%
%%%%%%%%%%%%%%%%%%%%%
\subsection{Case of unreplenished doses}
\label{sec:computational-unreplenished-doses}
Lemma~\ref{lemma:nodose-novacc} implies that in the case of unreplenished doses, \eqref{sys:general-form} exhibits straightforward long-term deterministic behavior.
Indeed, the behavior is governed by a reproduction number that does not take into account interventions, so that the number $D_V(0)$ and $D_T(0)$ of doses initially available plays no role in the long-term dynamics.

\begin{figure}[htbp]
    \centering
    \begin{subfigure}{.49\textwidth}
        \includegraphics[width=\textwidth]{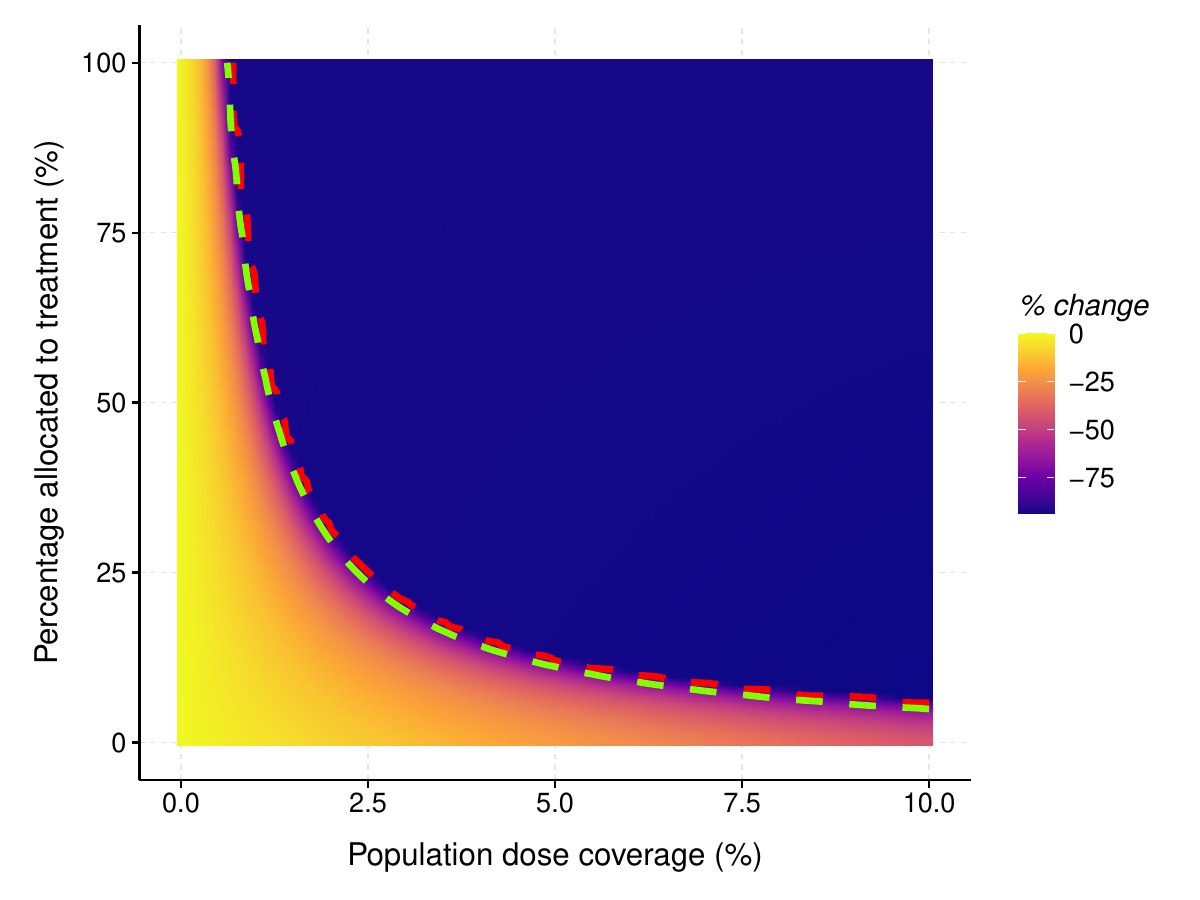}
        \caption{Percentage covered and treatment}
        \label{fig:change-AUC-I-vs-initial-doses-heatmap}
    \end{subfigure}
    \begin{subfigure}{.49\textwidth}
        \includegraphics[width=\textwidth]{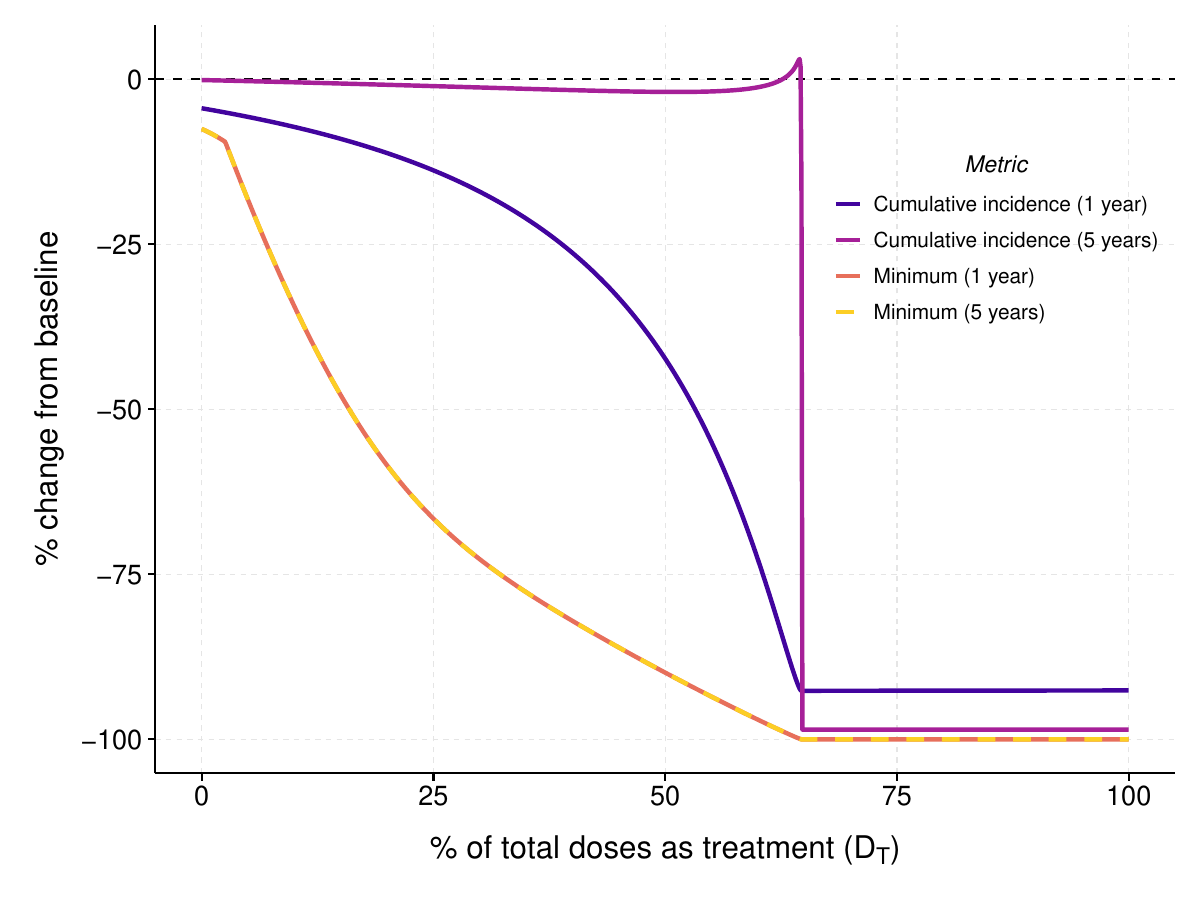}
        \caption{Function of percentage of treatment doses}
        \label{fig:change-AUC-I-vs-initial-doses-1d}
    \end{subfigure}
    \caption{Effect of the number of doses on various disease spread metrics when $\R_0=3$ with unreplenished doses.
    (a) Percentage change in the cumulative observed incidence over the first year compared to the case with no interventions, as a function of the percentage of the population that can initially be covered and the percentage of these doses that are treatment doses.
    The red and green curves show a 90\% and 75\% reduction, respectively.
    (b) Percentage change in the cumulative observed incidence and the value of the minimum of $I(t)$ over one and five years for 1\% initial population dose coverage (1,000 doses in a population of 100,000), as a function of the percentage of treatment doses among the total number of doses.}
    \label{fig:change-AUC-I-vs-initial-doses}
\end{figure}

Intuitively, though, interventions should play a role, at least until doses have run out, even when $\R_0>1$.
We therefore now explore this transient dynamics from a computational point of view.
To begin, let us confirm the intuition that the number of doses matters, even when $\R_0>1$.
To do so, we consider a baseline scenario consisting of \eqref{sys:no-intervention} with $\R_0=3$.
To focus on the effect of doses, we set the initial conditions using \eqref{eq:EE-no-intervention}, which is the endemic equilibrium the system with unreplenished doses tends to.
Thus, doses act essentially as a perturbation of the system away from its natural resting state.
We investigate this perturbation using two measures: the cumulative observed (symptomatic) incidence $\int_0^T \pi\varepsilon L(t)\,dt$ and the minimum value reached by $I(t)$.
(We get the same type of results, not shown, when considering total cumulative infections and total infected prevalence $L+I+A$.)
In Figure~\ref{fig:change-AUC-I-vs-initial-doses}, we consider a disease-free population size $N=100{,}000$ and $\R_0=3$.
With the baseline parameters of Table~\ref{tab:parameter_values}, the endemic equilibrium \eqref{eq:EE-no-intervention} has $I^\star \approx 257$ and a living population $N^\star = N - (\delta/d)I^\star \approx 96{,}700$ (about 97\% of $N$; population dose coverage on the horizontal axis is reported relative to $N$).
Figure~\ref{fig:change-AUC-I-vs-initial-doses-heatmap} shows how cumulative observed incidence changes with respect to baseline over a 1-year period, with initial population dose coverage varied in $[0, 10\%]$ and the percentage of treatment doses in that coverage varied from 0 to 100\%.
For this duration, the maximum effect is obtained by increasing treatment doses rather than vaccine doses.

To contrast this, Figure~\ref{fig:change-AUC-I-vs-initial-doses-1d} considers a vertical cut across Figure~\ref{fig:change-AUC-I-vs-initial-doses-heatmap} at 1\% initial population dose coverage (1,000 doses in $N=100{,}000$) and compares 1-year and 5-year time horizons.
While the minimum symptomatic prevalence $I(t)$ is reached within the first year and therefore coincides across both horizons, cumulative observed incidence behaves differently over 1 year versus 5 years.
When the treatment share lies below roughly 65\%, initial doses are exhausted before transmission is driven near zero and immunity wanes (vaccine protection in 200 days and disease-induced immunity in 5 years), so the epidemic rebounds well before year 5 and the 5-year cumulative reduction collapses toward zero (even rising slightly above baseline just below the threshold due to a post-exhaustion rebound oscillation).
Above roughly 65\% treatment share, treatment suppresses prevalence near zero during the first year and leftover treatment courses (with a 3-year shelf life) maintain suppression across the 5-year window, so the 5-year cumulative reduction (about 99\%) slightly exceeds the 1-year reduction (about 93\%).

\begin{figure}[htbp]
    \centering
    \begin{subfigure}[b]{0.49\textwidth}
        \centering
        \includegraphics[width=\textwidth]{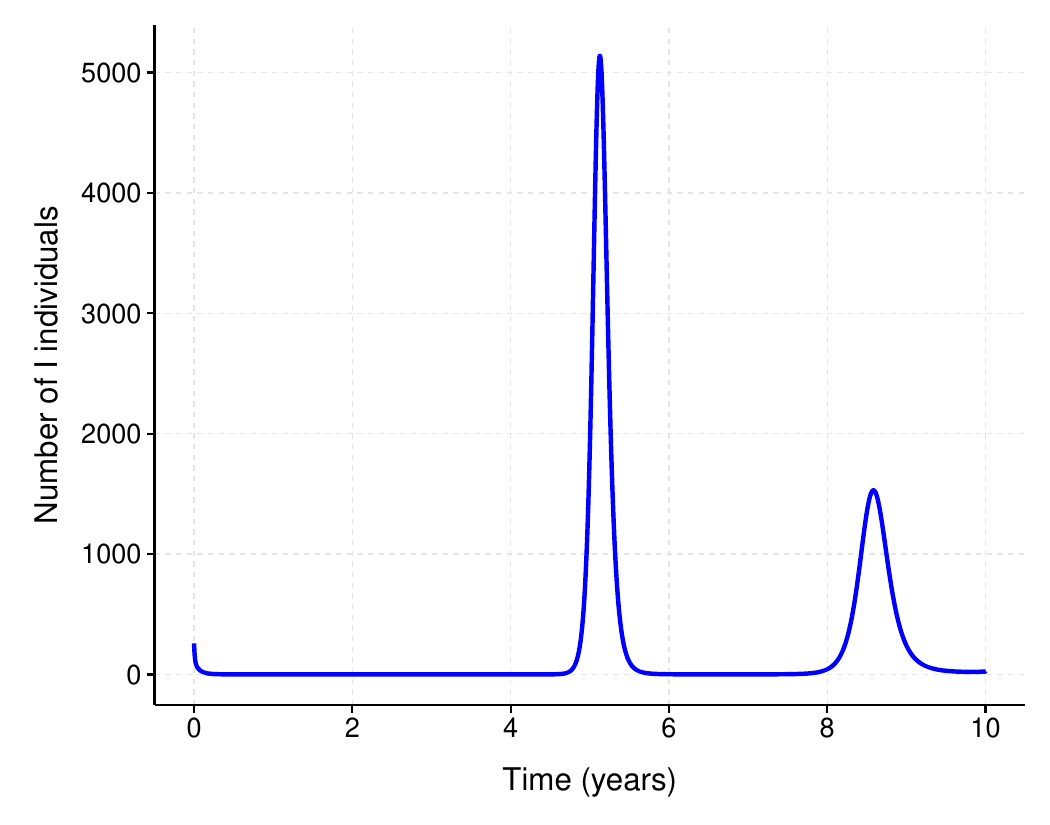}
        \caption{Linear scale}
        \label{fig:ODE-optimal-repartition-timeseries-linear}
    \end{subfigure}
    \hfill
    \begin{subfigure}[b]{0.49\textwidth}
        \centering
        \includegraphics[width=\textwidth]{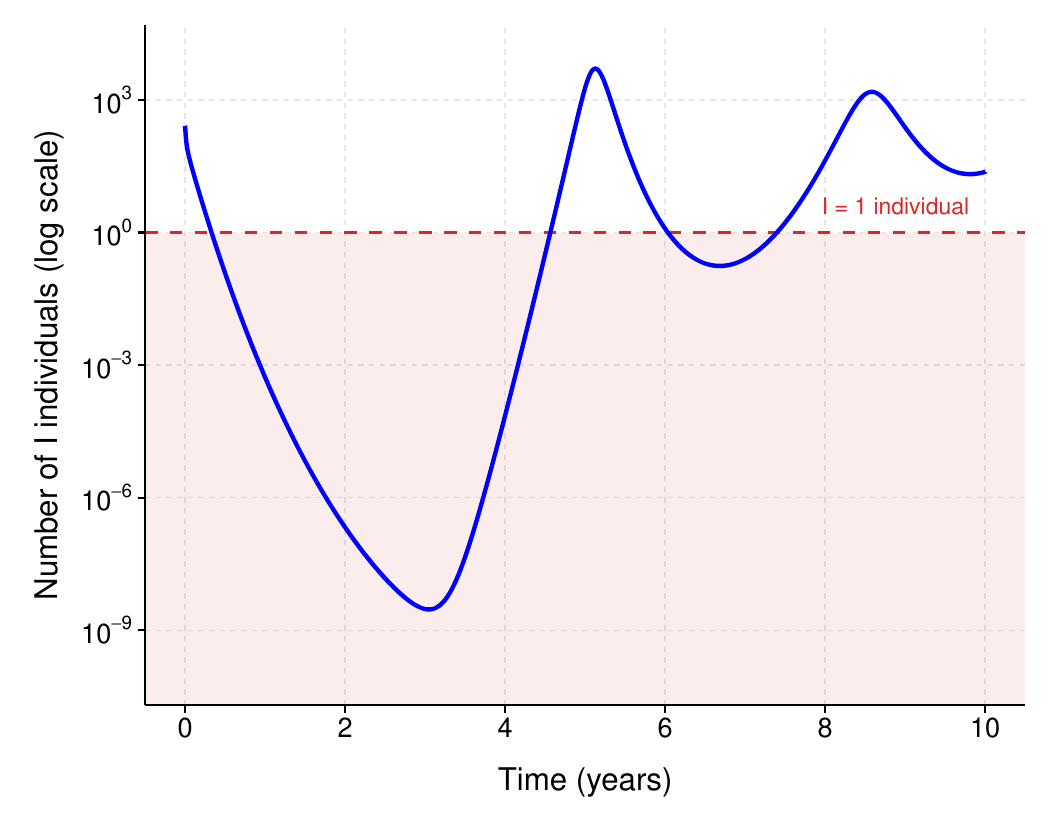}
        \caption{Logarithmic scale}
        \label{fig:ODE-optimal-repartition-timeseries-log}
    \end{subfigure}
    \caption{Time series of the number of symptomatically infectious individuals $I(t)$ over a 10-year horizon on (a) linear and (b) logarithmic scales for the repartition of 2,000 initial doses (2\% coverage: 638 treatment doses and 1,362 vaccine doses) that minimizes the 1-year minimum of $I(t)$, when $\R_0=3$.
    Initial conditions are taken to be the endemic equilibrium of the system without intervention.
    In (b), the dashed line marks $I=1$ individual and the shaded region indicates sub-unit prevalence ($I<1$).}
    \label{fig:ODE-optimal-repartition-timeseries}
\end{figure}

Figure~\ref{fig:ODE-optimal-repartition-timeseries} shows an atto-fox \cite{mollison1991dependence}, i.e., a spurious deterministic disease rebound resulting from continuous fractional individual counts, for the allocation of 2,000 doses (638 treatment and 1,362 vaccine doses) that minimizes $I(t)$ over the first year.
(This optimum sits on a nearly flat plateau across all treatment shares above roughly 32\%, where $\min I(t)$ varies by less than $10^{-3}\%$.)
Over the 10-year window, symptomatic prevalence $I(t)$ drops below 1 individual after roughly 4 months and reaches about $3\times 10^{-9}$ around year 3 (Figure~\ref{fig:ODE-optimal-repartition-timeseries-log}; total infected prevalence $L+I+A$ drops similarly, not shown) while leftover treatment doses gradually expire at rate $\kappa_T$.
Once $D_T$ falls below the half-saturation threshold $D_g^s$, treatment turns off; because the disease-free manifold $\{L=I=A=0\}$ is invariant under the flow of \eqref{sys:general-form} and solutions are unique, trajectories starting with $L(0)+I(0)+A(0)>0$ can never reach zero infected individuals in finite time, and because $\R_0=3>1$, prevalence subsequently rebounds toward the endemic equilibrium.

\begin{figure}[htbp]
    \centering
    \includegraphics[width=0.55\textwidth]{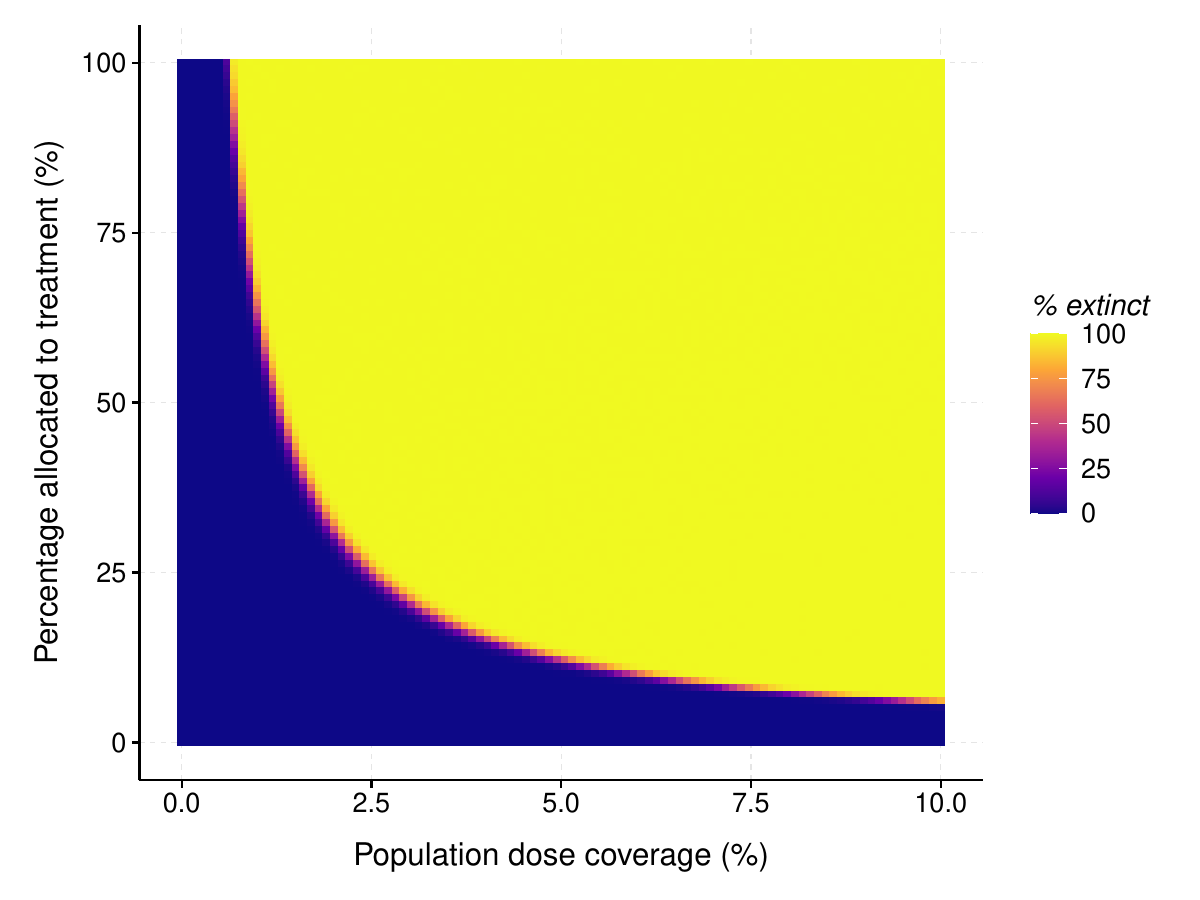}
    \caption{Percentage of 10,000 CTMC realizations achieving disease extinction within 1 year as a function of the percentage of the population that can be covered by doses and the percentage of these doses allocated to treatment (unreplenished doses scenario, $\R_0=3$).}
    \label{fig:CTMC-no-intervention-absorption}
\end{figure}

While the deterministic system ``automatically'' rebounds toward the endemic equilibrium, this hints that a stochastic system might fare differently, since the states with $L=I=A=0$ form an absorbing class.
We therefore consider the continuous-time Markov chain (CTMC) analog of \eqref{sys:general-form} (Appendix~\ref{supp:CTMC}) with transitions in Table~\ref{tab:general-CTMC-transitions}, evaluated here with $a_V(D_V)\equiv a_T(D_T)\equiv 0$.
In the CTMC, the dose-dependent rates $p(D_V)$, $v(D_V)$, and $g(D_T)$ use the same steep Hill parameters (Table~\ref{tab:parameter_values}) as the ODE, so rates remain near their maximal values whenever at least one dose is present and vanish identically once stocks hit zero.

In Figure~\ref{fig:CTMC-no-intervention-absorption}, we use the same initial condition as in Figure~\ref{fig:change-AUC-I-vs-initial-doses} (rounded to the nearest integers) and the same parameters.
Each pair of values $(D_V(0), D_T(0))$ is used for 10,000 realizations of the CTMC, and Figure~\ref{fig:CTMC-no-intervention-absorption} shows the percentage of realizations leading to disease extinction before 1 year.
Across the region where extinction occurs reliably, the conditional mean time to extinction is nearly uniform, with most extinctions occurring within 160--180 days.
Thus, starting with a sufficient number of doses, and in particular a sufficient share of treatment doses, brings the disease to stochastic extinction within 1 year with high probability at $\R_0=3$, whereas smaller initial stocks leave the epidemic persisting.

%%%%%%%%%%%%%%%%%%%%%%%%%%%%%%
%%%%%%%%%%%%%%%%%%%%%%%%%%%%%%
\subsection{Case of unlimited dose supply}
\label{subsec:numerics-unlimited-vaccine}

\begin{figure}[htbp]
    \centering
    \begin{subfigure}{0.49\textwidth}
        \includegraphics[width=\linewidth]{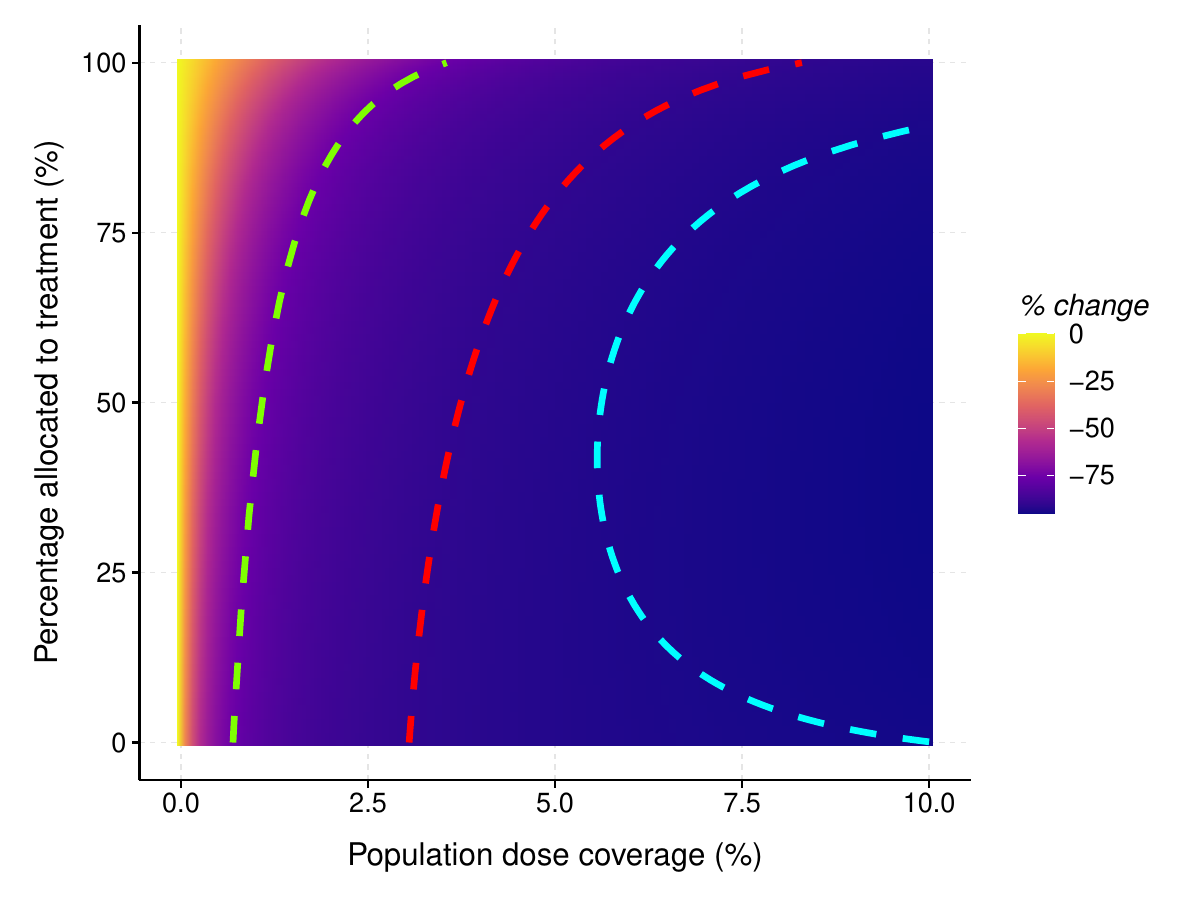}
        \caption{Reduction of cumulative incidence}
        \label{fig:coverage-vs-pcttreatment-unlimited-doses-AUC}
    \end{subfigure}
    \begin{subfigure}{0.49\textwidth}
        \includegraphics[width=\linewidth]{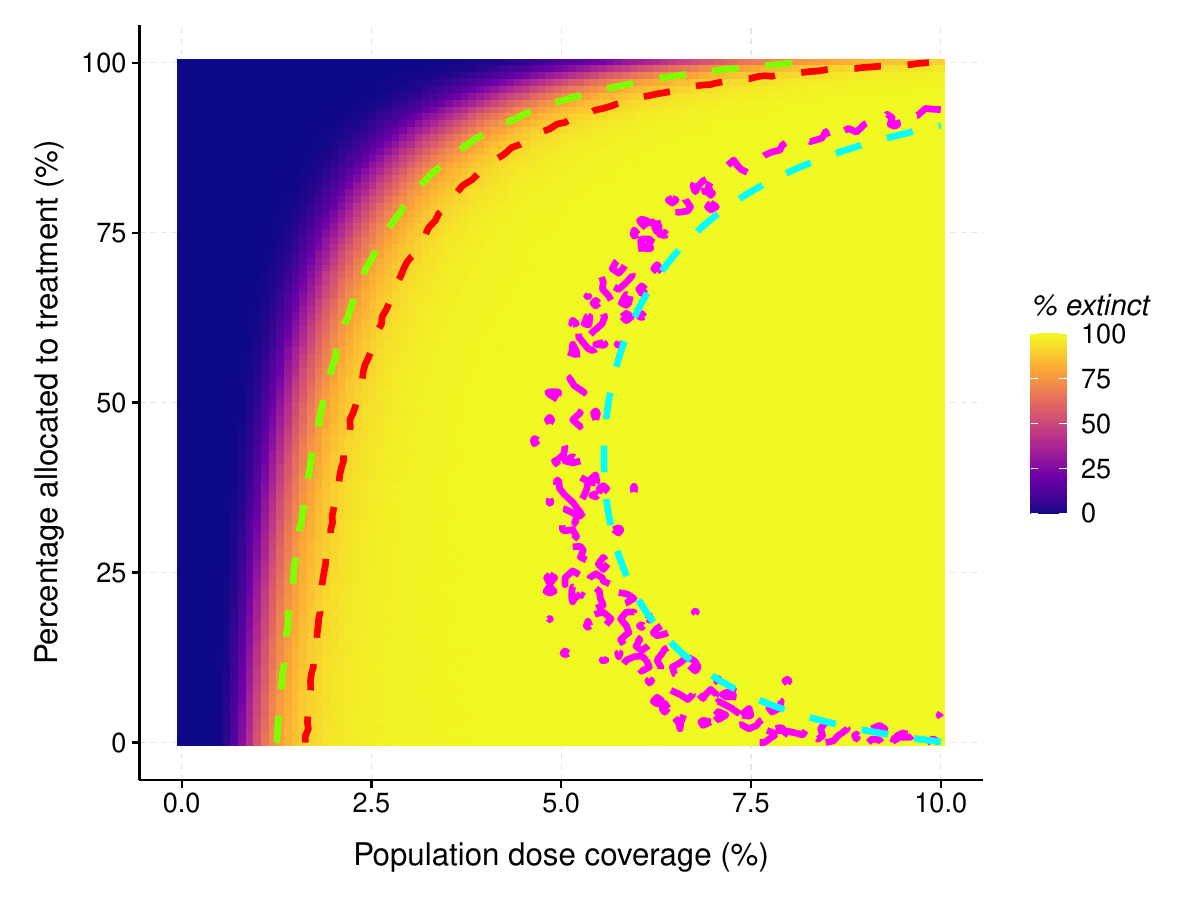}
        \caption{Percentage of extinctions}
        \label{fig:coverage-vs-pcttreatment-unlimited-doses-proba-extinction}
    \end{subfigure}
    \caption{Role of the percentage of the population covered by daily dose availability and the percentage of these doses allocated to treatment in the unlimited-dose system \eqref{sys:unlimited-vacc}, mapped via $p^\star=0$, $v^\star=(D_V/N)\text{ day}^{-1}$, and $g^\star=(D_T/N)\text{ day}^{-1}$ ($N=100{,}000$, $\R_0=3$, 1-year horizon).
    The cyan curve shows $\R_v^{\eqref{sys:unlimited-vacc}}=1$.
    (a) Percentage change in total cumulative incidence (cumulative entry into $L$) compared to a no-intervention baseline.
    (b) Percentage of 10,000 CTMC realizations that achieve disease extinction within 1 year.
    The red and green curves mark 90\% and 75\% reductions in total cumulative incidence in (a) and 90\% and 75\% extinction probabilities in (b).
    The magenta curve in (b) marks the 99.9\% extinction contour; its jaggedness reflects sampling variation across the 10,000 realizations.}
    \label{fig:coverage-vs-pcttreatment-unlimited-doses}
\end{figure}

To plot the unlimited-dose model \eqref{sys:unlimited-vacc} on the same coverage and treatment-share coordinate axes as Figures~\ref{fig:change-AUC-I-vs-initial-doses-heatmap} and \ref{fig:CTMC-no-intervention-absorption}, we set newborn vaccination to $p^\star = 0$ and convert each dimensionless coverage pair $(D_V/N, D_T/N)$ into constant daily \emph{per capita} rates $v^\star = (D_V/N)\text{ day}^{-1}$ and $g^\star = (D_T/N)\text{ day}^{-1}$ (with $N = 100{,}000$).
This dimensional mapping corresponds to replacing the indicated coverage every day rather than consuming a single initial stock, so comparison with Figures~\ref{fig:change-AUC-I-vs-initial-doses-heatmap} and \ref{fig:CTMC-no-intervention-absorption} is across shared coordinate axes rather than equal total dose consumption.
Figure~\ref{fig:coverage-vs-pcttreatment-unlimited-doses} shows that, quite similarly to the situation observed in Section~\ref{sec:computational-unreplenished-doses} with unreplenished doses, the mathematical analysis in the unlimited doses case only tells part of the story.
Although the situation is not quite as extreme here, we do note that while the curve $\R_v^{\eqref{sys:unlimited-vacc}}=1$ determines the theoretical threshold for elimination of the disease, there are regions where $\R_v^{\eqref{sys:unlimited-vacc}}>1$ that see a substantial reduction in the overall impact of the disease, as evidenced by the green and red curves in Figure~\ref{fig:coverage-vs-pcttreatment-unlimited-doses-AUC}, which show, respectively, 75\% and 90\% reductions in total cumulative incidence.
Likewise, the stochastic CTMC extinction probability in Figure~\ref{fig:coverage-vs-pcttreatment-unlimited-doses-proba-extinction} confirms that finite-time extinction can occur in parts of parameter space where $\R_v^{\eqref{sys:unlimited-vacc}}>1$, while approaching certainty inside the $\R_v^{\eqref{sys:unlimited-vacc}}<1$ region (see also Figure~\ref{fig:coverage-vs-pcttreatment-unlimited-doses-extra} in Appendix~\ref{supp:more-figures} for $\R_v^{\eqref{sys:unlimited-vacc}}$ and conditional mean extinction times).

%%%%%%%%%%%%%%%%%%%%%%%%%%%%%%
%%%%%%%%%%%%%%%%%%%%%%%%%%%%%%
\subsection{Impulsive intervention delivery}
\label{subsec:numerics-impulsive-doses}

In urgent epidemiological settings, relying on continuous intervention deployments or static initial stockpiles is often unrealistic.
Resources such as COVID-19 Vaccines Global Access (COVAX) batches or emergency interventions often arrive in discrete shipments.
To encode this, epidemiological models frequently employ impulsive differential equations \cite{donofrio2002stability,shulgin1998pulse}.

Here, we briefly investigate the consequences of such a type of dose replenishment.
The continuous ODE dynamics \eqref{sys:general-form} govern the system between shipments, while the discrete delivery of doses imposes instantaneous impulses on the available stockpiles $D_V(t)$ and $D_T(t)$,
\begin{subequations}
    \label{sys:impulses}
    \begin{align}
        D_V(t_k^+) &= D_V(t_k^-) + \Delta_V, \\
        D_T(t_k^+) &= D_T(t_k^-) + \Delta_T.
    \end{align}
\end{subequations}
Since this formulation departs from the assumptions of continuity of $D_V$ and $D_T$ used in the remainder of the manuscript, Appendix~\ref{supp:impulsive-existence} briefly outlines the existence and uniqueness of solutions for this formulation.

Below, we denote $D_{\text{total}}$ the total number of doses available for some period of time (one year in the examples below) at some remote central storage location.
We evaluate two primary logistical regimens.
\begin{enumerate}
    \item \emph{Fixed schedule delivery.}
    Shipments arrive at predetermined times.
    Here, the $t_k$ denote the $k$-th scheduled delivery time, i.e., $\{t_1,\ldots, t_{N_{\text{shipments}}}\}$ is a static schedule.
    At each $t_k$, $D_{\text{total}}/N_{\text{shipments}}$ doses are delivered.
    \item \emph{Reactive schedule delivery.}
    In this formulation, emergency stockpiles are deliberately held in reserve at the central location and deployed only if the observable incidence $\pi\varepsilon L$ exceeds a threshold $I_{\text{critical}}$.
    A fraction of the stockpile is then delivered, corresponding to $M$ months' worth of doses.
    If the observable incidence has not decreased below the threshold $I_{\text{critical}}$ after $M$ months, an additional delivery is made.
    Thus, we trigger a batch delivery at times $t_k$ satisfying one of the two conditions
    \begin{align*}
        \pi\varepsilon L(t_k^-) &= I_{\text{critical}} \text{ and } \frac{d}{dt}\pi\varepsilon L \Big|_{t=t_k^-} > 0, \\
        \pi\varepsilon L(t_k^-) &> I_{\text{critical}} \text{ and } t_k-t_{k-1}\geq M.
    \end{align*}
    This dynamical impulse sequence continues until the capacity of $D_{\text{total}}$ total deliveries is exhausted.
\end{enumerate}

\begin{figure}[htbp]
    \centering
    \begin{subfigure}[b]{0.49\textwidth}
        \centering
        \includegraphics[width=\textwidth]{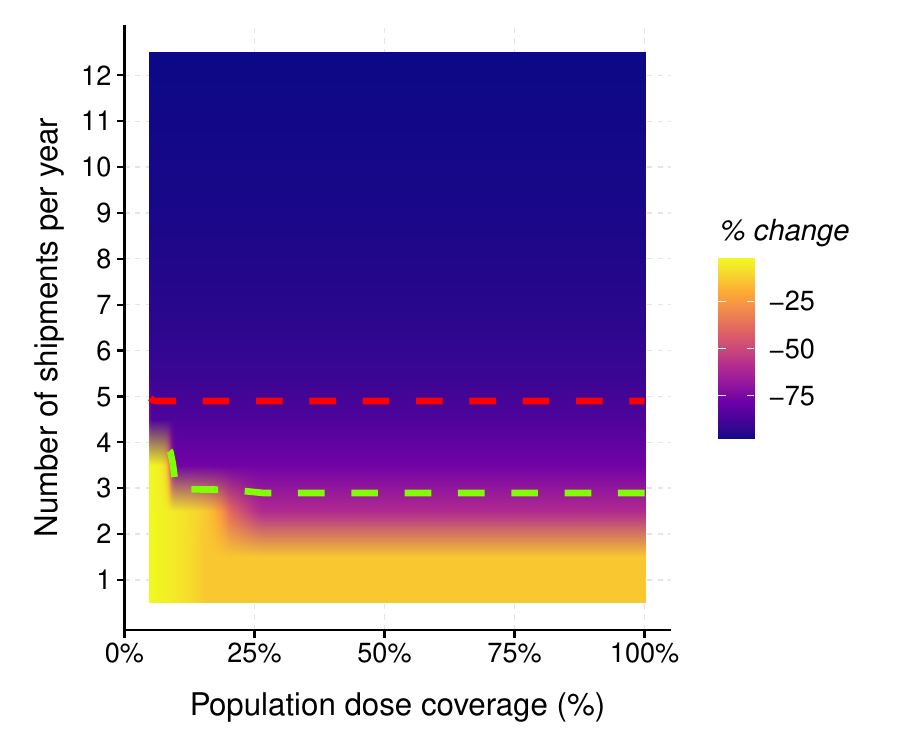}
        \caption{Fixed schedule}
        \label{fig:impulsive-heatmap-fixed}
    \end{subfigure}
    \hfill
    \begin{subfigure}[b]{0.49\textwidth}
        \centering
        \includegraphics[width=\textwidth]{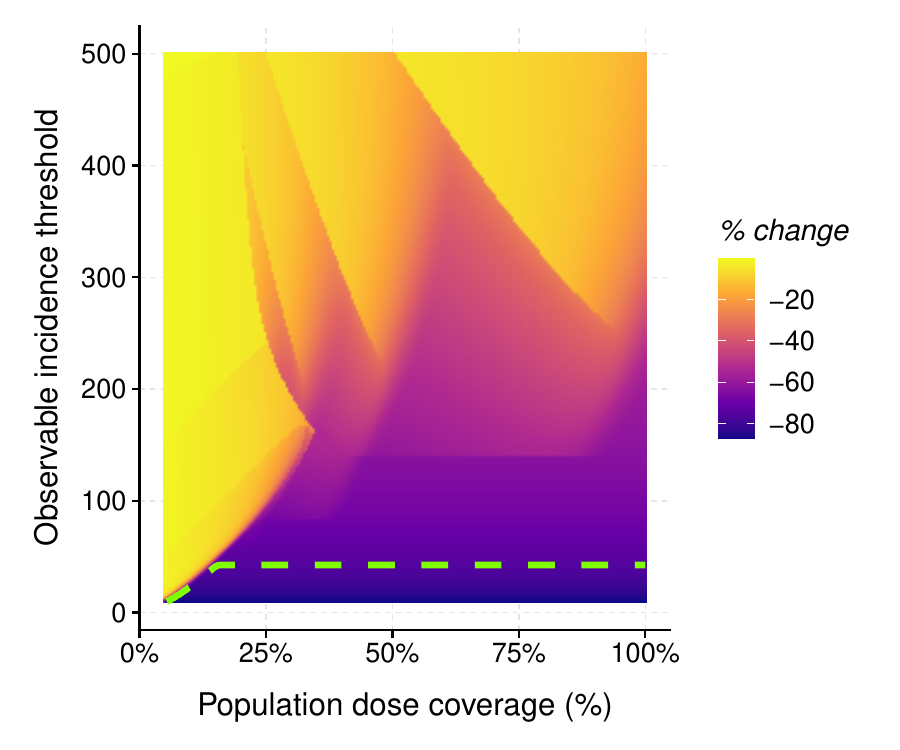}
        \caption{Reactive schedule (2 months)}
        \label{fig:impulsive-heatmap-reactive}
    \end{subfigure}
    \caption{Relative benefit (percentage change in cumulative detectable cases over 1 year compared to a scenario without intervention) when using impulsive treatment delivery under strict dose limitations ($\R_0=2.5$, $L(0)=5$, $D_V(0)=D_T(0)=0$, $N=100{,}000$).
    For both panels, the total number of treatment doses available over the year is shown on the $x$-axis as a percentage of $N$, and the red and green curves show a 90\% and 75\% reduction, respectively.
    (a) Fixed delivery schedule, with $D_{\text{total}}/N_{\text{shipments}}$ doses delivered at $t_k = 365k/(N_{\text{shipments}}+1)$ days for $k=1,\ldots,N_{\text{shipments}}$.
    (b) Reactive schedule, with 1/6 of available doses delivered when observable incidence $\pi\varepsilon L$ crosses the threshold $I_{\text{critical}}$ on the $y$-axis or remains above it two months after the previous delivery.}
    \label{fig:impulsive-heatmaps}
\end{figure}

Figure~\ref{fig:impulsive-heatmaps} explores impulsive strategies over a one-year period for an emerging outbreak with $\R_0=2.5$ seeded by $L(0)=5$ latently infected individuals in a susceptible population of $N=100{,}000$ ($S(0)=N-5$).
To simplify, only treatment doses are used, i.e., $D_V(0)=\Delta_V=0$, and both scenarios start with an empty local stockpile ($D_T(0)=0$).
In the fixed shipment case, $\Delta_T=D_{\text{total}}/N_{\text{shipments}}$ doses are delivered $N_{\text{shipments}}$ times a year at $t_k = 365k/(N_{\text{shipments}}+1)$ days for $k=1,\ldots,N_{\text{shipments}}$, dividing the year into $N_{\text{shipments}}+1$ equal segments.
In the reactive treatment case shown in Figure~\ref{fig:impulsive-heatmap-reactive}, two months' worth of treatment doses ($\Delta_T=D_{\text{total}}/6$) are sent each time the observable incidence $\pi\varepsilon L$ crosses the reaction threshold $I_{\text{critical}}$ from below or remains above $I_{\text{critical}}$ two months after the last delivery.

Figure~\ref{fig:impulsive-heatmap-fixed} shows that delivering a fixed annual volume of treatment doses in smaller, more frequent batches achieves a much larger reduction in cumulative detectable cases than larger, less frequent shipments.
Because the outbreak starts with an empty local stockpile ($D_T(0)=0$), increasing $N_{\text{shipments}}$ shifts the first arrival $t_1 = 365/(N_{\text{shipments}}+1)$ progressively earlier during exponential take-off.
A single shipment ($N_{\text{shipments}}=1$) arrives halfway through the year, after the untreated outbreak has already passed its observable incidence peak (around day 171), and therefore produces only a modest reduction regardless of the total annual dose allocation.
By contrast, monthly shipments ($N_{\text{shipments}}=12$) begin within the first month (around day 28) and intercept transmission early, producing substantial suppression across a wide range of coverages.

For the reactive schedule (Figure~\ref{fig:impulsive-heatmap-reactive}), waiting for a high incidence threshold (top of the $y$-axis) yields very little benefit.
If authorities wait too long to trigger reactive shipments, the exponential momentum of the outbreak outpaces the discrete delivery batches.
The same pattern appears in Figure~\ref{fig:impulsive-deliveries-reactive-schedule} (Appendix~\ref{supp:more-figures}), which compares one-month (Figure~\ref{fig:impulsive-deliveries-reactive-schedule-1m}) and six-month (Figure~\ref{fig:impulsive-deliveries-reactive-schedule-6m}) reactive batch sizes together with the corresponding numbers of triggered shipments (Figures~\ref{fig:impulsive-deliveries-reactive-schedule-1m-impulses} and \ref{fig:impulsive-deliveries-reactive-schedule-6m-impulses}).
Larger reactive batches perform better when the trigger threshold is high and coverage is moderate, because a larger initial shipment absorbs the surge more effectively than a small fraction of the annual allocation; delivering the entire annual allocation in a single triggered batch (not shown) improves further at moderate coverage and high thresholds compared to six-month batches.

Across the reactive schedules in Figures~\ref{fig:impulsive-heatmap-reactive}, \ref{fig:impulsive-deliveries-reactive-schedule-1m}, and \ref{fig:impulsive-deliveries-reactive-schedule-6m}, reductions above 75\% occur only when the trigger threshold $I_{\text{critical}}$ is very low, and the benefit degrades rapidly as $I_{\text{critical}}$ increases or as individual batch sizes shrink.
Altogether, frequent fixed deliveries consistently outperform reactive schedules in total case reduction when starting from an empty stockpile, though reactive schedules often consume fewer shipments when coverage is high or when triggers are set high (Figures~\ref{fig:impulsive-deliveries-reactive-schedule-1m-impulses} and \ref{fig:impulsive-deliveries-reactive-schedule-6m-impulses}).

%%%%%%%%%%%%%%%%%%%%%%%%%%%%%%
%%%%%%%%%%%%%%%%%%%%%%%%%%%%%%
%%%%%%%%%%%%%%%%%%%%%%%%%%%%%%
%%%%%%%%%%%%%%%%%%%%%%%%%%%%%%
\subsection{Incorporation of dose cost}
\label{subsec:incorporation-of-dose-cost}

Thus far, we have treated vaccination and treatment as having the same cost.
However, in reality, the costs of vaccination ($c_V$) and treatment ($c_T$) are often very different.
For example, in the case of influenza, both interventions are direct pharmacological products (the seasonal vaccine versus a course of antivirals like oseltamivir), leading to a moderate cost asymmetry where $c_T$ is typically 1.5 to 5 times higher than $c_V$ \cite{rothberg2003management,rothberg2005vaccination}.
Conversely, for a disease like measles in low- or middle-income countries, the cost asymmetry is extreme.
Measles-rubella vaccines are highly subsidized and inexpensive to procure and deliver \cite{portnoy2020producing,sim2021costs}, while treatment requires clinical supportive care (and often hospitalization for complications), leading to estimates of $c_T$ being 10 to 100 times higher than $c_V$ \cite{debroucker2020economic,portnoy2020producing}.

To illustrate this, consider the case with unreplenished doses and a total budget $B$ allocated between vaccine and treatment doses according to
\[
B = c_V D_{V}(0) + c_T D_{T}(0).
\]

First, let us consider a similar situation to that in Figure~\ref{fig:change-AUC-I-vs-initial-doses-heatmap}, i.e., with unreplenished doses initialized at the no-intervention endemic equilibrium \eqref{eq:EE-no-intervention} ($N=100{,}000$), but varying the total budget $B$ and the physical share of treatment doses $p_T = D_T(0)/(D_V(0)+D_T(0))$.
For a given budget $B$ and unit costs $c_V$ and $c_T$, the purchased stockpiles are
\[
D_T(0) = \frac{p_T B}{p_T c_T + (1-p_T)c_V}, \qquad D_V(0) = \frac{(1-p_T)B}{p_T c_T + (1-p_T)c_V},
\]
which corresponds to a monetary treatment budget share $c_T p_T/(c_T p_T + c_V(1-p_T))$.
Figure~\ref{fig:unreplenished-budget-comparison-1yr-flu} shows an influenza-like profile with a low untreated $\R_0=1.5$ (living endemic population $N^\star \approx 98{,}400$) and a moderate cost ratio $c_T/c_V=5$ ($c_V=2$, $c_T=10$).
Figure~\ref{fig:unreplenished-budget-comparison-1yr-measles} shows a measles-like high-transmissibility stress test with $\R_0=15$ ($N^\star \approx 95{,}400$) and extreme cost asymmetry $c_T/c_V=20$ ($c_V=2$, $c_T=40$).
Note that both profiles retain all other baseline parameters of Table~\ref{tab:parameter_values} (in particular $\sigma=0.7$, $1/\omega_v=200$ days, and $\delta=5\times 10^{-4}\text{ day}^{-1}$) in order to isolate the effect of $\R_0$ and $c_T/c_V$, rather than modeling the near-complete, multi-decade immunity of real measles-rubella vaccination.

\begin{figure}[htbp]
    \centering
    \begin{subfigure}[b]{0.49\textwidth}
        \centering
        \includegraphics[width=\textwidth]{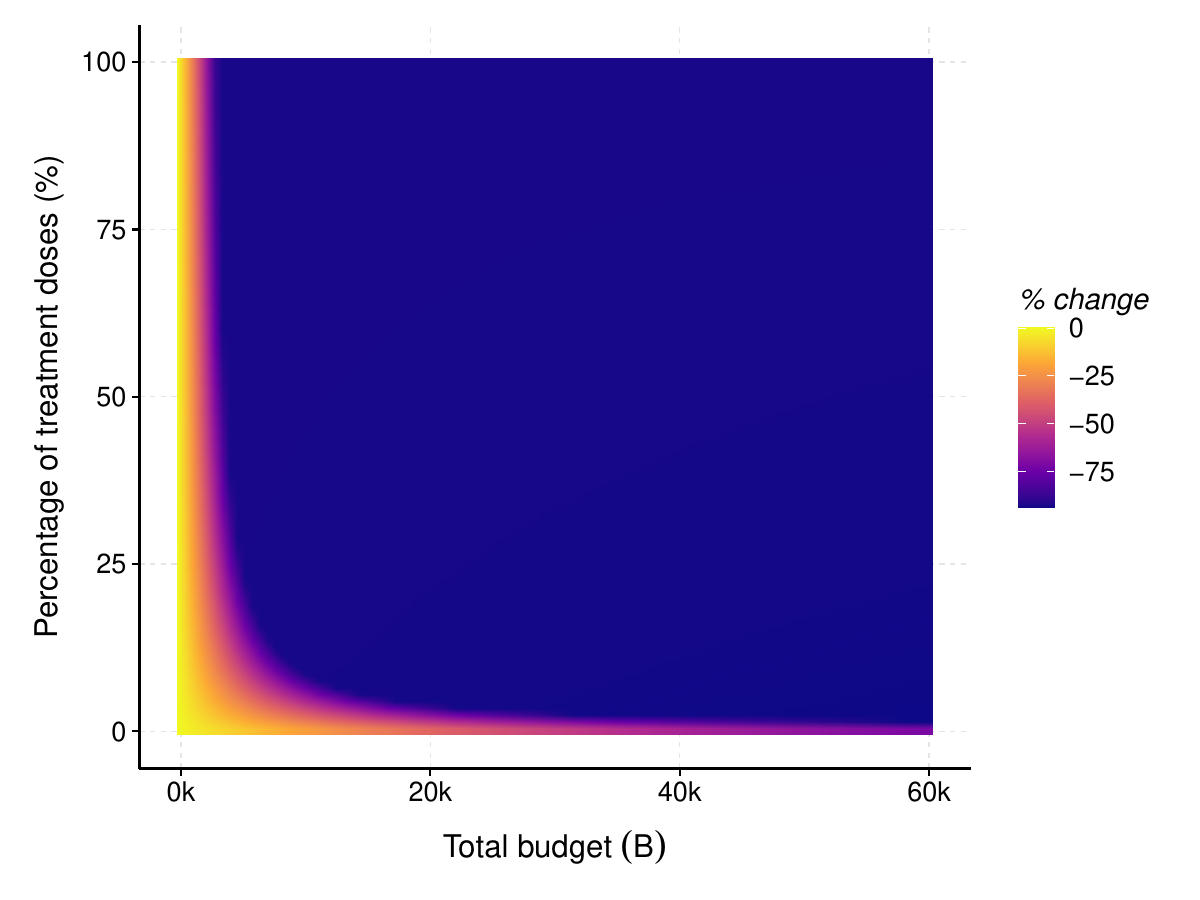}
        \caption{Influenza-like ($\R_0=1.5$, $c_V=2$, $c_T=10$)}
        \label{fig:unreplenished-budget-comparison-1yr-flu}
    \end{subfigure}
    \hfill
    \begin{subfigure}[b]{0.49\textwidth}
        \centering
        \includegraphics[width=\textwidth]{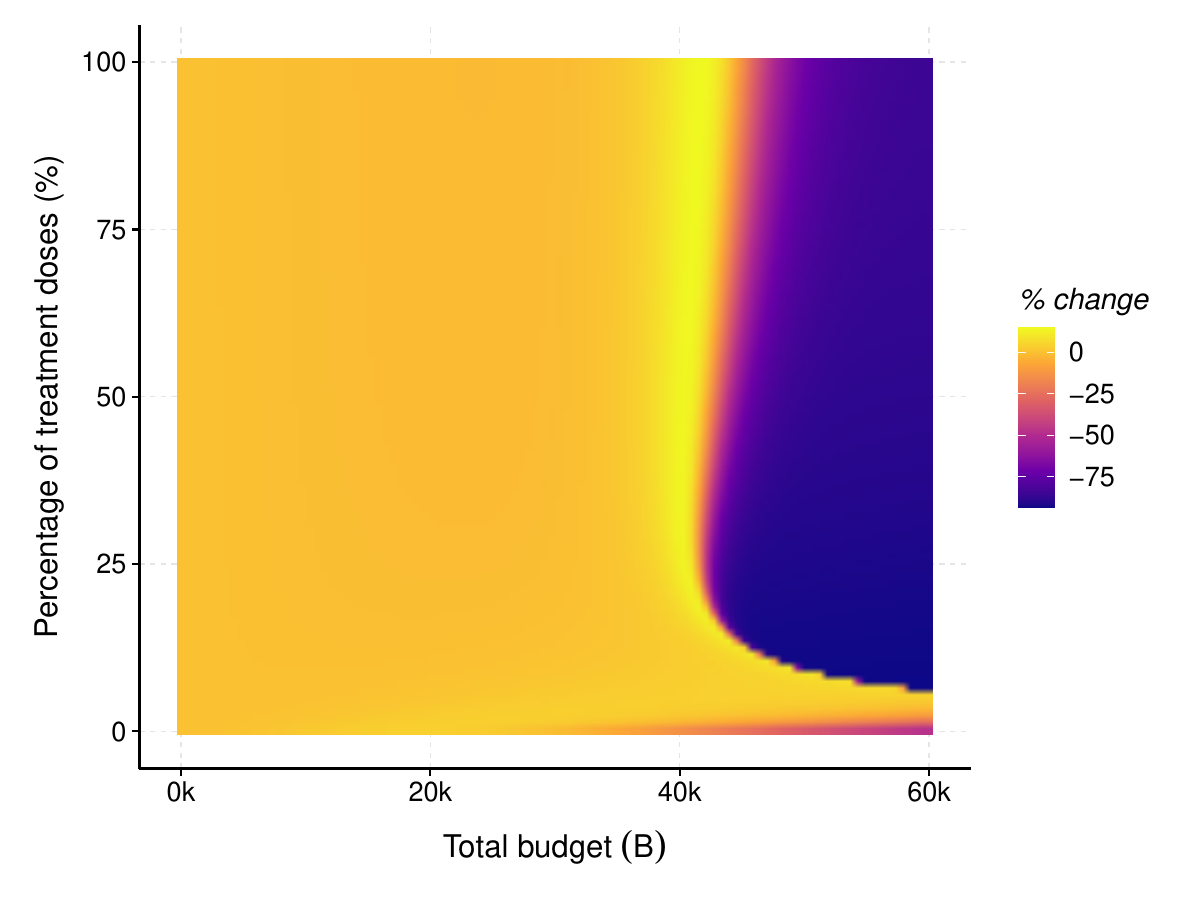}
        \caption{Measles-like ($\R_0=15$, $c_V=2$, $c_T=40$)}
        \label{fig:unreplenished-budget-comparison-1yr-measles}
    \end{subfigure}
    \caption{Effect of total budget $B$ and physical share of treatment doses $p_T = D_T(0)/(D_V(0)+D_T(0))$ purchased from that budget on the percentage change in 1-year cumulative observed incidence in an unreplenished scenario ($N=100{,}000$, $c_V=2$), for (a) an influenza-like profile ($\R_0=1.5$, $c_T/c_V=5$) and (b) a measles-like profile ($\R_0=15$, $c_T/c_V=20$).}
    \label{fig:unreplenished-budget-comparison}
\end{figure}

Comparing Figure~\ref{fig:change-AUC-I-vs-initial-doses-heatmap} and Figure~\ref{fig:unreplenished-budget-comparison-1yr-flu}, we observe that the influenza-like profile is even more favorable to treatment-heavy allocations.
Because $\R_0=1.5$ is small and $c_T/c_V=5$ is moderate, even modest budgets buy enough treatment doses to reduce cumulative observed incidence by more than 90\% across almost all positive treatment shares $p_T$, whereas pure vaccination ($p_T=0$) requires substantially larger budgets to achieve comparable reductions.
In the measles-like profile (Figure~\ref{fig:unreplenished-budget-comparison-1yr-measles}), the structure changes sharply.
Because $c_T/c_V=20$, even a small physical treatment share absorbs most of the budget (for instance, $p_T=0.2$ allocates over 80\% of $B$ to treatment).
At smaller budgets (below roughly $B=40{,}000$), treatment-heavy allocations buy too few treatment courses to cross the suppression threshold for $\R_0=15$, so vaccination-heavy allocations near $p_T=0$ perform substantially better.
In a narrow band just below that threshold (around $B=40{,}000$--$42{,}000$), premature exhaustion of treatment-heavy stocks triggers a synchronized post-exhaustion rebound oscillation that raises 1-year cumulative observed cases by up to about 15\% above the untreated equilibrium.
Once the budget exceeds roughly $45{,}000$, treatment crosses the suppression threshold, producing a sharp vertical transition to roughly 85\%--92\% incidence reduction across positive $p_T$, with mixed allocations near $p_T \approx 0.2$ achieving over 90\% reduction and outperforming both pure vaccination and pure treatment.

\begin{figure}[htbp]
    \centering
    \begin{subfigure}[b]{0.49\textwidth}
        \centering
        \includegraphics[width=\textwidth]{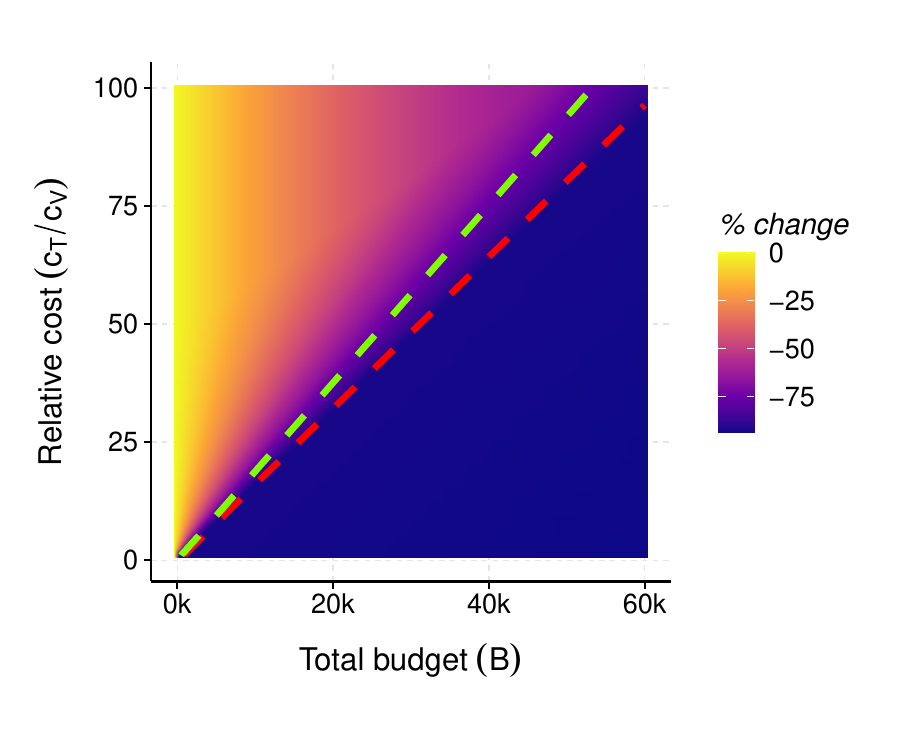}
        \caption{Influenza-like}
        \label{fig:limited-budget-1yr-flu}
    \end{subfigure}
    \hfill
    \begin{subfigure}[b]{0.49\textwidth}
        \centering
        \includegraphics[width=\textwidth]{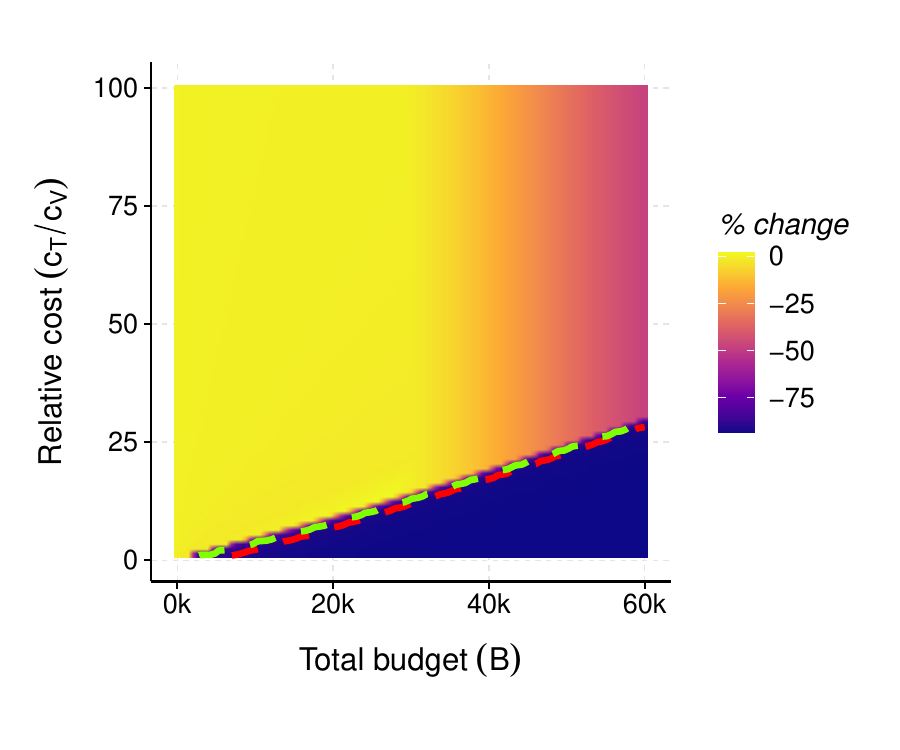}
        \caption{Measles-like}
        \label{fig:limited-budget-1yr-measles}
    \end{subfigure}
    \caption{Optimal budget allocation strategy over a 1-year horizon evaluated under the finite, unreplenished stockpile scenario ($N=100{,}000$, $c_V=2$) for two disease profiles: (a) an influenza-like low-transmissibility case ($\R_0=1.5$) and (b) a measles-like high-transmissibility case ($\R_0=15$).
    At a given total budget $B$ ($x$-axis) and relative cost $c_T/c_V$ ($y$-axis), the plot shows the largest percentage reduction in 1-year cumulative observed incidence achieved over fractional vaccination budget splits $f_V \in [0, 1]$ (with $D_V(0)=f_V B/c_V$ and $D_T(0)=(1-f_V)B/c_T$).
    The red and green curves show a 90\% and 75\% reduction, respectively.
    The corresponding optimal vaccination budget shares $f_V^\star$ are shown in Figure~\ref{fig:limited-budget-optimal-fv} (Appendix~\ref{supp:more-figures}).}
    \label{fig:limited-budget-comparison}
\end{figure}

Figure~\ref{fig:limited-budget-comparison} takes another view of the same allocation problem by optimizing directly over the vaccination budget share $f_V \in [0,1]$ across both total budget $B$ and relative cost $c_T/c_V$ (with $c_V=2$); the corresponding optimal vaccination shares $f_V^\star$ are displayed in Figure~\ref{fig:limited-budget-optimal-fv} (Appendix~\ref{supp:more-figures}).
In both panels, a sloped diagonal boundary marks where the affordable number of treatment courses $B/c_T$ reaches the threshold required for rapid suppression.
Below and to the right of this diagonal, the budget suffices to buy that critical treatment stock and suppress 1-year cumulative observed incidence by more than 90\%, with any surplus funds diverted into vaccination so that $f_V^\star$ rises smoothly from 0\% near the boundary toward 85\%--95\% at large budgets and low cost ratios (Figure~\ref{fig:limited-budget-optimal-fv}).
Above and to the left of the diagonal, $B/c_T$ falls short of rapid suppression, and the two profiles diverge.
In the influenza-like case ($\R_0=1.5$, Figures~\ref{fig:limited-budget-1yr-flu} and \ref{fig:limited-budget-optimal-fv-flu}), partial treatment still outperforms vaccination per unit cost across moderate cost ratios, creating a pure-treatment wedge ($f_V^\star=0\%$) whose radially fanned iso-reduction contours track constant $B/c_T$; only at higher cost ratios in the upper-left corner does the optimum switch to pure vaccination ($f_V^\star=100\%$), where contours become vertical.
In the measles-like case ($\R_0=15$, Figures~\ref{fig:limited-budget-1yr-measles} and \ref{fig:limited-budget-optimal-fv-measles}), sub-threshold treatment is rapidly overwhelmed and the suppression diagonal lies much lower in $c_T/c_V$: above this diagonal, budgets below roughly $B=33{,}000$ achieve less than a 5\% reduction under any split, whereas larger budgets switch directly to pure vaccination ($f_V^\star=100\%$), producing vertical contours that depend on $B$ alone.

%%%%%%%%%%%%%%%%%%%%%%%%%%%%%%
%%%%%%%%%%%%%%%%%%%%%%%%%%%%%%
%%%%%%%%%%%%%%%%%%%%%%%%%%%%%%
%%%%%%%%%%%%%%%%%%%%%%%%%%%%%%
\section{Discussion}
\label{sec:discussion}

In the case of unreplenished doses, we highlighted that the deterministic ODE gives a pessimistic view of the situation, since the model's long-term behavior follows that of the equivalent model with no intervention.
Thus, in particular, the disease always rebounds when $\R_0>1$.
However, consideration of both transient dynamics and of a corresponding continuous-time Markov chain highlights that a sufficiently large early intervention can push the infection close to extinction, and, in the stochastic setting, even to an absorbing extinction state before stocks deplete.
The discrepancy between the mathematical results and what actually happens during an outbreak is further illustrated when we consider unlimited doses scenarios.
There, we observe (Figure~\ref{fig:coverage-vs-pcttreatment-unlimited-doses}) that regions with a reproduction number larger than one still see a large reduction in the number of cases observed during a one-year period, which is again confirmed by using the CTMC analog of the model.

Thus, if all required doses are available at the start of the epidemic, substantial 1-year reductions and stochastic extinction remain achievable even when doses are not replenished.
However, this assumption is somewhat unrealistic when considering countries with low gross domestic product (GDP) or emergency situations where vaccine and treatment procurement is a challenge and keeping a large stockpile of doses may lead to excessive wastage.
To model this situation, we considered (computationally) two impulsive delivery schemes.
In the fixed schedule scheme, doses are delivered at a certain number of fixed time points during the year, while in the reactive scheme, doses are delivered when the incidence reaches a certain threshold.
When considering fixed versus reactive schemes, we observed that reactive schemes might save resources in a mild season, but they carry a severe penalty if the exponential growth of a fast-spreading outbreak escapes containment before the threshold is triggered.
Note that real-world logistical decisions are triggered by delayed reporting metrics, since time elapses before cases are observed and shipments are authorized.
We used an instantaneous trigger here to preserve the well-posedness and tractability of the ODE framework, leaving an impulsive delay differential equation (DDE) extension for future work; in practice, during exponential growth, a rough estimate of the reproduction number and recent incidence can also provide a practical approximation of current incidence without introducing explicit delays.

Finally, we concluded (in Section~\ref{subsec:incorporation-of-dose-cost}) our computational exploration of the model by incorporating vaccination and treatment costs.
We observed that when facing extreme cost asymmetry with a highly transmissible infection as, e.g., in our measles-like profile, pure vaccination is optimal whenever the budget cannot afford enough costly treatment courses to cross the rapid-suppression threshold, whereas larger budgets or lower cost ratios favor a mixed or treatment-led portfolio.

This is preliminary work and more effort is needed to fully explore the dynamics of dose-limited interventions.
We have barely considered, for instance, the consequences of non-constant dose delivery, having only considered computationally the case of impulsive dose deliveries.
We noted in Section~\ref{sec:analysis-full-model-at-EE} that the full model may exhibit bi- or multistability through endogenous stock-consumption feedback, depending on the nature of the dose replenishment functions $a_V$ and $a_T$ and the dose consumption functions $p$, $v$, and $g$.
Although the impulsive delivery scenario explored numerically in Section~\ref{subsec:numerics-impulsive-doses} is realistic in some contexts, it is not always the case that interventions are impulsive, and exploring the full autonomous dynamics of \eqref{sys:general-form} remains a challenge for future work.

Our mechanistic stockpile formulation contrasts with the existing literature on resource-constrained epidemic control.
In standard optimal control models \cite{hansen2011optimal,lee2010optimal,lenhart2007optimal,sharomi2017optimal}, intervention rates are treated as external control inputs chosen by solving a boundary-value problem over a finite horizon $[0,T]$ subject to control-cost penalties or an integral budget constraint.
While that approach identifies the theoretical best-case allocation if a health authority has complete foresight of the outbreak trajectory and can modulate administration rates at will, it does not describe what happens when doses are simply deployed upon demand until supplies run out.
Conversely, autonomous models with saturating treatment functions \cite{wang2006backward,zhou2012dynamics} capture instantaneous delivery capacity ceilings (such as a fixed number of hospital beds or clinicians) but assume that the underlying consumable supplies can be sustained indefinitely at that ceiling.
Our formulation instead aligns with the stock-aware perspective emphasized by Andrade et al.~\cite{andrade2024preparing}: by embedding the available vaccine and treatment stockpiles $D_V(t)$ and $D_T(t)$ directly as state variables governed by acquisition, administration, and spoilage ($\kappa_V, \kappa_T$), the intervention rates $p(D_V)$, $v(D_V)$, and $g(D_T)$ respond endogenously to inventory depletion.
In practice, this distinction matters for public health planning in low-resource settings: it reveals how premature stockpile exhaustion can trigger post-exhaustion epidemic rebounds in deterministic settings, why early aggressive deployment can nevertheless achieve stochastic extinction before stocks run dry, and why regular smaller shipments outperform large infrequent deliveries when an outbreak begins with an empty local stockpile by bringing the first batch earlier during exponential take-off.
This mechanistic baseline also provides a foundation upon which future optimal procurement and shipment-scheduling problems can be formulated.

Finally, note that our model uses \emph{per capita} rates $v(D_V)S$ and $g(D_T)I$ for the administration of interventions.
This is not an issue since the total population is bounded.
However, combining our consumable stockpile dynamics with saturating healthcare delivery capacities \cite{andrade2024preparing,wang2006backward,zhou2012dynamics} to simultaneously capture both inventory runouts and clinical administration bottlenecks would be an interesting avenue for future modeling efforts.

%%%%%%%%%%%%%%%%%%%%%
%%%%%%%%%%%%%%%%%%%%%
%%%%%%%%%%%%%%%%%%%%%

\section*{Acknowledgments}
GenAI (Gemini 4 Neon and Carbon) was used to polish the text, check the mathematical content, debug and comment the computational code, and assist with the figures.
After using these tools, the authors reviewed and verified all content and take full responsibility for the manuscript.

\section*{Authorship and Contributorship Statement}
All authors have made substantial intellectual contributions to the study conception, execution, and design of the work.
All authors have read and approved the final manuscript.

\section*{Conflicts of Interest}
The authors declare no competing interests.

\section*{Data \& Code Availability}
All data and code used in this study will be made publicly available at \url{https://github.com/julien-arino/dose-limited-interventions}.

\section*{Funding Statement}
The work of JA is supported in part by Natural Sciences and Engineering Research Council of Canada (NSERC) Discovery Grant RGPIN-2024-05436.

%%%%%%%%%%%%%%%%%%%%%
%%%%%%%%%%%%%%%%%%%%%

\clearpage
\appendix

\section{Mathematical background and further reading}
\label{app:non-math-summary}

\textbf{Explicit stockpile compartments and saturating administration.}
Classical compartmental epidemic models represent vaccination and treatment using fixed \emph{per capita} rates, which implicitly assumes an inexhaustible supply chain where medical countermeasures are always available on demand.
To capture supply constraints in low-resource or emergency settings, our model couples standard epidemiological compartments (susceptible, latently infected, symptomatically and asymptomatically infectious, recovered, and vaccinated individuals) with two explicit inventory variables tracking the available vaccine doses and treatment courses.
Each administered dose deducts from the corresponding inventory, while cold-chain spoilage removes usable stock and new deliveries replenish it.
As an inventory approaches zero, the administration rate smoothly drops to zero, naturally representing the transition from active intervention to a supply stockout.

\textbf{Threshold quantities and long-term asymptotic regimes.}
Disease persistence is governed by the basic reproduction number (the average number of secondary infections generated by a single case in an entirely susceptible population without intervention) and the intervention reproduction number (the corresponding quantity when vaccination and treatment operate at a steady supply level).
Asymptotic analysis of the differential equations reveals two limiting regimes.
When a finite initial stockpile is never replenished, consumption and spoilage inevitably exhaust the inventory, causing the long-term dynamics to revert to those of an uncontrolled outbreak.
Conversely, when replenishment is sufficient to prevent stockouts, the threshold behavior matches the classical unlimited-intervention model, yielding local elimination when $\R_v^{\eqref{sys:general-form}} < 1$ and uniform persistence of infection when $\R_v^{\eqref{sys:general-form}} > 1$ (Theorem~\ref{th:full-model-DFE-and-persistence}).

\textbf{Deterministic rebounds versus stochastic extinction.}
Because ordinary differential equations treat populations as continuous quantities, the disease-free set (zero infected individuals) cannot be reached in finite time from an active outbreak.
Even when an aggressive campaign drives prevalence down to a tiny fraction of an individual (a mathematical artifact known as the ``atto-fox'' problem), the infection cannot reach exact zero and inevitably rebounds once an unreplenished stockpile runs out when $\R_0 > 1$.
To account for whole-person discreteness at low prevalence, we pair the differential equations with a continuous-time Markov chain (stochastic) formulation.
In the stochastic regime, driving case counts near zero frequently allows the last infectious individual to recover before transmitting the pathogen, achieving permanent local extinction (an absorbing state) before supplies are exhausted.

\begin{reading}
The following references provide an accessible introduction to the mathematical and epidemiological topics underlying this work:
\begin{itemize}
    \item \textbf{Compartmental epidemic modeling:} foundational theory, reproduction numbers, and endemic thresholds in Hethcote~\cite{hethcote2000mathematics} and Brauer, Castillo-Chavez, and Feng~\cite{brauer2019mathematical}.
    \item \textbf{Stochastic epidemic processes and extinction:} continuous-time Markov chain epidemic models and extinction probabilities in Allen~\cite{allen2008introduction}; the deterministic ``atto-fox'' phenomenon in Mollison~\cite{mollison1991dependence}.
    \item \textbf{Optimal control and resource constraints:} optimal control theory for biological systems in Lenhart and Workman~\cite{lenhart2007optimal}; health-system resource constraints in infectious disease models in Bozzani et al.~\cite{bozzani2021building}.
    \item \textbf{Stock-aware epidemic logistics:} finite resource stock dynamics in Andrade et al.~\cite{andrade2024preparing}; cold-chain and conflict-related delivery barriers in Ashok, Brison, and LeTallec~\cite{ashok2017improving} and Sabahelzain et al.~\cite{sabahelzain2025implications}.
\end{itemize}
\end{reading}

%%%%%%%%%%%%%%%%%%%%%
%%%%%%%%%%%%%%%%%%%%%
%%%%%%%%%%%%%%%%%%%%%

\clearpage
\begin{center}
    \Large\bfseries Supplementary Material
\end{center}
\medskip

These Supplementary Materials complement the main text.
They contain all the proofs in the paper, as well as most ``mathematical discussions,'' although some material remains in the text if it is useful to understand the model.
The structure of this document follows roughly that of the main text, with complements on the model, the mathematical analysis, and finally results.

%%%%%%%%%%%%%%%%%%%%%
%%%%%%%%%%%%%%%%%%%%%
%%%%%%%%%%%%%%%%%%%%%
\section{The general continuous-time Markov chain corresponding to \texorpdfstring{\eqref{sys:general-form}}{(2.1)}}
\label{supp:CTMC}
For computational work on \eqref{sys:general-form}, we also use a continuous-time Markov chain analog of this model.
The CTMC uses transitions in Table~\ref{tab:general-CTMC-transitions}.

\begin{table}[htbp]
	\centering
	\footnotesize
	\begin{tabular}{@{}lll@{}}
		\toprule
		Event & Transition   & Transition rates \\
		\midrule
		Recruitment without vaccination & $\to(S+1)$ & $(1-p(D_V))b$ \\
		Recruitment with vaccination & $\to(V+1, D_V-1)$ & $p(D_V)b$ \\
		\midrule
		Natural death of an $S$ & $\to(S-1)$ & $d S$ \\
		Natural death of an $L$ & $\to(L-1)$ & $d L$ \\
		Natural death of an $I$ & $\to(I-1)$ & $d I$ \\
		Natural death of an $A$ & $\to(A-1)$ & $d A$ \\
		Natural death of an $R$ & $\to(R-1)$ & $d R$ \\
		Natural death of a $V$ & $\to(V-1)$ & $d V$ \\
		\midrule
		Infection $S$ by $I$ &  $\to(S-1,L+1)$ & $\beta S I$ \\
		Infection $S$ by $A$ &  $\to(S-1,L+1)$ & $\eta\beta S A$ \\
		Infection $V$ by $I$ &  $\to(V-1,L+1)$ & $(1-\sigma)\beta V I$ \\
		Infection $V$ by $A$ &  $\to(V-1,L+1)$ & $(1-\sigma)\eta\beta V A$ \\
		Latent to symptomatic & $\to(L-1,I+1)$ & $\pi\varepsilon L$ \\
		Latent to asymptomatic & $\to(L-1,A+1)$ & $(1-\pi)\varepsilon L$ \\
		Disease-induced death in $I$ & $\to(I-1)$ & $\delta I$ \\
		Recovery in $I$ & $\to(I-1,R+1)$ & $\gamma I$ \\
		Recovery in $A$ & $\to(A-1,R+1)$ & $\gamma A$ \\
		End of vaccine protection & $\to(V-1,S+1)$ & $\omega_v V$ \\
		End of disease-induced immunity & $\to(R-1,S+1)$ & $\omega_r R$ \\
		\midrule
		Vaccination within the population & $\to(S-1,V+1,D_V-1)$ & $v(D_V)S$ \\
		Treatment of a symptomatically infectious individual & $\to(I-1,R+1,D_T-1)$ & $g(D_T)I$ \\
		\midrule
		Purchase of a vaccine dose & $\to(D_V+1)$ & $a_V(D_V)$ \\
		Expiration of a vaccine dose & $\to(D_V-1)$ & $\kappa_VD_V$ \\
		Purchase of a treatment dose & $\to(D_T+1)$ & $a_T(D_T)$ \\
		Expiration of a treatment dose & $\to(D_T-1)$ & $\kappa_TD_T$ \\
		\bottomrule
	\end{tabular}
	\caption{Transition rates for the CTMC model corresponding to \eqref{sys:general-form}.
	Only states that change through the transition are shown.}
	\label{tab:general-CTMC-transitions}
\end{table}

%%%%%%%%%%%%%%%%%%%%%
%%%%%%%%%%%%%%%%%%%%%
%%%%%%%%%%%%%%%%%%%%%
%%%%%%%%%%%%%%%%%%%%%
\section{Mathematical analysis}
\label{supp:math-analysis}

To begin, it is clear that by the assumption that $a_V$, $a_T$, $p$, $v$, and $g$ be $C^1$, solutions to \eqref{sys:general-form} exist and are unique.
Invariance of $\IR_+^8$ under the flow of \eqref{sys:general-form} also follows from the assumption on $p$, $v$, and $g$ near or at 0.
It is also simple to show that the set
\begin{equation}\label{eq:Omega-full-system}
	\Omega_{L=I=A=0}^\eqref{sys:general-form} = \left\{\bX \in \mathbb{R}_+^8 \mid L = I = A = 0 \right\}
\end{equation}
is positively invariant under the flow of \eqref{sys:general-form}, where we have denoted
\[
\bX = (S, L, I, A, R, V, D_V, D_T)\in\IR_+^8.
\]

%%%%%%%%%%%%%%%%%%%%%
%%%%%%%%%%%%%%%%%%%%%
\subsection{The model without intervention}
\label{supp:math-analysis-proof-no-intervention}
To use the next generation matrix method of \cite{van2002reproduction}, we write the vectors of new infections and other flows in the infected compartments,
\[
\F = \begin{pmatrix}
	\beta S(I+\eta A) \\ 0\\ 0
\end{pmatrix}
\text{ and }
\W = \begin{pmatrix}
	(\varepsilon+d)L \\ -\pi\varepsilon L+(\gamma+\delta +d)I \\ -(1-\pi)\varepsilon L + (\gamma+d)A
\end{pmatrix}.
\]
It follows that the next generation matrix takes the form
\begin{equation}\label{eq:FW-inv}
	FW^{-1}=
	%\frac{1}{(\varepsilon+d)(\delta+d)}
	\frac{b}{d}
	\begin{pmatrix}
		0 & \beta & \beta\eta \\ 0 & 0 & 0 \\ 0 & 0 & 0
	\end{pmatrix}
	\begin{pmatrix}
		\varepsilon+d & 0 & 0 \\ -\pi\varepsilon & \gamma+\delta +d & 0 \\ -(1-\pi)\varepsilon & 0 & \gamma+d
	\end{pmatrix}^{-1}
\end{equation}
and thus the basic reproduction number without vaccination takes the form \eqref{eq:R0-no-vaccination} in the main text.

\begin{proof}[Proof of Proposition~\ref{prop:isolated-novacc-GAS}]
Because the total population $N=S+L+I+A+R$ satisfies $N'=b-dN-\delta I \leq b-dN$, we have $\limsup_{t\to\infty} N(t) \leq b/d$.
It follows that the compact set
\[
\Omega = \left\{(S,L,I,A,R)\in\IR^5_+ \mid S+L+I+A+R\leq \frac bd\right\}
\]
is positively invariant under the flow of \eqref{sys:no-intervention} and attracts all solutions in $\IR_+^5$.

\emph{Biological relevance of $\bE_{\star}^{\eqref{sys:no-intervention}}$.}
This depends on the sign of $I^\star$ (noting that $b, d, \beta, \varepsilon, \gamma > 0$ and $\pi \in (0, 1]$ ensure all denominators in \eqref{eq:EE-no-intervention} are strictly positive).
Let us show that
\[
\frac{(\varepsilon+d)(\gamma+\delta+d)}{\pi\varepsilon} - \frac{\omega_r \gamma \Theta}{\omega_r+d} > 0.
\]
Because $\omega_r/(\omega_r+d) < 1$, the subtracted term is bounded above by $\gamma \Theta$.
Expanding $\gamma \Theta$ yields
\[
\gamma \Theta = \gamma \left( 1 + \frac{(1-\pi)(\gamma+\delta+d)}{\pi(\gamma+d)} \right).
\]
Since $\gamma/(\gamma+d) < 1$, we have $\gamma \Theta \leq \gamma + (1-\pi)(\gamma+\delta+d)/\pi$.
Conversely, the first term satisfies
\begin{multline*}
\frac{(\varepsilon+d)(\gamma+\delta+d)}{\pi\varepsilon} = \frac{\varepsilon+d}{\varepsilon} \frac{\gamma+\delta+d}{\pi} > \\
\frac{\pi(\gamma+\delta+d) + (1-\pi)(\gamma+\delta+d)}{\pi} = \gamma+\delta+d + \frac{(1-\pi)(\gamma+\delta+d)}{\pi}.
\end{multline*}
This is strictly greater than $\gamma + (1-\pi)(\gamma+\delta+d)/\pi$, confirming that the denominator of $I^\star$ is strictly positive.
Thus, the sign of $(\R_0-1)/\R_0$ determines the sign of $I^\star$ and the biological relevance of $\bE_{\star}^{\eqref{sys:no-intervention}}$.

\emph{Case $\R_0 < 1$.}
On the disease-free subspace $\{L=I=A=0\}$, system \eqref{sys:no-intervention} reduces to $S' = b + \omega_r R - dS$ and $R' = -(\omega_r+d)R$, whose Jacobian has strictly negative eigenvalues $-d$ and $-(\omega_r+d)$, so hypothesis (A5) of \cite{van2002reproduction} holds (and (A1)--(A4) hold directly).
By Theorem~2 of \cite{van2002reproduction}, the disease-free equilibrium $\bE_{0}^{\eqref{sys:no-intervention}}$ is locally asymptotically stable when $\R_0 < 1$ and unstable when $\R_0 > 1$.

To establish global attractivity when $\R_0 < 1$, we adapt the left eigenvector method of \cite{shuai2013global}.
The left eigenvector $v$ of the next-generation matrix $FW^{-1}$ defined in \eqref{eq:FW-inv} corresponding to the dominant eigenvalue $\R_0$ takes the form
\[
v = \left(1,\; \frac{\beta S_0}{\R_0(\gamma+\delta+d)},\; \frac{\beta \eta S_0}{\R_0(\gamma+d)} \right).
\]
Omitting the division by $\R_0$ in the components of $v$ (equivalently, using coefficients derived from the first row of $FW^{-1}$) shifts the factor $\R_0-1$ onto the $L$ compartment.
We therefore define the Lyapunov function candidate
\[
\L
=
L
+
\frac{\beta S_0}{\gamma+\delta+d} I
+
\frac{\beta \eta S_0}{\gamma+d} A,
\]
where $S_0 = b/d$.
Contrary to the model studied in \cite{lamichhane2015global}, adding a term in $R$ to $\L$ is unnecessary for the disease-free equilibrium.

Differentiating $\L$ along solutions of \eqref{sys:no-intervention}, we obtain
\[
\begin{aligned}
\L'
=\;&
\beta S(I+\eta A) - (\varepsilon+d)L \\
&+ \frac{\beta S_0}{\gamma+\delta+d} \left( \pi\varepsilon L - (\gamma+\delta+d)I \right)
+ \frac{\beta \eta S_0}{\gamma+d} \left( (1-\pi)\varepsilon L - (\gamma+d)A \right).
\end{aligned}
\]
Rearranging terms yields
\[
\begin{aligned}
\L'
=\;&
L \left( -(\varepsilon+d) + \frac{\beta S_0 \pi \varepsilon}{\gamma+\delta+d} + \frac{\beta \eta S_0 (1-\pi)\varepsilon}{\gamma+d} \right) \\
&+ I \left( \beta S - \beta S_0 \right) + A \left( \beta \eta S - \beta \eta S_0 \right).
\end{aligned}
\]
Recognizing the expression for $\R_0$ from \eqref{eq:R0-no-vaccination} in the coefficient of $L$, this simplifies to
\[
\L' = (\varepsilon+d)(\R_0-1)L + \beta(S-S_0) (I+\eta A).
\]
Since $S \leq S_0$ within $\Omega$ and $\R_0 < 1$, we have $\L' \leq 0$ on $\Omega$, and $\L' = 0$ in $\Omega$ implies $L = 0$ (because $(\varepsilon+d)(\R_0-1) < 0$ and $\beta(S-S_0)(I+\eta A) \leq 0$).
Let $\mathcal{I}$ be the largest compact invariant subset of $\{\L' = 0\} \cap \Omega$ under \eqref{sys:no-intervention}.
Along any full trajectory $(S(t), L(t), I(t), A(t), R(t))$ contained in $\mathcal{I}$ for all $t \in \IR$, we have $L(t) \equiv 0$.
Substituting $L \equiv 0$ into \eqref{sys:no-vacc-dI} and \eqref{sys:no-vacc-dA} yields $I(t) = I(0)e^{-(\gamma+\delta+d)t}$ and $A(t) = A(0)e^{-(\gamma+d)t}$; since $\mathcal{I} \subset \Omega$ is bounded, these solutions remain bounded as $t \to -\infty$ only if $I(t) \equiv 0$ and $A(t) \equiv 0$.
Equations \eqref{sys:no-vacc-dR} and \eqref{sys:no-vacc-dS} then reduce to $R' = -(\omega_r+d)R$ and $S' = b - dS$, whose boundedness as $t \to -\infty$ likewise forces $R(t) \equiv 0$ and $S(t) \equiv b/d = S_0$.
Hence $\mathcal{I} = \{\bE_{0}^{\eqref{sys:no-intervention}}\}$, and LaSalle's invariance principle implies that every solution starting in $\Omega$ converges to $\bE_{0}^{\eqref{sys:no-intervention}}$.

To extend global attractivity from $\Omega$ to all of $\IR_+^5$ (including solutions with $N(0) > b/d$, which may approach $\Omega$ only asymptotically), consider any solution $\by(t)$ of \eqref{sys:no-intervention} in $\IR_+^5$.
Because $\limsup_{t\to\infty} N(t) \leq b/d$, the trajectory $\by(t)$ is bounded and its $\omega$-limit set $\omega(\by(0))$ is nonempty, compact, invariant under \eqref{sys:no-intervention}, and contained in $\Omega$.
Let $\by^\star \in \omega(\by(0))$ be a point at which the continuous function $\L$ attains its maximum on $\omega(\by(0))$.
Because $\omega(\by(0)) \subset \Omega$ is invariant and $\L' \leq 0$ on $\Omega$, the backward trajectory through $\by^\star$ in $\omega(\by(0))$ satisfies $\L \equiv \L(\by^\star)$ and hence $\L' \equiv 0$ for all $t \leq 0$.
The same boundedness argument on $(-\infty, 0]$ forces $L = I = A = 0$ at $\by^\star$, so $\L(\by^\star) = 0$.
Consequently, $\L \equiv 0$ on $\omega(\by(0))$; since $\L \geq L \geq 0$, it follows that $L \equiv 0$ on $\omega(\by(0))$ and thus $\omega(\by(0)) \subset \mathcal{I} = \{\bE_{0}^{\eqref{sys:no-intervention}}\}$.
Combining global attractivity in $\IR_+^5$ with the local asymptotic stability of $\bE_{0}^{\eqref{sys:no-intervention}}$ proves that $\bE_{0}^{\eqref{sys:no-intervention}}$ is globally asymptotically stable in $\IR_+^5$.

\emph{Case $\R_0 > 1$.}
When $\R_0 > 1$, the disease-free equilibrium $\bE_{0}^{\eqref{sys:no-intervention}}$ is unstable by Theorem~2 of \cite{van2002reproduction}, and the endemic equilibrium \eqref{eq:EE-no-intervention} is biologically relevant, meaning all its components are nonnegative with $I^\star > 0$.
\end{proof}

\begin{remark}
Contrary to the SLIAR model in \cite{lamichhane2015global}, the inclusion of waning immunity (the flow from $R$ back to $S$ at rate $\omega_r$) prevents the direct use of standard Volterra-type Lyapunov functions to establish global asymptotic stability for the endemic equilibrium.
The derivative of the susceptible component $\mathcal{L}_S = S-S^\star-S^\star\ln(S/S^\star)$ along trajectories generates cross terms, such as $\omega_r R(1 - S^\star/S)$, which cannot be systematically balanced or canceled by terms arising from the infected and recovered compartments.
Consequently, it is generally not possible to combine these into a globally negative semi-definite derivative using standard coefficients.
Note also that we could use the same reasoning as in Theorem~\ref{th:full-model-DFE-and-persistence} to show that the disease is uniformly persistent when $\R_0>1$.
\end{remark}
%%%%%%%%%%%%%%%%%%%%%%%%%%%%%%%%%%%%%%%%%%%%%%%%%%%%%%%%%%%%%%%%%%%%%%%%%%%%%%%%%%%%%%%%%%%

%%%%%%%%%%%%%%%%%%%%%%%%%%%%%%%
%%%%%%%%%%%%%%%%%%%%%%%%%%%%%%%
\subsection{Unreplenished doses lead to no-intervention scenario}
\label{supp:analysis-nodose-novacc}
\begin{proof}[Proof of Lemma~\ref{lemma:nodose-novacc}]
Under the assumptions $a_V(D_V)\equiv 0$, $a_T(D_T)\equiv 0$, and $\kappa_V,\kappa_T>0$, equations \eqref{sys:general-form-dD} and \eqref{sys:general-form-dT} satisfy
\begin{align*}
D_V' &= - p(D_V) b - v(D_V) S - \kappa_V D_V \leq -\kappa_V D_V, \\
D_T' &= - g(D_T) I - \kappa_T D_T \leq -\kappa_T D_T.
\end{align*}
Hence $D_V(t)\leq D_V(0)e^{-\kappa_V t}\to 0$ and $D_T(t)\leq D_T(0)e^{-\kappa_T t}\to 0$ exponentially as $t\to\infty$.
The total population $N=S+L+I+A+R+V$ satisfies $N'=b-dN-\delta I \leq b-dN$, so $\limsup_{t\to\infty} N(t) \leq b/d$ and every solution in $\IR_+^8$ is bounded for $t\geq 0$.
Because $S(t)$ is bounded and $p(0)=v(0)=0$ by continuity at $0$, we have $\psi(t) := p(D_V(t))b + v(D_V(t))S(t) \to 0$ as $t\to\infty$.
Equation \eqref{sys:general-form-dV} then gives
\[
V'(t) \leq \psi(t) - (\omega_v + d)V(t).
\]
For any $\xi > 0$, choosing $T_\xi \geq 0$ such that $\psi(t) \leq \xi(\omega_v+d)$ for all $t \geq T_\xi$ implies $\limsup_{t\to\infty} V(t) \leq \xi$; since $\xi > 0$ is arbitrary, $V(t)\to 0$ as $t\to\infty$.
Consequently, for any solution $\bX(t)$ of \eqref{sys:general-form} in $\IR_+^8$, the $\omega$-limit set $\omega(\bX(0))$ is nonempty, compact, invariant under \eqref{sys:general-form}, and contained in the positively invariant subspace
\[
\mathcal{M}_0 := \left\{\bX \in \IR_+^8 \mid V = D_V = D_T = 0,\; S+L+I+A+R \leq \frac{b}{d}\right\},
\]
on which \eqref{sys:general-form} reduces identically to the no-intervention system \eqref{sys:no-intervention}.
In particular, the $(S, L, I, A, R)$-projection $\Pi(\omega(\bX(0)))$ is a nonempty, compact, invariant subset of $\Omega$ for \eqref{sys:no-intervention}.

To determine the local stability of $\bE_{0}^{\eqref{sys:general-form}} = (b/d, 0, 0, 0, 0, 0, 0, 0)$, consider the Jacobian matrix of \eqref{sys:general-form} at $\bE_{0}^{\eqref{sys:general-form}}$ in the coordinate ordering $(S, L, I, A, R, V, D_V, D_T)$.
Because $p(0)=v(0)=g(0)=0$ and $L=I=A=V=D_V=D_T=0$ at $\bE_{0}^{\eqref{sys:general-form}}$, all partial derivatives of \eqref{sys:general-form-dV}--\eqref{sys:general-form-dT} with respect to $(S, L, I, A, R)$ vanish.
The Jacobian is therefore block upper-triangular: its upper-left $5\times 5$ block coincides with the Jacobian of \eqref{sys:no-intervention} at $\bE_{0}^{\eqref{sys:no-intervention}}$, while its lower-right $3\times 3$ block for $(V, D_V, D_T)$ is upper-triangular with strictly negative diagonal entries $-(\omega_v+d)$, $-(p'(0)b + v'(0)b/d + \kappa_V) \leq -\kappa_V$, and $-\kappa_T$.

If $\R_0 < 1$, all five eigenvalues of the upper-left $5\times 5$ block have strictly negative real parts (by Theorem~2 of \cite{van2002reproduction}, as in Proposition~\ref{prop:isolated-novacc-GAS}), so all eight eigenvalues of the Jacobian at $\bE_{0}^{\eqref{sys:general-form}}$ have strictly negative real parts and $\bE_{0}^{\eqref{sys:general-form}}$ is locally asymptotically stable in $\IR_+^8$.
By the proof of Proposition~\ref{prop:isolated-novacc-GAS}, the only nonempty compact invariant subset of $\Omega$ for \eqref{sys:no-intervention} when $\R_0 < 1$ is $\{\bE_{0}^{\eqref{sys:no-intervention}}\}$, whence $\Pi(\omega(\bX(0))) = \{\bE_{0}^{\eqref{sys:no-intervention}}\}$ and therefore $\omega(\bX(0)) = \{\bE_{0}^{\eqref{sys:general-form}}\}$ for every $\bX(0)\in\IR_+^8$.
Combining global attractivity with local asymptotic stability proves that $\bE_{0}^{\eqref{sys:general-form}}$ is globally asymptotically stable in $\IR_+^8$, while $\bE_{b}^{\eqref{sys:general-form}}$ is not biologically relevant by Proposition~\ref{prop:isolated-novacc-GAS}.

If $\R_0 > 1$, the upper-left $5\times 5$ block has an eigenvalue with strictly positive real part (again by Theorem~2 of \cite{van2002reproduction}), so $\bE_{0}^{\eqref{sys:general-form}}$ is unstable in $\IR_+^8$, while $\bE_{b}^{\eqref{sys:general-form}} = (\bE_{\star}^{\eqref{sys:no-intervention}}, 0, 0, 0)$ is biologically relevant by Proposition~\ref{prop:isolated-novacc-GAS}.
\end{proof}

%%%%%%%%%%%%%%%%%%%%%
%%%%%%%%%%%%%%%%%%%%%
\subsection{Existence of a backward bifurcation}
\label{supp:backward-bifurcation}

As mentioned in the text, the case without limitation \eqref{sys:unlimited-vacc} is almost identical to the model in \cite{ArinoMilliken2022a}.
As a consequence, \eqref{sys:unlimited-vacc} can undergo a backward bifurcation under the conditions below, with the following results from \cite{ArinoMilliken2022a} adapted to \eqref{sys:unlimited-vacc} and given here without proof.

\begin{proposition}
	\label{prop:pos_eqb}
	Assume $\sigma \in [0, 1)$, and let
	\begin{align*}
		a_0 &=d(\varepsilon+d)(v^\star+\omega_v+d)(\R_v^{\eqref{sys:unlimited-vacc}}-1), \\
		a_1 &= \lambda^2(1-\sigma)b - \lambda(1-\sigma)d(\varepsilon+d) \\
		&\quad + \lambda((1-\sigma)v^\star + \omega_v+d)\left(\frac{\omega_r\varepsilon \left( \gamma(\gamma+g^\star+\delta+d) + \pi(g^\star d - \gamma\delta) \right)}{(\omega_r+d)(\gamma+d)(\gamma+g^\star+\delta+d)}-(\varepsilon+d)\right), \\
		a_2 &= \lambda^2(1-\sigma)\left(\frac{\omega_r\varepsilon \left( \gamma(\gamma+g^\star+\delta+d) + \pi(g^\star d - \gamma\delta) \right)}{(\omega_r+d)(\gamma+d)(\gamma+g^\star+\delta+d)}-(\varepsilon+d)\right) < 0,
	\end{align*}
	where $\lambda$ is given by \eqref{eq:lambda}.
	System \eqref{sys:unlimited-vacc} has
	\begin{enumerate}[label=(\roman*), ref=(\roman*)]
		\item a unique positive endemic equilibrium if $\R_v^{\eqref{sys:unlimited-vacc}}>1$, if $\R_v^{\eqref{sys:unlimited-vacc}}=1$ and $a_1>0$, or if $\R_v^{\eqref{sys:unlimited-vacc}}<1$, $a_1>0$, and $a_1^2-4a_0a_2=0$;
		\item two positive endemic equilibria if $\R_v^{\eqref{sys:unlimited-vacc}}<1$, $a_1>0$, and $a_1^2-4a_0a_2>0$;
		\item zero positive equilibria otherwise, and in particular when $\lambda \bar S_0/(\varepsilon+d)<1$.
	\end{enumerate}
	When $\sigma = 1$, we have $a_2 = 0$ and $a_1 < 0$, so \eqref{sys:unlimited-vacc} admits a unique positive endemic equilibrium if $\R_v^{\eqref{sys:unlimited-vacc}}>1$ and none if $\R_v^{\eqref{sys:unlimited-vacc}}\leq 1$.
\end{proposition}

\begin{proposition}
	\label{prop:sufficient-condition-BB}
	Assume $\sigma \in [0, 1)$, and consider the condition
	\begin{multline}
        \label{eq:condition-H1}
		\frac{\sigma v_c\left(\bE_0^{\eqref{sys:unlimited-vacc}}\right)}{1 - \sigma v_c\left(\bE_0^{\eqref{sys:unlimited-vacc}}\right)}\frac{(1-\sigma)d(\varepsilon+d)}{(1-\sigma)v^\star+\omega_v+d}
		> \\
        \varepsilon+d - \frac{\omega_r\varepsilon \left( \gamma(\gamma+g^\star+\delta+d) + \pi(g^\star d - \gamma\delta) \right)}{(\gamma+d)(\gamma+g^\star+\delta+d)(\omega_r+d)},\tag{\textbf{H1}}
	\end{multline}
	where $\bE_0^{\eqref{sys:unlimited-vacc}}$ is the disease-free equilibrium and $v_c(\bE_0^{\eqref{sys:unlimited-vacc}})$ is given by \eqref{eq:v-c-DFE-unlimited-vacc} (note that at $\R_v^{\eqref{sys:unlimited-vacc}}=1$, we have $(\lambda/(\varepsilon+d))\bar S_0 = 1/(1-\sigma v_c(\bE_0^{\eqref{sys:unlimited-vacc}}))$).
	System \eqref{sys:unlimited-vacc} undergoes a backward bifurcation at $\bE_0^{\eqref{sys:unlimited-vacc}}$ and $\R_v^{\eqref{sys:unlimited-vacc}}=1$ whenever condition \eqref{eq:condition-H1} holds.
\end{proposition}

To explore where condition \eqref{eq:condition-H1} is satisfied, we conducted a 500,000-sample Latin hypercube search (with fixed baseline $N = 100{,}000$, $1/d = 70$ years, and $b = dN$) over $v^\star \in [0.001, 0.2]$, $\sigma \in [0.01, 0.99]$, $\omega_v, \omega_r \in [10^{-5}, 0.1]$, $\varepsilon \in [0.05, 0.3]$, $\gamma \in [0.01, 0.2]$, $\pi \in [0.1, 0.9]$, $\delta \in [10^{-5}, 10^{-2}]$, $g^\star \in [0, 0.5]$, $\eta \in [0.1, 1.0]$, and $p^\star \in [0, 1.0]$, with $\beta$ evaluated on the $\R_v^{\eqref{sys:unlimited-vacc}}=1$ manifold (exact evaluation of \eqref{eq:condition-H1} agreed with the normalized sign of $a_1$ across all samples).
Only 757 samples (0.15\%) satisfied \eqref{eq:condition-H1}, and every one of them required disease-induced immunity to wane much faster than vaccine-induced immunity ($\omega_r/\omega_v \geq 5.7$, median $\approx 51$, up to $\approx 8{,}300$, with median $1/\omega_r \approx 16$ days versus median $1/\omega_v \approx 2.5$ years), high disease-free vaccine coverage ($v_c$ median $\approx 97\%$, minimum $\approx 48\%$), and typically high vaccine efficacy (median $\sigma \approx 0.9$).
Under the Table~\ref{tab:parameter_values} baseline ($1/\omega_r = 5$ years versus $1/\omega_v = 200$ days, i.e., $\omega_r/\omega_v \approx 0.1$, and $\sigma = 0.7$), condition \eqref{eq:condition-H1} fails and no backward bifurcation occurs.

%%%%%%%%%%%%%%%%%%%%%%%%%%%
%%%%%%%%%%%%%%%%%%%%%%%%%%%
\subsection{The full model \emph{at equilibrium} is typically unlimited}
\label{supp:analysis-full-model}

%%%%%%%%%%%%%%%%%%%%%%%%%%%
\subsubsection{Dynamics without disease and the disease-free equilibrium}
\label{supp:analysis-full-model-on-invariant-set}
We start by considering \eqref{sys:general-form} in the positively invariant set $\Omega_{L=I=A=0}^{\eqref{sys:general-form}}$ defined by \eqref{eq:Omega-full-system}, where \eqref{sys:general-form} takes the form
\begin{subequations}
   \label{sys:general-form-DFE}
   \begin{align}
        S' &= (1 - p(D_V)) b + \omega_r R + \omega_v V - (v(D_V) + d) S,
        \label{sys:general-form-DFE-dS} \\
        R' &= -(\omega_r + d) R,
        \label{sys:general-form-DFE-dR} \\
        V' &= p(D_V) b + v(D_V) S - (\omega_v + d) V,
        \label{sys:general-form-DFE-dV} \\
        D_V' &= a_V(D_V) - p(D_V) b - v(D_V) S - \kappa_V D_V,
        \label{sys:general-form-DFE-dD} \\
        D_T' &= a_T(D_T) - \kappa_T D_T.
        \label{sys:general-form-DFE-dDT}
   \end{align}
\end{subequations}

An equilibrium of \eqref{sys:general-form-DFE} is a disease-free equilibrium of \eqref{sys:general-form}.
Recall that by assumption, $p$, $v$, and $g$ are nonnegative, nondecreasing $C^1$ functions on $\IR_+$ with $p(0)=v(0)=g(0)=0$, with the Hill function \eqref{eq:p-function} being a prototypical example, while $a_V$ and $a_T$ are nonnegative, nonincreasing $C^1$ functions on $\IR_+$ satisfying $a_V(0)>0$ and $a_T(0)>0$, and $\kappa_V,\kappa_T>0$.
This implies that $\IR_+^5$ is positively invariant under the flow of \eqref{sys:general-form-DFE}.

At an equilibrium of \eqref{sys:general-form-DFE}, equation \eqref{sys:general-form-DFE-dR} immediately yields $R^\star=0$.
Summing \eqref{sys:general-form-DFE-dS} and \eqref{sys:general-form-DFE-dV} with $R^\star=0$ gives $S^\star + V^\star = b/d$.
Substituting $V^\star = b/d - S^\star$ and $R^\star=0$ into \eqref{sys:general-form-DFE-dS} then yields
\[
S^\star = \frac{(1-p(D_V))d+\omega_v}{v(D_V)+d+\omega_v}\frac{b}{d}.
\]
Substituting this expression directly into \eqref{sys:general-form-DFE-dD} at equilibrium, $a_V(D_V) = p(D_V)b + v(D_V)S^\star + \kappa_V D_V$, avoids any division by $v(D_V)$ and gives
\begin{equation}\label{eq:ai-eq-Psii}
    a_V(D_V) = \Psi(D_V) + \kappa_V D_V,
\end{equation}
where
\[
\Psi(D_V) = p(D_V)b + v(D_V)S^\star = \frac{b(d+\omega_v)}{d} \frac{p(D_V)d + v(D_V)}{v(D_V)+d+\omega_v}.
\]
We have $\Psi(0)=0$.
Furthermore, the derivative of $\Psi$ takes the form
\[
\Psi'(D_V) = \frac{b(d+\omega_v)}{d(v(D_V)+d+\omega_v)^2} \left[ p'(D_V)d(v(D_V)+d+\omega_v) + v'(D_V)(d(1 - p(D_V))+\omega_v) \right],
\]
whence $\Psi$ is nondecreasing on $\IR_+$ because $p', v' \geq 0$ and $p \leq 1$ (and strictly increasing whenever $p'(D_V)>0$ or $v'(D_V)>0$ for $D_V>0$, even if $p'(0)=v'(0)=0$ as holds for the Hill function \eqref{eq:p-function} with $k_p,k_v>1$).
Since $a_V(0)>0$ and $a_V$ is continuous and nonincreasing on $\IR_+$, while $\kappa_V>0$ ensures that $D_V\mapsto \Psi(D_V)+\kappa_V D_V$ is continuous, vanishes at $D_V=0$, and satisfies $(\Psi(D_V)+\kappa_V D_V)'\geq \kappa_V>0$ (hence increases strictly to $+\infty$ as $D_V\to\infty$), it follows that \eqref{eq:ai-eq-Psii} admits a unique solution $D_V^\star>0$.
Similarly, \eqref{sys:general-form-DFE-dDT} yields $a_T(D_T)=\kappa_T D_T$.
Since $a_T(0)>0$, $a_T$ is continuous and nonincreasing on $\IR_+$, and $\kappa_T>0$, this equation admits a unique solution $D_T^\star>0$.
Thus, \eqref{sys:general-form-DFE} has a unique equilibrium $(S^\star, 0, V^\star, D_V^\star, D_T^\star)$ in $\IR_+^5$, so that in the absence of disease the general model \eqref{sys:general-form} operates \emph{at equilibrium} under the unlimited-dose regime of Section~\ref{subsec:model-unlimited-vaccine}.

We now show that this equilibrium is both locally and globally asymptotically stable in $\Omega_{L=I=A=0}^{\eqref{sys:general-form}}$.
Writing $N_0 = S + V$ and ordering the coordinates of \eqref{sys:general-form-DFE} as $(R, N_0, D_T, S, D_V)$ yields $N_0' = b + \omega_r R - d N_0$, with $V = N_0 - S$.
In these coordinates, the Jacobian matrix of \eqref{sys:general-form-DFE} at $(0, b/d, D_T^\star, S^\star, D_V^\star)$ is block lower-triangular.
Its upper-left $3\times 3$ block is lower-triangular with diagonal entries $-(\omega_r+d)<0$, $-d<0$, and $a_T'(D_T^\star)-\kappa_T \leq -\kappa_T < 0$, because $a_T$ is nonincreasing.
Its lower-right $2\times 2$ block for $(S, D_V)$ is
\[
J_{S, D_V} =
\begin{pmatrix}
    -(v(D_V^\star)+\omega_v+d) & -p'(D_V^\star)b - v'(D_V^\star)S^\star \\
    -v(D_V^\star) & a_V'(D_V^\star) - p'(D_V^\star)b - v'(D_V^\star)S^\star - \kappa_V
\end{pmatrix}.
\]
Because $v(D_V^\star)\geq 0$, $p'(D_V^\star)\geq 0$, $v'(D_V^\star)\geq 0$, $a_V'(D_V^\star)\leq 0$, and $S^\star\geq 0$, we have
\[
\operatorname{tr}(J_{S, D_V}) = -(v(D_V^\star)+\omega_v+d) + a_V'(D_V^\star) - p'(D_V^\star)b - v'(D_V^\star)S^\star - \kappa_V \leq -(\omega_v + d + \kappa_V) < 0
\]
and, after canceling $\pm v(D_V^\star)(p'(D_V^\star)b + v'(D_V^\star)S^\star)$,
\[
\det(J_{S, D_V}) = (v(D_V^\star)+\omega_v+d)(\kappa_V - a_V'(D_V^\star)) + (\omega_v+d)(p'(D_V^\star)b + v'(D_V^\star)S^\star) \geq (\omega_v + d)\kappa_V > 0.
\]
Hence all five eigenvalues of the Jacobian of \eqref{sys:general-form-DFE} at $(S^\star, 0, V^\star, D_V^\star, D_T^\star)$ have strictly negative real parts, which establishes local asymptotic stability in $\Omega_{L=I=A=0}^{\eqref{sys:general-form}}$.

For global asymptotic stability in $\Omega_{L=I=A=0}^{\eqref{sys:general-form}}$, observe first that along any solution of \eqref{sys:general-form-DFE} in $\IR_+^5$, $R(t) = R(0)e^{-(\omega_r+d)t}\to 0$ and therefore $N_0(t)\to b/d$ as $t\to\infty$.
In the scalar equation \eqref{sys:general-form-DFE-dDT}, $a_T(D_T)-\kappa_T D_T$ is strictly positive on $[0, D_T^\star)$ and strictly negative on $(D_T^\star, \infty)$, so $D_T(t)\to D_T^\star$ monotonically as $t\to\infty$.
In addition, \eqref{sys:general-form-DFE-dD} satisfies $D_V' \leq a_V(0) - \kappa_V D_V$, whence $\limsup_{t\to\infty} D_V(t) \leq a_V(0)/\kappa_V$.
Every forward trajectory of \eqref{sys:general-form-DFE} is therefore bounded, so its $\omega$-limit set is nonempty, compact, and invariant, and every point in this $\omega$-limit set satisfies $R=0$, $S+V=b/d$, $D_T=D_T^\star$, and $(S, D_V)\in\mathcal{B} := [0, b/d]\times [0, a_V(0)/\kappa_V]$.
On the invariant manifold $\{R=0,\, S+V=b/d,\, D_T=D_T^\star\}$, the flow of \eqref{sys:general-form-DFE} is governed by the planar $C^1$ system
\begin{align*}
    S' &= (1 - p(D_V))b + \omega_v(b/d - S) - (v(D_V) + d)S =: f_1(S, D_V), \\
    D_V' &= a_V(D_V) - p(D_V)b - v(D_V)S - \kappa_V D_V =: f_2(S, D_V),
\end{align*}
whose unique equilibrium in $\mathcal{B}$ is $\bm{e}^\star := (S^\star, D_V^\star)$.
Let $\Lambda\subset\mathcal{B}$ be the projection of the $\omega$-limit set onto the $(S, D_V)$-plane; then $\Lambda$ is nonempty, compact, and invariant under this planar flow.
Because $p(0)=v(0)=0$, we have $f_2(S, 0)=a_V(0)>0$ for all $S\in\IR$, so no point with $D_V=0$ can belong to the invariant set $\Lambda\subset\mathcal{B}$ (moving backward in time would immediately enter $\{D_V<0\}$).
Choosing $\xi>0$ small enough that $\xi \max_{D_V\in[0, a_V(0)/\kappa_V+1]} v'(D_V) < (\omega_v+d+\kappa_V)/2$, it follows that $\Lambda$ lies inside the open, simply connected rectangle $\mathcal{U} := (-\xi, b/d+1)\times(0, a_V(0)/\kappa_V+1)$, on which $(f_1, f_2)$ is $C^1$ and its divergence satisfies
\[
\frac{\partial f_1}{\partial S} + \frac{\partial f_2}{\partial D_V} = -(v(D_V)+\omega_v+d) + a_V'(D_V) - p'(D_V)b - v'(D_V)S - \kappa_V < -\frac{\omega_v + d + \kappa_V}{2} < 0.
\]
By Bendixson's criterion, no periodic orbit exists in $\mathcal{U}$.
Take any $\by\in\Lambda$.
Because $\Lambda\subset\mathcal{B}\cap\mathcal{U}$ is compact and invariant, the $\alpha$-limit set $\alpha(\by)$ under the planar flow is nonempty, compact, and contained in $\Lambda\subset\mathcal{B}\cap\mathcal{U}$.
By the Poincar\'e--Bendixson theorem (applied to the time-reversed flow on $\mathcal{U}$) and the absence of periodic orbits in $\mathcal{U}$, $\alpha(\by)$ must contain an equilibrium in $\mathcal{B}$, so $\bm{e}^\star\in\alpha(\by)$.
However, $\bm{e}^\star$ is locally asymptotically stable for the planar system (its Jacobian $J_{S, D_V}$ has strictly negative trace and strictly positive determinant), so no point $\by\neq \bm{e}^\star$ can satisfy $\bm{e}^\star\in\alpha(\by)$ (otherwise points along the backward orbit arbitrarily close to $\bm{e}^\star$ would return to $\by$ in forward time, contradicting Lyapunov stability of $\bm{e}^\star$).
Thus $\by = \bm{e}^\star$ for every $\by\in\Lambda$, proving $\Lambda=\{\bm{e}^\star\}$.
Consequently, the $\omega$-limit set of every trajectory in $\Omega_{L=I=A=0}^{\eqref{sys:general-form}}$ is $\{\bE_0^{\eqref{sys:general-form}}\}$, which together with local asymptotic stability proves that $\bE_0^{\eqref{sys:general-form}}$ is globally asymptotically stable in $\Omega_{L=I=A=0}^{\eqref{sys:general-form}}$.

%%%%%%%%%%%%%%%%%%%%%%%%%%%
\subsubsection{Stability of the disease-free equilibrium}
\label{supp:analysis-full-model-at-DFE}

\begin{proof}[Proof of Theorem~\ref{th:full-model-DFE-and-persistence}]
We first verify local asymptotic stability when $\R_v^{\eqref{sys:general-form}} < 1$ and instability when $\R_v^{\eqref{sys:general-form}} > 1$ via the next-generation matrix framework of \cite{van2002reproduction}.
Ordering the compartments as $(L, I, A, S, R, V, D_V, D_T)$, let $\F$ and $\W$ be the vectors of new infections and other transition rates, respectively, in the infected compartments $(L, I, A)$, given by
\[
\F =
\begin{pmatrix}
    \beta (S+(1-\sigma) V)(I+\eta A) \\ 0 \\ 0
\end{pmatrix}
\quad\text{and}\quad
\W =
\begin{pmatrix}
    (\varepsilon + d) L \\
    -\pi\varepsilon L +(\gamma+g(D_T)+\delta +d) I \\
    -(1-\pi)\varepsilon L + (\gamma+d) A
\end{pmatrix}.
\]
At the disease-free equilibrium $\bE_0^{\eqref{sys:general-form}}=(S^\star, 0, 0, 0, 0, V^\star, D_V^\star, D_T^\star)$ found in Section~\ref{supp:analysis-full-model-on-invariant-set}, $L=I=A=0$, so all partial derivatives of the infected equations with respect to the uninfected variables vanish at $\bE_0^{\eqref{sys:general-form}}$ (in particular, $\partial(g(D_T)I)/\partial D_T = g'(D_T)I = 0$ when $I=0$).
Consequently, denoting $g^\star = g(D_T^\star)$, the new-infections matrix $F$ and the transitions matrix $W$ evaluated at $\bE_0^{\eqref{sys:general-form}}$ are
\[
F =
\begin{pmatrix}
    0 & \beta (S^\star+(1-\sigma) V^\star) & \beta \eta (S^\star+(1-\sigma) V^\star) \\ 0 & 0 & 0 \\ 0 & 0 & 0
\end{pmatrix}
\]
and
\[
W =
\begin{pmatrix}
    \varepsilon + d & 0 & 0 \\
    -\pi\varepsilon & \gamma+g^\star+\delta +d & 0 \\
    -(1-\pi)\varepsilon & 0 & \gamma+d
\end{pmatrix},
\]
where $S^\star$ and $V^\star$ coincide with $S_0$ and $V_0$ in \eqref{eq:S0_V0_fct_vcE0} evaluated at the steady-state vaccine rates $p^\star = p(D_V^\star)$ and $v^\star = v(D_V^\star)$ of Section~\ref{subsec:model-unlimited-vaccine}.
Therefore, the intervention reproduction number $\R_v^{\eqref{sys:general-form}} = \rho(FW^{-1})$ is equal to $\R_v^{\eqref{sys:unlimited-vacc}}$ in \eqref{eq:Rv_SVLIARS} evaluated at the steady-state dose values $D_V^\star$ and $D_T^\star$.

Hypotheses (A1)--(A4) of \cite{van2002reproduction} hold directly for \eqref{sys:general-form} because $p$, $v$, and $g$ are nonnegative and vanish at zero.
To verify hypothesis (A5), set $\F \equiv \b0$ in \eqref{sys:general-form}.
At $\bE_0^{\eqref{sys:general-form}}$, the Jacobian matrix of the resulting system is block lower-triangular with diagonal blocks $-W$ and the $5\times 5$ Jacobian of the disease-free subsystem \eqref{sys:general-form-DFE}.
The eigenvalues of $-W$ are $-(\varepsilon+d)<0$, $-(\gamma+g^\star+\delta+d)<0$, and $-(\gamma+d)<0$, while all five eigenvalues of the Jacobian of \eqref{sys:general-form-DFE} have strictly negative real parts by Section~\ref{supp:analysis-full-model-on-invariant-set}.
Thus hypothesis (A5) of \cite{van2002reproduction} holds.
By Theorem~2 of \cite{van2002reproduction}, $\bE_0^{\eqref{sys:general-form}}$ is locally asymptotically stable if $\R_v^{\eqref{sys:general-form}} < 1$ and unstable if $\R_v^{\eqref{sys:general-form}} > 1$.

Now assume that $\R_v^{\eqref{sys:general-form}} > 1$ and $\eta > 0$, and let $X = \IR_+^8$, $X_0 = \{\bX\in\IR_+^8 \mid L+I+A>0\}$, and $\partial X_0 = X\setminus X_0 = \Omega_{L=I=A=0}^{\eqref{sys:general-form}}$.
Along any solution of \eqref{sys:general-form} in $\IR_+^8$, the total population $N=S+L+I+A+R+V$ satisfies $N' = b - dN - \delta I \leq b - dN$, so $\limsup_{t\to\infty} N(t) \leq b/d$.
Likewise, since $a_V$ and $a_T$ are nonincreasing and $\kappa_V,\kappa_T>0$, the dose compartments satisfy $D_V'\leq a_V(0)-\kappa_V D_V$ and $D_T'\leq a_T(0)-\kappa_T D_T$, whence $\limsup_{t\to\infty} D_V(t) \leq a_V(0)/\kappa_V$ and $\limsup_{t\to\infty} D_T(t) \leq a_T(0)/\kappa_T$.
Hence every solution starting in $X$ exists for all $t\geq 0$ and remains bounded, defining a continuous, point-dissipative semiflow $\Phi_t$ on $X$ that admits a compact global attractor.
Both $\partial X_0$ and $X_0$ are positively invariant under $\Phi_t$: invariance of $\partial X_0$ was noted in \eqref{eq:Omega-full-system}, and summing \eqref{sys:general-form-dL}--\eqref{sys:general-form-dA} gives
\begin{multline*}
    (L+I+A)' = \beta(S+(1-\sigma)V)(I+\eta A) - dL - (\gamma+g(D_T)+\delta+d)I - (\gamma+d)A \\
    \geq -(\gamma+g_{\max}+\delta+d)(L+I+A),
\end{multline*}
so $(L+I+A)(t) \geq (L+I+A)(0)e^{-(\gamma+g_{\max}+\delta+d)t} > 0$ for all $t\geq 0$ whenever $\bX(0)\in X_0$.

By the global asymptotic stability of $\bE_0^{\eqref{sys:general-form}}$ in $\partial X_0 = \Omega_{L=I=A=0}^{\eqref{sys:general-form}}$ proved in Section~\ref{supp:analysis-full-model-on-invariant-set}, we have $\omega(\bX) = \{\bE_0^{\eqref{sys:general-form}}\}$ for every $\bX\in\partial X_0$, so $\bigcup_{\bX\in\partial X_0}\omega(\bX) = \{\bE_0^{\eqref{sys:general-form}}\}$ is a single, acyclic compact invariant set in $\partial X_0$.
We next show that $\bE_0^{\eqref{sys:general-form}}$ is a weak uniform repeller for $X_0$, i.e., there exists $\nu_0 > 0$ such that
\[
\limsup_{t\to\infty} \|\Phi_t(\bX(0)) - \bE_0^{\eqref{sys:general-form}}\|_1 \geq \nu_0
\]
for every $\bX(0)\in X_0$.
Set $B^\star = \beta(S^\star + (1-\sigma)V^\star)$; because $\R_v^{\eqref{sys:general-form}} > 1$, we have $B^\star > 0$.
Consider the continuous, strictly decreasing function
\[
h(s) = \frac{\pi\varepsilon B^\star}{\gamma + g^\star + \delta + d + s} + \frac{(1-\pi)\varepsilon\eta B^\star}{\gamma + d + s} - (\varepsilon + d + s)
\]
for $s\geq 0$.
Since $h(0) = (\varepsilon+d)(\R_v^{\eqref{sys:general-form}} - 1) > 0$ and $\lim_{s\to\infty} h(s) = -\infty$, there is a unique root $s > 0$ of $h(s)=0$.
Define the vector $\bm{u} = (u_1, u_2, u_3)^T$ by $u_1 = 1$, $u_2 = B^\star/(\gamma+g^\star+\delta+d+s) > 0$, and $u_3 = \eta B^\star/(\gamma+d+s) > 0$ (strictly positive since $\eta>0$).
Direct computation verifies that $\bm{u}^T (F - W) = s\,\bm{u}^T$.
Along any solution of \eqref{sys:general-form}, define $U(t) = u_1 L(t) + u_2 I(t) + u_3 A(t)$, which is strictly positive whenever $\bX(t)\in X_0$.
Differentiating $U(t)$ along \eqref{sys:general-form-dL}--\eqref{sys:general-form-dA} yields
\[
U'(t) = s U(t) + \beta\bigl((S + (1-\sigma)V) - (S^\star + (1-\sigma)V^\star)\bigr)(I + \eta A) - u_2\bigl(g(D_T) - g^\star\bigr)I.
\]
Because $g^\star \geq 0$ and $\delta \geq 0$, we have $\eta/u_3 = (\gamma+d+s)/B^\star \leq (\gamma+g^\star+\delta+d+s)/B^\star = 1/u_2$, whence
\[
I + \eta A = \frac{1}{u_2}(u_2 I) + \frac{\eta}{u_3}(u_3 A) \leq \frac{1}{u_2}(u_2 I + u_3 A) \leq \frac{1}{u_2} U(t)
\]
and $u_2 I \leq U(t)$.
Therefore,
\[
U'(t) \geq \left(s - \frac{\beta}{u_2}\bigl|(S + (1-\sigma)V) - (S^\star + (1-\sigma)V^\star)\bigr| - \bigl|g(D_T) - g^\star\bigr|\right) U(t).
\]
By continuity of $g$ at $D_T^\star$, there exists $\nu_0 > 0$ such that whenever $\|\bX - \bE_0^{\eqref{sys:general-form}}\|_1 < \nu_0$,
\[
\frac{\beta}{u_2}\bigl|(S + (1-\sigma)V) - (S^\star + (1-\sigma)V^\star)\bigr| + \bigl|g(D_T) - g^\star\bigr| \leq \frac{s}{2},
\]
and hence $U'(t) \geq (s/2)U(t)$.
Suppose that $\limsup_{t\to\infty} \|\Phi_t(\bX(0)) - \bE_0^{\eqref{sys:general-form}}\|_1 < \nu_0$ for some $\bX(0)\in X_0$.
Then there exists $T_0 \geq 0$ such that $\|\Phi_t(\bX(0)) - \bE_0^{\eqref{sys:general-form}}\|_1 < \nu_0$ for all $t\geq T_0$.
Because $X_0$ is positively invariant, $U(T_0) > 0$, so $U(t) \geq U(T_0)e^{(s/2)(t-T_0)} \to \infty$ as $t\to\infty$, contradicting the boundedness of solutions.
Thus $\limsup_{t\to\infty} \|\Phi_t(\bX(0)) - \bE_0^{\eqref{sys:general-form}}\|_1 \geq \nu_0$ for all $\bX(0)\in X_0$.

In particular, the stable set of $\bE_0^{\eqref{sys:general-form}}$ satisfies $W^s(\bE_0^{\eqref{sys:general-form}}) \cap X_0 = \emptyset$.
This repeller property also implies that $\{\bE_0^{\eqref{sys:general-form}}\}$ is an isolated invariant set in $X$: any compact invariant set $K$ contained in the neighborhood $\{\bX\in X \mid \|\bX - \bE_0^{\eqref{sys:general-form}}\|_1 \leq \nu_0/2\}$ cannot intersect $X_0$, so $K\subset\partial X_0$, and global asymptotic stability of $\bE_0^{\eqref{sys:general-form}}$ in $\partial X_0$ forces $K = \{\bE_0^{\eqref{sys:general-form}}\}$.
Applying uniform persistence theory (Theorem~4.6 of \cite{thieme1993persistence}; see also Theorem~1.3.1 of \cite{zhao2003dynamical}), there exists a constant $c > 0$ such that
\[
\liminf_{t\to\infty} \operatorname{dist}_1(\Phi_t(\bX(0)), \partial X_0) \geq c
\]
for every $\bX(0)\in X_0$, where $\operatorname{dist}_1$ denotes the $\ell^1$ distance in $\IR_+^8$.
Since $\operatorname{dist}_1(\bX, \partial X_0) = L + I + A$ for all $\bX\in\IR_+^8$, it follows that $\liminf_{t\to\infty} (L(t) + I(t) + A(t)) \geq c$ for all initial conditions with $L(0)+I(0)+A(0)>0$.

Finally, if $\eta = 0$ and $\R_v^{\eqref{sys:general-form}} > 1$, then asymptomatic individuals do not transmit infection (so trajectories starting with $L(0)=I(0)=0$ and $A(0)>0$ satisfy $A(t)\to 0$), and repeating the same argument with $u_3=0$ on $\tilde{X}_0 = \{\bX\in\IR_+^8 \mid L+I>0\}$ (where $A(t)\to 0$ on $\partial\tilde{X}_0=\{L=I=0\}$) gives uniform persistence of $L(t)+I(t)$ for all initial conditions with $L(0)+I(0)>0$.
\end{proof}

%%%%%%%%%%%%%%%%%%%%%
%%%%%%%%%%%%%%%%%%%%%
%%%%%%%%%%%%%%%%%%%%%
%%%%%%%%%%%%%%%%%%%%%
\section{Computational considerations}
\label{supp:computational}

\begin{figure}[htbp]
    \centering
    \begin{tikzpicture}
        \pgfmathdeclarefunction{hillFunction}{1}{%
          \pgfmathparse{5 * ((#1/10)*(#1/10)*(#1/10)*(#1/10)) / (1 + (#1/10)*(#1/10)*(#1/10)*(#1/10))}%
        }
        % --- Set up the axes and plot styles ---
        \begin{axis}[
            xlabel={$D_V$},
            ylabel={$p(D_V)$},
            xmin=0, xmax=25,
            ymin=0, ymax=6,
            axis lines=left,
            enlargelimits=false,
            clip=false, % Allow drawing outside the axis box
            xtick={10},
            xticklabels={$D_p^s$},
            ytick={2.5, 5},
            yticklabels={$p_{\max}/2$, $p_{\max}$},
        ]
        % --- Plot the function ---
        \addplot[
            domain=0:25,
            samples=101,
            color=green!50!black,
            thick,
        ] {hillFunction(x)};
        % Dashed lines for the half-saturation point
        \addplot[
            dashed,
            color=gray,
            thick,
        ] coordinates {(10, 0) (10, 2.5) (0, 2.5)};
        % Dashed line for the asymptote
        \addplot[
            dashed,
            color=gray,
            thick,
        ] coordinates {(0, 5) (25, 5)};
        % Mark the starting point (0, 0)
        \node[circle, fill=red, inner sep=2pt] at (axis cs:0,0) {};
        % Mark the half-saturation point on the curve
        \node[circle, fill=blue, inner sep=1.5pt] at (axis cs:10, 2.5) {};
        \end{axis}
    \end{tikzpicture}
    \caption{The Hill function $p(D_V)$ defined by \eqref{eq:p-function} (illustrated with $k_p=4$).}
    \label{fig:hill-p}
\end{figure}

For the dose-dependent rates $p(D_V)$, $v(D_V)$, and $g(D_T)$, we use Hill functions, which satisfy $p(0)=v(0)=g(0)=0$ without shifting.
The function $p(D_V)$, for instance, takes the form
\begin{equation}\label{eq:p-function}
p(D_V) = p_{\max} \frac{D_V^{k_p}}{(D_p^s)^{k_p} + D_V^{k_p}},
\end{equation}
where $p_{\max}=\lim_{D_V\to\infty}p(D_V)\in[0,1]$, $D_p^s>0$ is the half-saturation constant satisfying $p(D_p^s)=p_{\max}/2$, i.e., halfway through the switch, and $k_p\geq 1$ is the Hill coefficient controlling the steepness of the switch.
See Figure~\ref{fig:hill-p}.
Requiring $k_p\geq 1$ ensures that $p\in C^1(\IR_+)$; whenever $k_p>1$, we also have $p'(0)=0$ and an S-shaped profile.
To model a switch that is almost instantaneous and occurs when the number of available doses is very close to $0$, we can set the half-saturation point $D_p^s \leq 1/2$ and choose a large Hill coefficient (e.g., $k_p = 20$).
The same functional form is used for $v(D_V)$ and $g(D_T)$, with corresponding parameters $v_{\max}$, $D_v^s>0$, and $k_v\geq 1$, and $g_{\max}$, $D_g^s>0$, and $k_g\geq 1$.
There is no mathematical upper bound on $v_{\max}$ or $g_{\max}$ because they represent \emph{per capita} rates.
While the absolute slope at the half-saturation point, $p'(D_p^s)=p_{\max}k_p/(4D_p^s)$, is proportional to the maximum value, the transition interval from 10\% to 90\% of the maximum is $[9^{-1/k_p}D_p^s, 9^{1/k_p}D_p^s]$, whose width $(9^{1/k_p}-9^{-1/k_p})D_p^s$ depends only on $D_p^s$ and $k_p$ and is independent of the asymptotic maximum.
Thus, setting $D_v^s, D_g^s \leq 1/2$ and $k_v, k_g = 20$ yields an equally sharp relative switch for $v$ and $g$ regardless of their maximal values.

Baseline no-intervention reference trajectories and cumulative incidence integrals are computed as described in Section~\ref{sec:computational}.
All ODE systems are integrated in Julia using the fifth-order Runge--Kutta solver \texttt{Tsit5} from \texttt{DifferentialEquations.jl} with tolerances $\mathrm{abstol}=10^{-10}$ and $\mathrm{reltol}=10^{-8}$, enforcing nonnegativity through the \texttt{PositiveDomain()} callback and an out-of-domain rejection guard (\texttt{isoutofdomain}).
For the trajectory in Figure~\ref{fig:ODE-optimal-repartition-timeseries}, tighter tolerances ($\mathrm{abstol}=10^{-26}$, $\mathrm{reltol}=10^{-10}$) are used to resolve deep sub-unit prevalence troughs without sign-flip artifacts.
Cumulative incidence integrals are evaluated from the saved trajectories by trapezoidal quadrature (\texttt{Trapz.jl}).
Continuous-time Markov chain sweeps use an exact Gillespie direct-method implementation with 10,000 realizations per grid point.

%%%%%%%%%%%%%%%%%%%%%
%%%%%%%%%%%%%%%%%%%%%
%%%%%%%%%%%%%%%%%%%%%
%%%%%%%%%%%%%%%%%%%%%
\section{Existence and uniqueness of solutions with impulsive interventions}
\label{supp:impulsive-existence}

To guarantee existence and uniqueness of solutions for the impulsive system \eqref{sys:general-form} and \eqref{sys:impulses} studied numerically in Section~\ref{subsec:numerics-impulsive-doses}, we rely on standard continuation results for impulsive differential equations \cite{bainov1989systems,lakshmikantham1989theory}.

Let $\bX(t) = (S(t), L(t), I(t), A(t), R(t), V(t), D_V(t), D_T(t))$ denote the state vector.
The dynamics between impulses are governed by the continuous flow $\bX'(t) = f(\bX(t))$ defined by \eqref{sys:general-form}.
At each impulse time $t_k$, the state undergoes an instantaneous jump $\bX(t_k^+) = \bX(t_k^-) + \Delta\bX_k$, where $\Delta\bX_k = (0, 0, 0, 0, 0, 0, \Delta_V, \Delta_T)$ with $\Delta_V, \Delta_T \geq 0$.

Because $f(\bX)$ is continuously differentiable on $\IR_+^8$ (in particular, the Hill functions of the form \eqref{eq:p-function} belong to $C^1(\IR_+)$ since $k_p, k_v, k_g \geq 1$), it is locally Lipschitz continuous on $\IR_+^8$, ensuring local existence and uniqueness of solutions.
In addition, $\IR_+^8$ is positively invariant under \eqref{sys:general-form}, and the linear upper bounds $N' \leq b - dN$, $D_V' \leq a_V(0) - \kappa_V D_V \leq a_V(0)$, and $D_T' \leq a_T(0) - \kappa_T D_T \leq a_T(0)$ prevent finite-time blow-up of any component.
Consequently, starting from any initial condition $\bX_0 \in \IR_+^8$, the continuous system admits a unique nonnegative solution on the entire interval $[0, t_1]$, where $t_1 > 0$ is the time of the first impulse.
At $t = t_1$, the state jumps to a uniquely defined new state $\bX(t_1^+) \in \IR_+^8$ because $\Delta_V, \Delta_T \geq 0$ preserve nonnegativity, providing a valid initial condition for the subsequent continuous interval $(t_1, t_2]$.

By induction, this piecewise continuation argument extends uniquely across all subsequent intervals $(t_k, t_{k+1}]$.
In both logistical regimens of Section~\ref{subsec:numerics-impulsive-doses}, only finitely many impulses occur: the fixed schedule uses a predetermined finite set of $N_{\text{shipments}}$ strictly increasing delivery times, while the state-dependent reactive regimen partitions a finite central stockpile $D_{\text{total}}$ into fixed positive batches $\Delta = \Delta_V + \Delta_T > 0$, permitting at most $\lfloor D_{\text{total}}/\Delta \rfloor$ impulses (after which the system evolves continuously under \eqref{sys:general-form} on $(t_K, \infty)$, where $t_K$ is the final impulse time).
In particular, impulse times cannot accumulate in finite time (no Zeno behavior), so there exists a unique left-continuous, piecewise-continuous solution $\bX(t) \in \IR_+^8$ defined for all $t \geq 0$.

%%%%%%%%%%%%%%%%%%%%%
%%%%%%%%%%%%%%%%%%%%%
%%%%%%%%%%%%%%%%%%%%%
%%%%%%%%%%%%%%%%%%%%%
\section{A few more figures}
\label{supp:more-figures}

Figure~\ref{fig:coverage-vs-pcttreatment-unlimited-doses-extra} shows additional information to Figure~\ref{fig:coverage-vs-pcttreatment-unlimited-doses}, namely, the value of $\R_v^{\eqref{sys:unlimited-vacc}}$ and the conditional mean time to extinction.
Figure~\ref{fig:impulsive-deliveries-reactive-schedule} shows the relative benefit when using impulsive dose delivery under strict dose limitations ($\R_0=2.5$, $L(0)=5$), as well as the number of reactive shipments in a one-year period.
Figure~\ref{fig:limited-budget-optimal-fv} complements Figure~\ref{fig:limited-budget-comparison} by displaying the optimal vaccination budget share $f_V^\star$ across total budgets and treatment-to-vaccine cost ratios for the influenza-like and measles-like regimes.

\begin{figure}[htbp]
    \centering
    \begin{subfigure}{0.49\textwidth}
        \includegraphics[width=\linewidth]{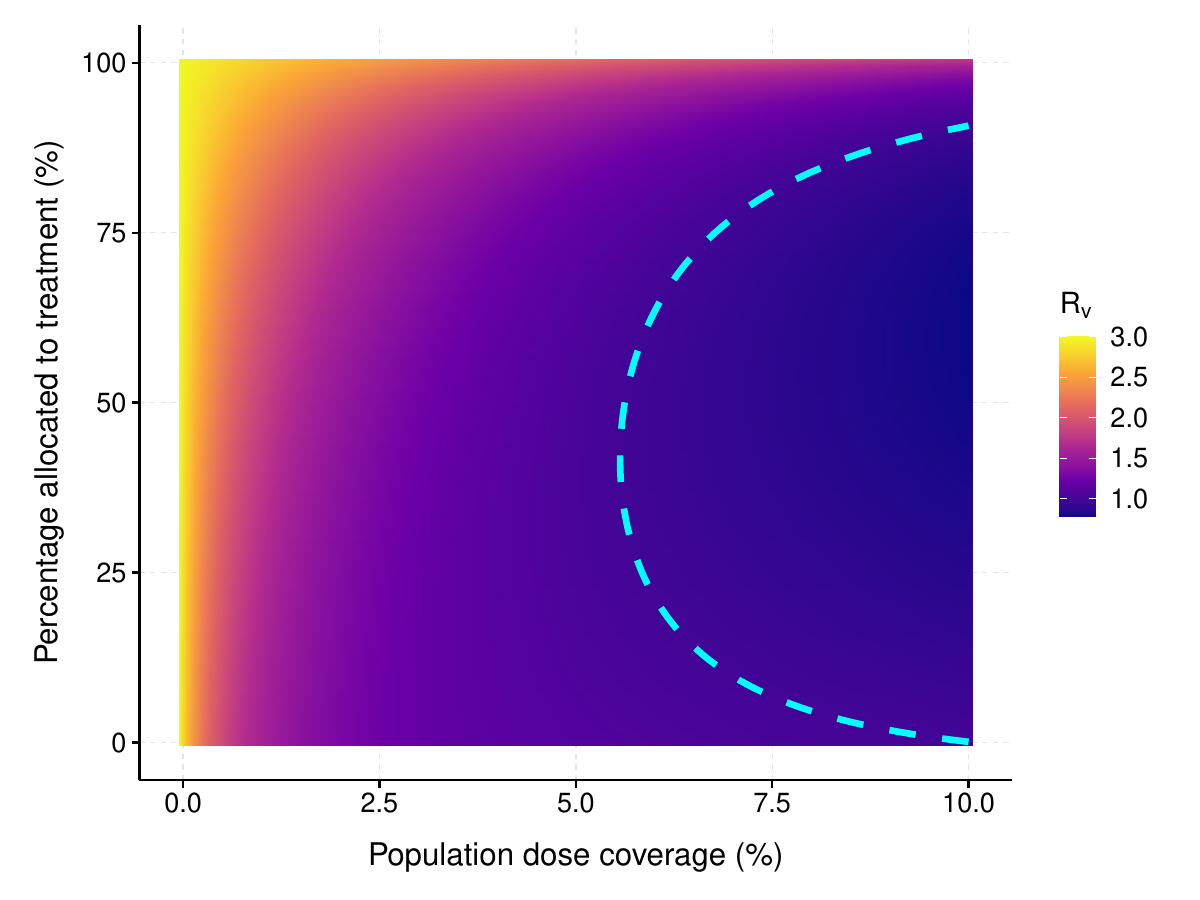}
        \caption{$\R_v$}
        \label{fig:coverage-vs-pcttreatment-unlimited-doses-Rv}
    \end{subfigure}
    \begin{subfigure}{0.49\textwidth}
        \includegraphics[width=\linewidth]{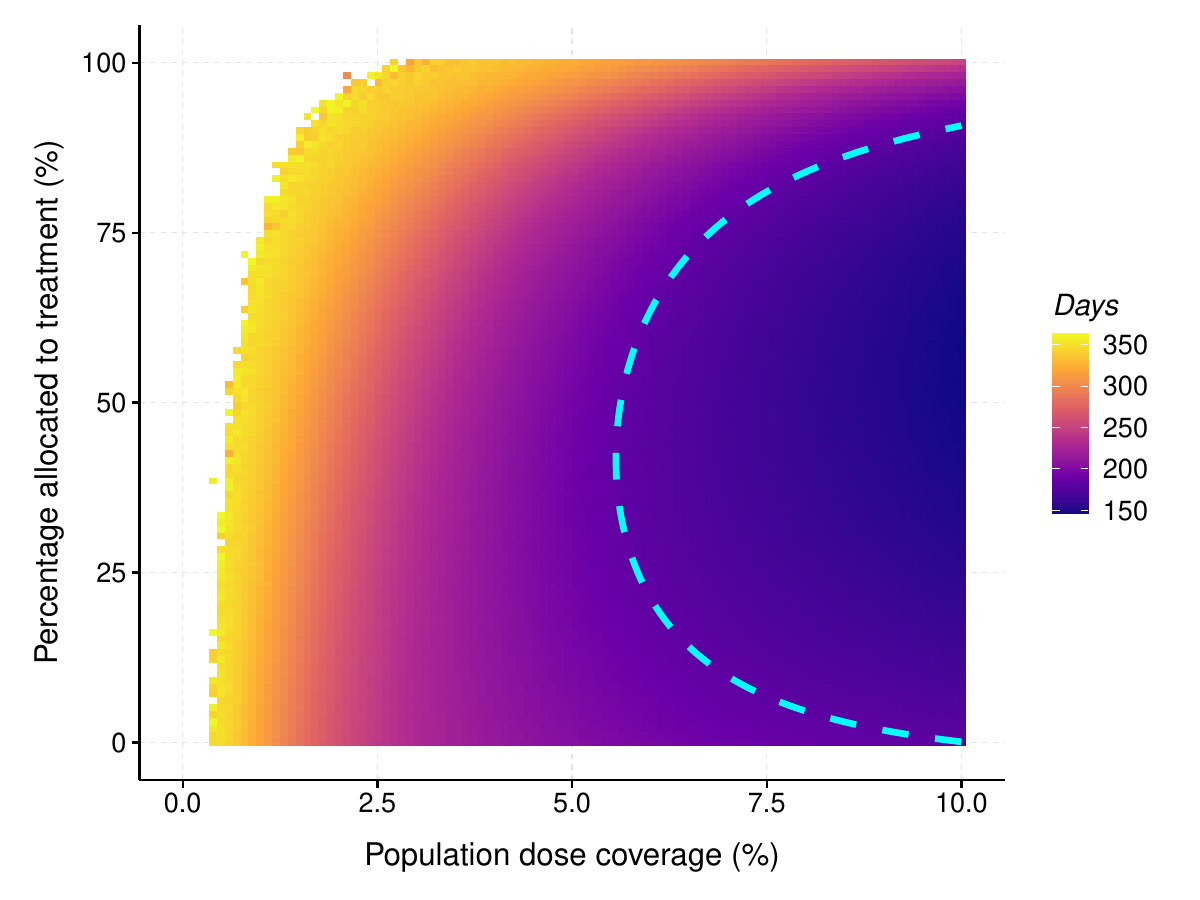}
        \caption{Time to extinction (days)}
        \label{fig:coverage-vs-pcttreatment-unlimited-doses-time}
    \end{subfigure}
    \caption{Role of the percentage of the population that can be covered by doses and the percentage of these doses that are treatment doses in the unlimited-dose system \eqref{sys:unlimited-vacc}, mapped via $p^\star=0$, $v^\star=(D_V/N)\text{ day}^{-1}$, and $g^\star=(D_T/N)\text{ day}^{-1}$.
    Population of 100,000 individuals, with $\R_0=3$.
    The cyan curve shows $\R_v^{\eqref{sys:unlimited-vacc}}=1$.
    (a) Reproduction number $\R_v^{\eqref{sys:unlimited-vacc}}$.
    (b) Average number of days until absorption in the CTMC analog of \eqref{sys:unlimited-vacc} over 10,000 realizations, conditioned on realizations that achieve extinction within 1 year.
    The white region indicates that extinction did not occur within the simulation period of 1 year.}
    \label{fig:coverage-vs-pcttreatment-unlimited-doses-extra}
\end{figure}

\begin{figure}[htbp]
    \centering
    \begin{subfigure}{0.49\textwidth}
        \includegraphics[width=\linewidth]{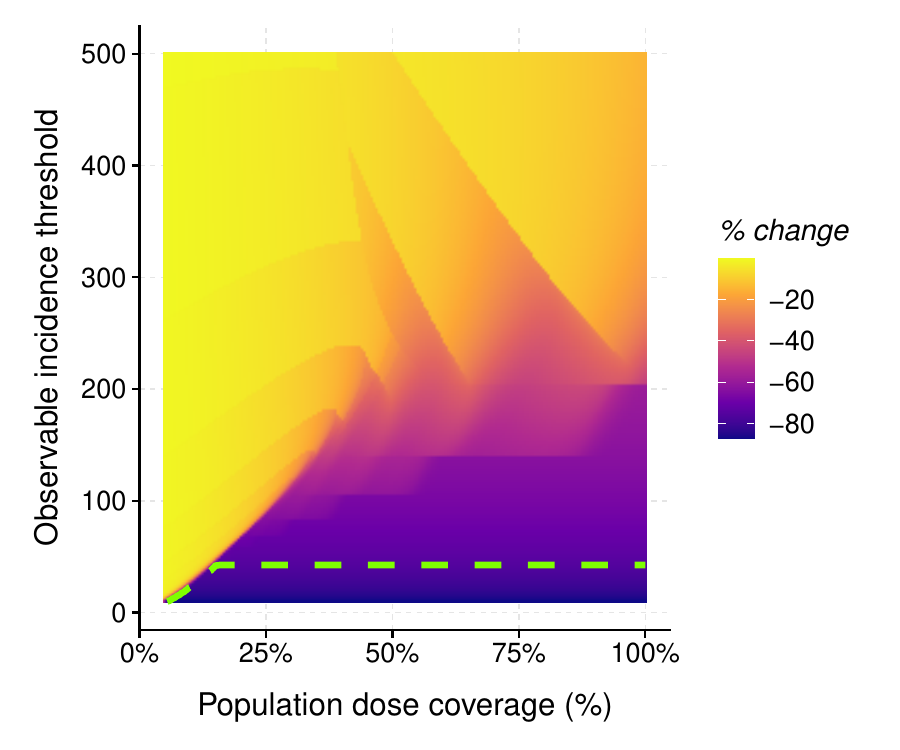}
        \caption{One month (effect)}
        \label{fig:impulsive-deliveries-reactive-schedule-1m}
    \end{subfigure}
    \begin{subfigure}{0.49\textwidth}
        \includegraphics[width=\linewidth]{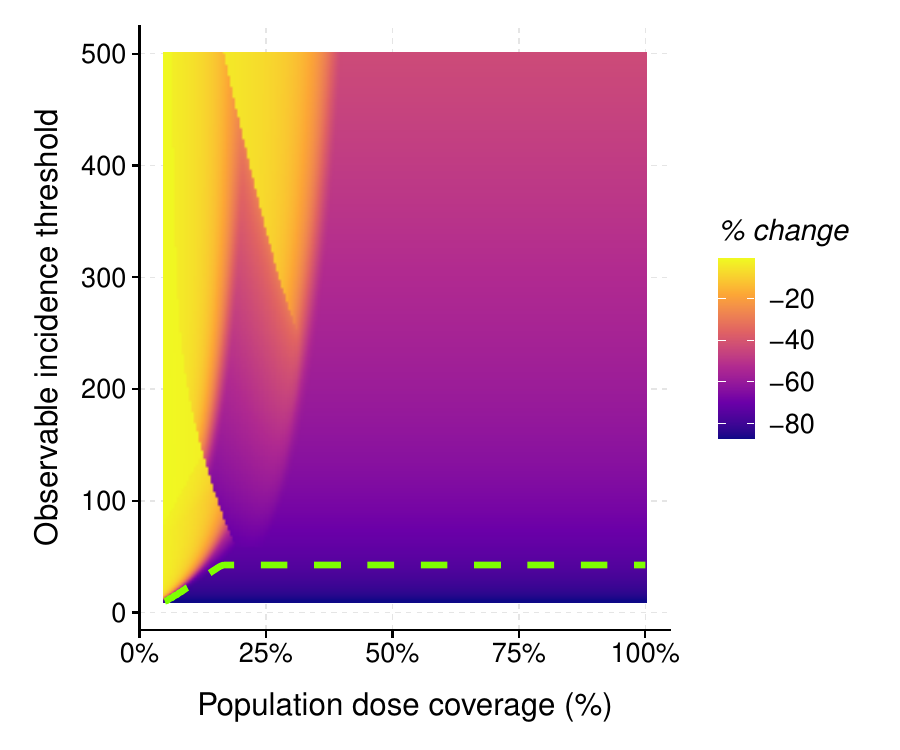}
        \caption{Six months (effect)}
        \label{fig:impulsive-deliveries-reactive-schedule-6m}
    \end{subfigure} \\
    \begin{subfigure}{0.49\textwidth}
        \includegraphics[width=\linewidth]{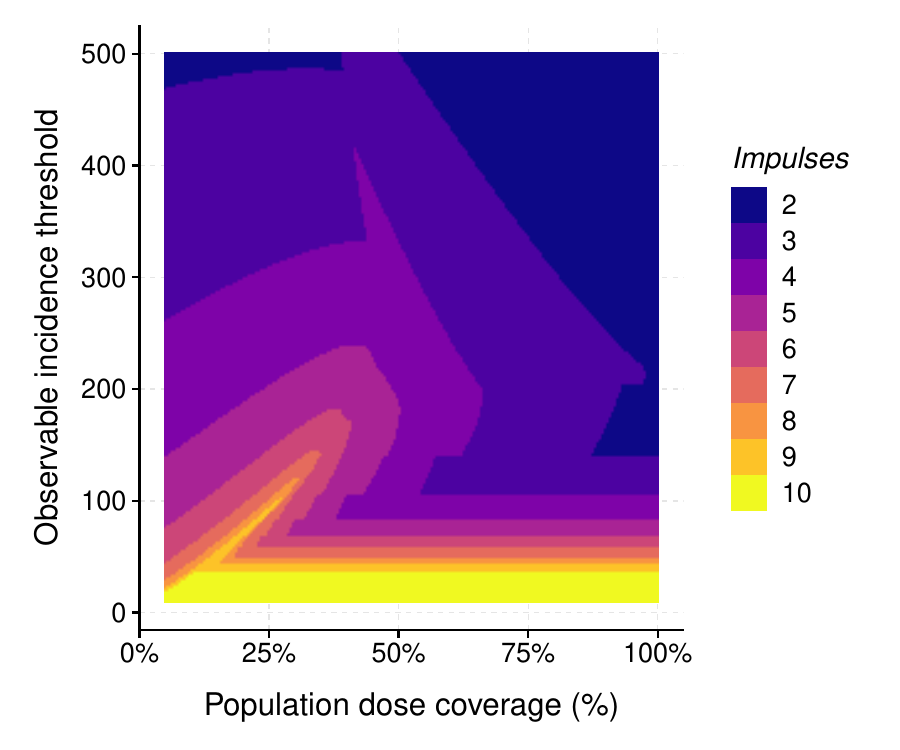}
        \caption{One month (impulses)}
        \label{fig:impulsive-deliveries-reactive-schedule-1m-impulses}
    \end{subfigure}
    \begin{subfigure}{0.49\textwidth}
        \includegraphics[width=\linewidth]{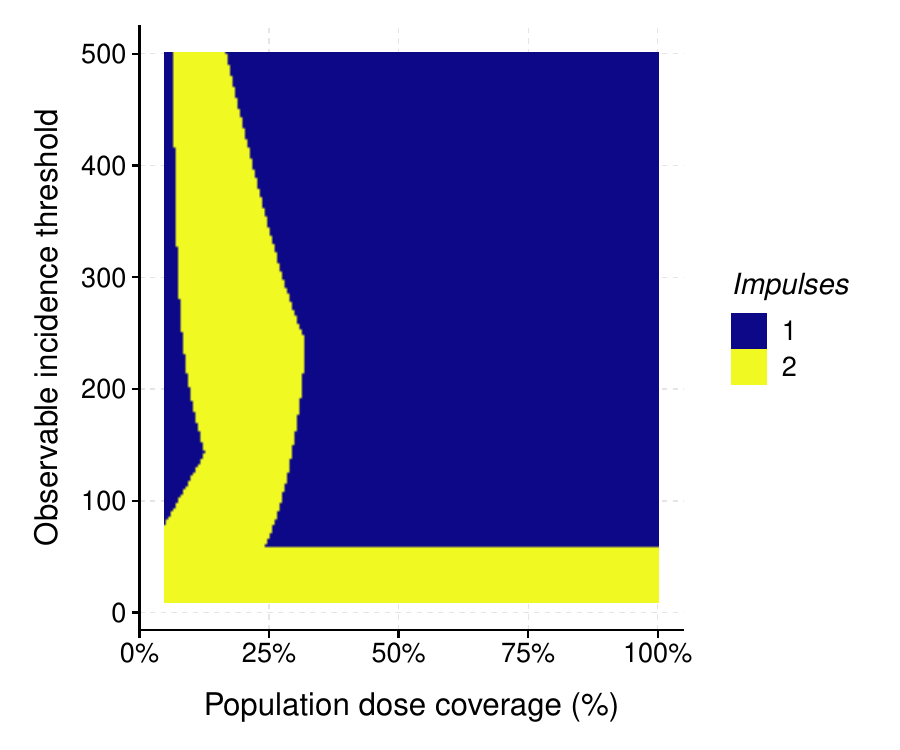}
        \caption{Six months (impulses)}
        \label{fig:impulsive-deliveries-reactive-schedule-6m-impulses}
    \end{subfigure}
    \caption{Reactive impulsive treatment delivery with $\R_0=2.5$ and initial condition $L(0)=5$ ($D_V(0)=D_T(0)=0$).
    (a,b) Relative benefit (percentage change in cumulative detectable cases over 1 year compared to a scenario without intervention) under 1-month and 6-month batch sizes.
    (c,d) Number of reactive shipments triggered in a one-year period corresponding to (a,b).}
    \label{fig:impulsive-deliveries-reactive-schedule}
\end{figure}

\begin{figure}[htbp]
    \centering
    \begin{subfigure}[b]{0.49\textwidth}
        \centering
        \includegraphics[width=\textwidth]{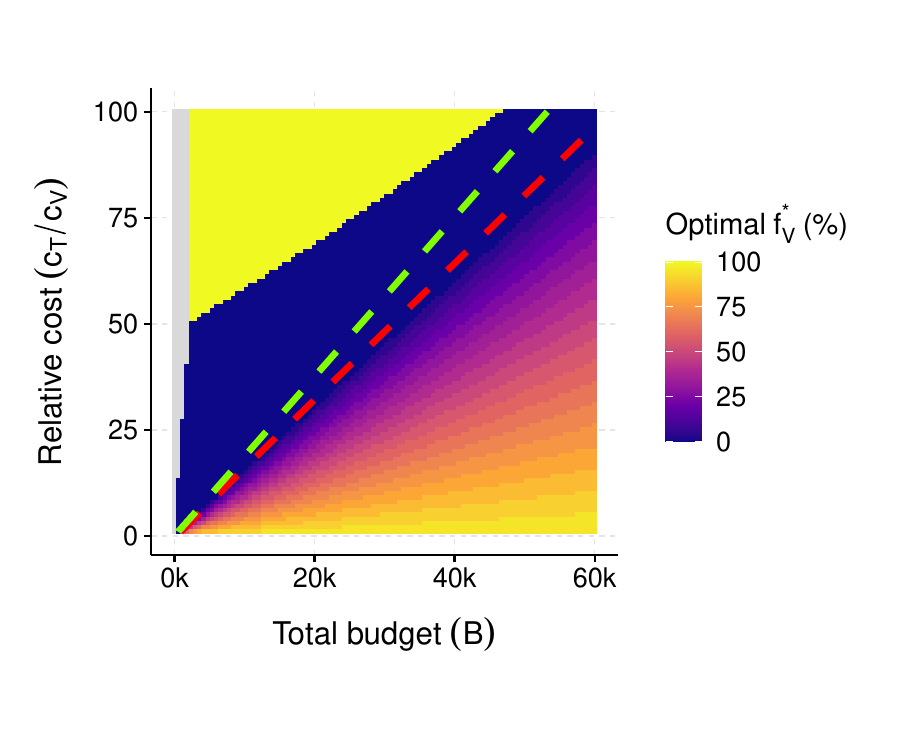}
        \caption{Influenza-like}
        \label{fig:limited-budget-optimal-fv-flu}
    \end{subfigure}
    \hfill
    \begin{subfigure}[b]{0.49\textwidth}
        \centering
        \includegraphics[width=\textwidth]{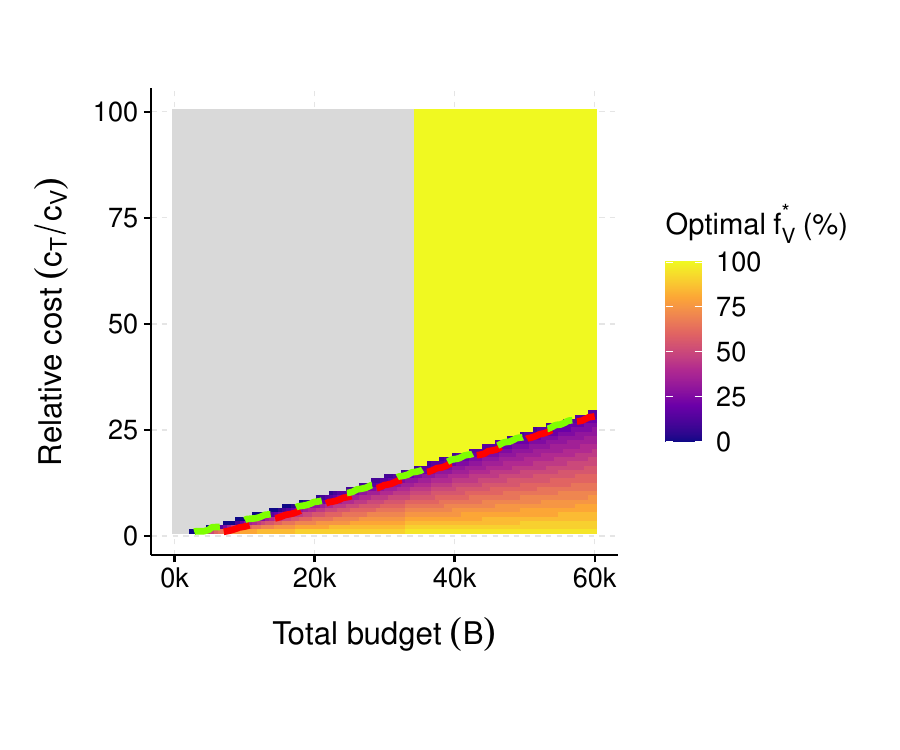}
        \caption{Measles-like}
        \label{fig:limited-budget-optimal-fv-measles}
    \end{subfigure}
    \caption{Optimal vaccination budget share $f_V^\star$ achieving the maximum 1-year reduction in cumulative observed incidence shown in Figure~\ref{fig:limited-budget-comparison} ($N=100{,}000$, $c_V=2$, $D_V(0)=f_V B/c_V$, $D_T(0)=(1-f_V)B/c_T$, unreplenished stock) for (a) an influenza-like case ($\R_0=1.5$) and (b) a measles-like case ($\R_0=15$).
    The red and green curves show the 90\% and 75\% reduction contours from Figure~\ref{fig:limited-budget-comparison}.
    Cells shaded in light gray correspond to budget and cost-ratio combinations where no allocation reduces 1-year cumulative observed incidence by at least 5\%, so $f_V^\star$ is not practically distinguishable.}
    \label{fig:limited-budget-optimal-fv}
\end{figure}

\clearpage

%%%%%%%%%%%%%%%%%%%%%
%%%%%%%%%%%%%%%%%%%%%

\clearpage
\bibliographystyle{plainurl}
\bibliography{bibliography}
\end{document}